\documentclass{article}
\usepackage{fullpage}

\newcommand{\move}[1]{}

\usepackage{graphicx} 
\usepackage{stmaryrd}
\usepackage[toc,page]{appendix} 
\usepackage{latexsym}
\usepackage{amsmath}
\usepackage{amsfonts}
\usepackage{amsthm}
\usepackage{hyperref}
\usepackage{cleveref}
\usepackage{enumitem}
\usepackage{amssymb}
\usepackage[svgnames]{xcolor}
\usepackage[T1]{fontenc}
\usepackage{mathrsfs}
\usepackage{mathtools}
\usepackage{subcaption}
\usepackage{xcolor}
\usepackage{amsmath}
\usepackage[color]{changebar}
\cbcolor{blue}

\usepackage{tikz}
\usepackage{pgfmath}
\usetikzlibrary{scopes}
\usetikzlibrary{backgrounds}
\usetikzlibrary{shapes}
\usetikzlibrary{calc}
\usetikzlibrary{automata}
\usetikzlibrary{decorations.pathmorphing}
\usetikzlibrary{shapes.symbols}

\tikzstyle{stvertex} =
  [ fill=black,
    inner sep=0pt,
    minimum size=3pt,
    circle
  ]
\tikzstyle{stterminal} =
  [ fill=black,
    inner sep=0pt,
    minimum size=3pt,
    rectangle
  ]

\newcommand{\drop}[1]{}

\newcommand{\set}[1]{\left\{ {#1} \right\}}

\DeclareMathOperator{\poly}{poly}

\newcommand*{\RR}{\mathbb{R}}

\newcommand*{\resolved}[1]{}

\DeclareMathOperator{\CW}{CW} 
\DeclareMathOperator{\ACW}{ACW} 

\newcommand*{\expect}[1]{\mathbb E\left[{#1}\right]}
\newcommand*{\prob}[1]{\Pr\left[{#1}\right]}

\newcommand{\arxiv}[1]{}

\DeclareMathOperator{\support}{supp}
\DeclareMathOperator{\Interior}{Int}
\DeclareMathOperator{\Exterior}{Ext}
\DeclareMathOperator{\InteriorOpen}{int}

\newcommand{\jp}[1]{\textcolor{orange}{JP: {#1}}}

\usepackage{xcolor}
\usepackage{mdframed}

\newenvironment{subproof}[1][\proofname]
{%
	\renewcommand*{\proofname}{Proof}
	\proof[#1]%
	
}
{%
	\endproof%
}

\newtheorem{theorem}{Theorem}[section]
\newtheorem{lemma}[theorem]{Lemma}
\newtheorem{claim}[theorem]{Claim}
\newtheorem{proposition}[theorem]{Proposition}
\newtheorem{corollary}[theorem]{Corollary}

\theoremstyle{definition}

\newtheorem{definition}[theorem]{Definition}

\newtheorem{question}{Question}
\usepackage[normalem]{ulem}
\usepackage{authblk}

\title{Uncrossed Multiflows and Applications to Disjoint Paths}

\author[1]{Chandra Chekuri}
\author[2]{Guyslain Naves}
\author[3]{Joseph Poremba}
\author[3]{F. Bruce Shepherd}
\affil[1]{Department of Computer Science, University of Illinois, Urbana-Champaign}
\affil[2]{Laboratoire d'Informatique et des Systèmes, Aix-Marseille Université}
\affil[3]{Department of Computer Science, University of British Columbia}

\date{\vspace{-1cm}}

\begin{document}

\maketitle


\begin{abstract}
A multiflow in a planar graph is {\em uncrossed} if the curves identified by its support paths do not cross in the plane.
Recently, uncrossed flows have played a role in approximation algorithms for maximum disjoint paths in ``fully-planar'' instances, where the combined supply-plus-demand graph is planar.
They are also used in algorithms to find low-congestion unsplittable flows for both fully-planar and single-source instances.
For these two instance classes, any fractional multiflow can be converted into one that is uncrossed, which these algorithms then exploit to obtain their results.

We investigate the utility of uncrossed flow more generally and
ask three key questions.
First, are there other interesting planar multiflow instances that admit uncrossed flows (beyond fully-planar and single-source)?
We answer affirmatively, demonstrating a new family of ``pairwise-planar'' instances whose fractional flows can be uncrossed.
This family subsumes fully-planar but includes substantially more, such as (2-connected) fully-compliant series-parallel instances and some instances that have large clique demand graphs. 
Second, can we always round a fractional uncrossed flow to a ``good'' integral flow?
We again answer positively.
For maximization problems, we show any fractional uncrossed flow can be rounded to an integral flow with a constant fraction of its value.
For congestion problems (where we must fully route all given demands), we give a rounding procedure that yields an integral multiflow with edge congestion 2.
Consequently, we obtain constant-factor approximation algorithms for maximum disjoint paths and minimum congestion integer multiflow for pairwise-planar instances, and show such instances have a constant integral flow-multicut gap.
Finally we ask, given an arbitrary planar instance, can we determine if there exists a congestion-1 uncrossed fractional flow (congestion setting) or find the maximum value uncrossed fractional flow (maximization setting)?
For congestion, we show this problem is NP-hard, but finding uncrossed edge-disjoint paths is polytime solvable if the demands span a bounded number of faces.
For maximization, we present a strong (almost-polynomial) inapproximability result.


\end{abstract}


\section{Introduction}

In {\em multicommodity flow}, we are given an edge-capacitated {\em supply graph} $(G, u)$ and a set of source-sink pairs called {\em demands} or {\em commodities}.
The demands are represented by a {\em demand graph} $H$ on the same node set as $G$.
Unless otherwise specified, graphs are {\em undirected}.
We seek to route flow for the demands simultaneously, while respecting the capacities of the supply graph.
There are two flavours of the multiflow problem.
In the {\em maximization} setting, the objective is to find a flow $x$ that respects the capacities and maximizes the total sum of demands being routed.
A feasible flow for $(G, H, u)$ needs to respect the capacities, but has freedom in the quantities of each demand that it routes.
In the {\em congestion} setting, we are given target amounts $d(s, t)$ for each $(s, t) \in E(H)$. To be feasible for $(G, H, u, d)$, a flow must (in addition to respecting capacities) fully route this amount for each demand. 
The objective is to determine whether $d$ can be satisfied in the capacities $u$, or to find the minimum scaling factor $\alpha \ge 1$ such that $d$ can be satisfied in capacities $\alpha u$.

Both fractional and integer versions of these problems are fundamental in combinatorial optimization (cf. \cite{schrijver2003combinatorial}).
Due to applications, the literature often emphasizes problems where the supply graph is planar \cite{korte1990paths}.
The integer problems are hard even for planar supply graphs with 3 demand edges \cite{vygen1995np}.  There exist, however, good approximation algorithms  in so-called {\em fully-planar} instances where the joint supply-demand graph $G+H$ is planar\footnote{In this paper, if we say a multiflow instance is {\em planar} we only mean that the supply graph $G$ is planar. If we mean $G+H$ is planar, we will say fully-planar.}.
Often approximation algorithms follow a ``relax-and-round'' approach: relax the integer problem to its fractional version, find a fractional solution with ``nice properties'', and then round the nice fractional solution to an integral solution.

A recurring ``nice property'' 
for planar multiflow is that the support paths (i.e. paths routing non-zero flow) do not cross in the plane.
A multiflow satisfying this property is {\em uncrossed} (see \Cref{fig:fully-planar-ex}).
It is known how to construct uncrossed multiflows for single-source ($H$ is a star) and fully-planar instances, the latter using an algorithm of Matsumoto et al. \cite{matsumoto1986planar}.
This was recently leveraged to give constant-factor approximation algorithms for {\em maximum edge-disjoint paths} (MEDP) \cite{garg2022integer} and {\em maximum node-disjoint paths} (MNDP) \cite{schlomberg2023packing,schlomberg2024improved} for fully-planar instances.
Very recently uncrossed flows were used to round fractional flows to {\em unsplittable} flows, routing the same demand quantities with small additive error, in both single-source \cite{traub2024single} and fully-planar contexts \cite{espinosa2026unsplittable}.
In an unsplittable flow, each demand is routed on a single path.

\begin{figure}[htbp]
\begin{subfigure}{0.5\textwidth}
\centering
\begin{tikzpicture}[scale=0.7]
    \node (s1) at (0, 0) {};
    \node (t1) at (0, 4) {};
    \node (s2) at (4, 0) {};
    \node (t2) at (4, 4) {};

    \node (x1) at (2, 1) {};
    \node (x2) at (2, 3) {};
    
    \node (y1) at (1, 2) {};
    \node (y2) at (3, 2) {};

    \foreach \v in {s1,t1,s2,t2,x1,x2,y1,y2}  {    
        \draw[fill] (\v) circle (3pt);
    }

    \draw[black] (s1) to (x1);
    \draw[black] (s2) to (x1);
    \draw[black] (x1) to (y1);
    \draw[black] (x1) to (y2);
    \draw[black] (y1) to (x2);
    \draw[black] (y2) to (x2);
    \draw[black] (x2) to (t1);
    \draw[black] (x2) to (t2);

    \draw[dashed] (s1) to (t1);
    \draw[dashed] (s2) to (t2);

    \foreach \a / \b in {s1/x1,x1/y2,y2/x2,x2/t1} {
        \draw[line width=3pt, semitransparent, blue] (\a) to (\b);
    }
    \foreach \a / \b in {s2/x1,x1/y1,y1/x2,x2/t2} {
        \draw[line width=3pt, semitransparent, ForestGreen] (\a) to (\b);
    }
\end{tikzpicture}
\label{fig:fully-planar-ex1}
\end{subfigure}%
\begin{subfigure}{0.5\textwidth}
    \centering
\begin{tikzpicture}[scale=0.7]
    \node (s1) at (0, 0) {};
    \node (t1) at (0, 4) {};
    \node (s2) at (4, 0) {};
    \node (t2) at (4, 4) {};

    \node (x1) at (2, 1) {};
    \node (x2) at (2, 3) {};
    
    \node (y1) at (1, 2) {};
    \node (y2) at (3, 2) {};

    \foreach \v in {s1,t1,s2,t2,x1,x2,y1,y2}  {    
        \draw[fill] (\v) circle (3pt);
    }

    \draw[black] (s1) to (x1);
    \draw[black] (s2) to (x1);
    \draw[black] (x1) to (y1);
    \draw[black] (x1) to (y2);
    \draw[black] (y1) to (x2);
    \draw[black] (y2) to (x2);
    \draw[black] (x2) to (t1);
    \draw[black] (x2) to (t2);

    \draw[dashed] (s1) to (t1);
    \draw[dashed] (s2) to (t2);

    \foreach \a / \b in {s1/x1,x1/y1,y1/x2,x2/t1} {
        \draw[line width=3pt, semitransparent, blue] (\a) to (\b);
    }
    \foreach \a / \b in {s2/x1,x1/y2,y2/x2,x2/t2} {
        \draw[line width=3pt, semitransparent, ForestGreen] (\a) to (\b);
    }
\end{tikzpicture}
\label{fig:fully-planar-ex2}
\end{subfigure}
\caption{A fully-planar multiflow instance with two different multiflows routing its demands. Dashed edges represent demands, and coloured paths are support paths of the flow. The support paths on the left cross, while the multiflow on the right is uncrossed.}
\label{fig:fully-planar-ex}
\end{figure}

In the above results, the uncrossed property of the fractional solution plays a key role, but the arguments are still specialized to the structure of $(G, H)$.
For example, for the fully-planar context, the uncrossed flow can be shown to be {\em laminar}: if one considers the family of cycles obtained by joining demand edges of $H$ to their support paths in $G$, the result is a nested family of cycles in the plane.
However, this same property does not hold for all instances that admit uncrossed flows.
For instance, if $G$ is a cycle and $H$ consists of mutually crossing demands, then no paths within $G$ cross, but the path-plus-demand cycles are not nested.
Our aim is to expand on a theory of uncrossed flows - when can we find them, and do they always admit good rounding algorithms, even outside the laminar/single-source cases?
We have three motivating questions.


\begin{question}
    \label{ques:uncrossable-classes}
    Do there exist other (interesting) families of $(G, H)$ satisfying the following property: for all capacities $u$ and demand values $d$, if $(G, H, u, d)$ has a feasible (fractional) multiflow $x$, then $(G, H, u, d)$ has a feasible multiflow $y$ that is uncrossed?
\end{question}

We say $(G, H)$ is an {\em uncrossable layout} if it satisfies the conditions of \Cref{ques:uncrossable-classes}.
For instance, if $(G, H)$ is fully-planar or $H$ is single-source, then $(G, H)$ is an uncrossable layout.


\begin{question}
    \label{ques:rounding}
    Suppose we are given a feasible fractional multiflow $x$ that is uncrossed (not assuming that the multiflow $x$ is laminar).
    Can we round $x$ to a ``good'' integer multiflow $y$? Specifically:
    \begin{enumerate}[label=\roman*.]
        \item (Congestion Setting) Given feasible uncrossed fractional multiflow $x$ for $(G, H, u, d)$, can we compute an integer multiflow $y$ feasible for $(G, H, \alpha u, d)$, for some constant $\alpha$?
        \item (Maximization Setting) Given feasible uncrossed fractional multiflow $x$ for $(G, H, u)$, can we compute an integer multiflow $y$ feasible for $(G, H, u)$ whose value is $1/\beta$ times that of $x$, for some constant $\beta$?
    \end{enumerate}
\end{question}

The answer to \Cref{ques:uncrossable-classes} and both parts of \Cref{ques:rounding} is: yes!
First, we define a new class of $(G, H)$ and prove it gives uncrossable layouts.

\begin{theorem}
    \label{thm:pairwise-planar-intro}
    Let $G$ be a planar graph with a fixed embedding.
    Let $H$ be such that for any pair of demands $h_1, h_2 \in E(H)$, $h_1$ and $h_2$ can be embedded inside (possibly different) faces of $G$ so they do not cross each other.
    Then $(G, H)$ is an uncrossable layout.
\end{theorem}

We say any $(G, H)$ satisfying this property is {\em pairwise-planar} (see \Cref{fig:pairwise-planar} in \Cref{sec:new-classes} for a depiction).
The pairwise-planar class includes fully-planar as a subset, but includes substantially more. For instance, it contains layouts where  $G + H$ is highly non-planar (and thus, do not admit laminar flows). It also includes the 2-connected {\em fully-compliant}\footnote{$(G, H)$ is fully-compliant if $G$ is series-parallel and $G+h$ is series-parallel for every $h \in E(H)$.} family introduced in \cite{chekuri2013flow}, which were not previously known to give uncrossed layouts to our knowledge.
Pairwise-planar also differs from previous uncrossable layouts in an important structural sense.
Namely, fully-planar and single-source $(G, H)$ are {\em cut-sufficient} (essentially, they have a Max-Flow Min-Cut theorem in the congestion setting), while some pairwise-planar $(G, H)$ are not.



As for \Cref{ques:rounding}, we prove the following results.
We begin with the maximization setting.

\begin{theorem}
    \label{thm:MEDP}
    Let $(G, H, u)$ be a planar maximization instance where $u$ is integral.
    Let $f$ be a feasible fractional flow that is uncrossed.
    There exists a feasible integral flow $f'$ that achieves an $\Omega(1)$ fraction of the value of $f$.
    This integral flow can be computed in polynomial time.
\end{theorem}
We prove this result by establishing an ``approximate integer decomposition property'', exploiting colouring results on a restricted class of string graphs (path intersection graphs).
This integer decomposition also implies $\Omega(1)$ rounding holds for the {\em weighted} maximization version, where each $(s, t) \in E(H)$ has an associated weight $w(s,t)$, and a flow achieves $w(s, t)$ value for each unit of $(s, t)$-flow that it routes.
Additionally, this result extends to {\em node-capacitated} multiflows (i.e. $u$ is a capacity function on nodes rather than edges).

For the congestion setting, we give a factor 2 rounding for converting a fractional uncrossed multiflow into an integral multiflow.
Moreover, for an extra additive penalty we can round to an unsplittable flow.
\begin{theorem}
    \label{thm:congestion}
    Let $(G, H, u, d)$ be a planar congestion instance where $u, d$ are integral.
    Let $f$ be a feasible fractional multiflow that is strongly uncrossed\footnote{We defer the more technical definition of strongly uncrossed to Section~\ref{sec:strongly-uncrossed}. Importantly, the uncrossable layouts mentioned in this paper, including pairwise-planar, admit strongly uncrossed flows.}.
    Then there exists both:
    \begin{enumerate}
        \item an integral multiflow $f'$ that is feasible for $(G, H, u', d)$ where $u'(e) \le 2u(e)$ for each $e \in E(G)$, and
        \item an unsplittable multiflow $f''$ that is feasible for $(G, H, u'', d)$ where $u''(e) < 2u(e) + d_{\max}$ for every $e \in E(G)$, where $d_{\max} = \max_{h \in E(H)} d(h)$.
    \end{enumerate}
    The integral and unsplittable multiflows can be computed in polynomial time.
\end{theorem}

Using Theorems \ref{thm:pairwise-planar-intro}, \ref{thm:MEDP}, and \ref{thm:congestion}, we obtain constant-factor approximation algorithms for integer multiflow problems with pairwise-planar inputs.
\begin{corollary}
\label{cor:pairwise-planar-algorithm}
    There exist polynomial time constant-factor approximation algorithms for maximum edge-disjoint paths and maximum node-disjoint paths (and more generally, maximum integer multiflow), as well as minimum congestion integer multiflow, for inputs where $(G, H)$ is pairwise-planar. 
\end{corollary}


Another interesting corollary concerns the {\em integral flow-multicut gap}, which is the ratio of the minimum capacity of a multicut (a set $F \subseteq E(G)$ whose removal destroys all paths connecting demands of $H$) to the maximum value of an {\em integral} multiflow.
Using their flow rounding algorithms, \cite{huang2021approximation} \cite{garg2022integer} show the gap is constant for fully-planar $(G, H)$, but the gap can be as large as $\Theta(|E(H)|)$ even for planar supply graphs \cite{garg_primal-dual_1997}.
Combining our results with a theorem of \cite{tardos_improved_1993}, we obtain a constant gap for pairwise-planar instances.

\begin{corollary}
\label{cor:pairwise-planar-multicut-gap}
    For any maximization instance $(G, H, u)$ where $(G, H)$ is pairwise-planar, the minimum capacity of a multicut is at most $O(1)$ times the maximum value of an integral multiflow.
\end{corollary}



The third key question requires slightly more motivation.
If $(G, H)$ is not an uncrossable layout, it only means the structure of $(G, H)$ alone is insufficient to guarantee uncrossed solutions.
However for specific capacities $u$ and demands $d$ we may still be able to find an uncrossed flow and round it.
Thus it would be nice to have algorithms to determine whether (in the congestion setting) we can find a feasible fractional flow that is uncrossed, or (in the maximization setting) find the largest feasible flow that is uncrossed.

\begin{question}
    \label{ques:algorithmic}
    Suppose we are given an {\em arbitrary} planar multiflow instance.
    Can we solve the fractional multiflow problem if we constrain the feasible space to only allow uncrossed solutions? Specifically:
    \begin{enumerate}[label=\roman*.]
        \item\label{ques:algorithmic-congestion} (Congestion Setting) Given $(G, H, u, d)$, can we determine whether there exists a feasible fractional multiflow that is uncrossed?
        \item (Maximization Setting) Given $(G, H, u)$, can we compute a feasible fractional uncrossed multiflow whose value is (approximately) maximum amongst feasible uncrossed solutions ({\em maximum uncrossed multiflow})?
    \end{enumerate}
\end{question}

These questions concern {\em fractional} multiflows and may feel like linear optimization problems.
However the requirement of uncrossed flow paths makes these computationally much more challenging.
Interestingly, the positive result of Theorem~\ref{thm:MEDP} can be combined with results of \cite{chuzhoy2018almost,chuzhoyjournalv017a006} to establish strong almost-polynomial inapproximability for maximum uncrossed fractional multiflow (for both edge and node capacities).

\begin{theorem}
\label{thm:inapprox}
    For planar $G$, maximum uncrossed fractional multiflow is not approximable to within a factor of $2^{\Omega(\log^{1-\epsilon} n)}$ for any $\epsilon>0$ assuming that $NP \not\subseteq DTIME(n^{polylog~n})$.
Moreover, there is no $n^{O\left(\frac{1}{(\log \log n)^2}\right)}$-approximation 
polynomial-time algorithm assuming that for some constant $\delta > 0$,
$NP \not\subseteq DTIME(2^{n^{\delta}})$.
These results hold even when $G$ is maximum degree $3$ (or even $G$ is a wall graph) and the capacities are unit.
\end{theorem}


For the congestion setting, we prove that the problem posed by \Cref{ques:algorithmic} is NP-complete.
\begin{theorem} 
\label{thm:npc}
It is NP-complete to determine whether planar congestion instance $(G, H, u, d)$ has a feasible fractional uncrossed multiflow, even with unit-valued capacities and demands.
\end{theorem}

We remark that since this problem is not a linear program anymore, it is not even obvious that the problem is in NP, however we establish this
in \Cref{sec:poly-size}.
On the positive side, we show that finding uncrossed edge-disjoint paths can be (polytime) reduced to a node-disjoint paths instance.
Hence we may find uncrossed solutions if $E(H)$ is incident with a bounded number of faces using a result of Schrijver \cite{schrijver1991disjoint}.

\begin{theorem}
\label{thm:reduction}
For instances where $G$ is planar, the edges of $H$ are incident with a bounded number of faces, and the capacities and demands are unit-valued, there is a polynomial time algorithm to find a feasible integral uncrossed multiflow for $(G, H, u, d)$ or declare that none exists. 
\end{theorem}

\subsection{Organization of the Paper}
\Cref{sec:prelims} lays the groundwork for the theory of uncrossed flows, including detailed definitions.
\Cref{sec:new-classes} discusses pairwise-planar instances, proves they admit uncrossed flows, and shows a constant integral flow-multicut gap.
\Cref{sec:maximization-all} covers our maximization results, both rounding (\Cref{thm:MEDP}) and inapproximability (\Cref{thm:inapprox}).
The rounding result for congestion (\Cref{thm:congestion}) appears in \Cref{sec:congestion}.
The NP-completeness result (\Cref{thm:npc}) is deferred to the appendix (\Cref{sec:npc}).
\Cref{sec:reduction} contains the algorithm to solve integral congestion instances when demands are incident with a bounded number of faces (\Cref{thm:reduction}).

\subsection{Related Work}

Uncrossed flows have played a role in previous routing algorithms in planar graphs, such as Schrijver's Homotopy Method \cite{schrijver1990homotopic} for node-disjoint paths.
Most relevant are recent algorithms for maximum disjoint paths in fully-planar instances ($G+H$ planar).
In \cite{huang2021approximation,garg2022integer}, constant-factor approximation algorithms are given for maximum edge-disjoint path, then extended to the node-disjoint setting in \cite{schlomberg2023packing,schlomberg2024improved}.
These algorithms glue uncrossed flow paths to their respective demand edges to obtain an uncrossed (laminar) family of cycles in a planar graph.
Our rounding algorithms are agnostic in the sense that they just operate on uncrossed flow paths, not assuming
that demands lie within a face or that demands of a common face are non-crossing (indeed for our pairwise-planar instances, it is even possible for demands within a face to form a clique).


Uncrossed flows for fully-planar instances are typically constructed by a primal algorithm \cite{matsumoto1986planar} that computes an uncrossed flow from scratch.
We build on a different approach, taking a feasible flow (with many crossings) and iteratively decreasing total crossing (used for example in \cite{huang2023approximating} to find cycles in higher-genus surfaces that cross at most once).
We discuss obstacles faced in extending this to pairwise-planar in Section~\ref{sec:new-classes}.

In \cite{garg2022integer} the initial flow is first rounded to a half-integral flow (using linear programming theory such as total unimodularity), then this half-integral flow is rounded to an integral flow using a 4-colouring of the (planar) intersection graph of the cycles.
Extending to the node-disjoint setting, \cite{schlomberg2023packing} uses a different argument, iteratively constructing a high-value solution by showing one can always find an ``efficient'' cycle that only intersects the other cycles on a small set of its nodes.
Our rounding technique for \Cref{thm:MEDP} instead uses an approach based on establishing an ``approximate integer decomposition property'' \cite{chekuri2007multicommodity,chekuri2009approximate}.
A key ingredient comes from looking at a subset of {\em string graphs}, which are the intersection graphs from curves in the plane.
Specifically, we are interested in those that arise from curves that do not cross, studied for example in \cite{esperet_coloring_2009}, and for which an important colouring result was proven by Fox and Pach \cite{fox_touching_nodate}\footnote{This manuscript is unpublished, but the result is referenced in \cite{van_batenburg_coloring_2017} and their writeup was shared with us \cite{pach}.}.

The fully-planar maximum disjoint paths algorithms were extended \cite{schlomberg2023packing,huang2023approximating} to uncrossable families of cycles in higher-genus surfaces, yielding approximation algorithms for maximum disjoint paths when $G+H$ is embedded in a surface of genus $g$; the approximation ratio is polynomial in the genus $g$.
By contrast our work extends uncrossing techniques to a broader family of multiflow instances (pairwise-planar) while remaining {\em in the plane}.
Some pairwise-planar instances have $G+H$ of arbitrarily high genus (e.g. $H$ contains a large clique), but we show nevertheless that they admit uncrossed flows in the plane  in order to achieve constant-factor approximations.
Results in \cite{schlomberg2023packing} use the context of uncrossed cycle families, as defined by Goemans and Williamson \cite{goemans1998primal}, which is a combinatorial rather than topological definition of crossing.

The works of \cite{huang2021approximation}, \cite{garg2022integer} use their rounding algorithms to establish that the integral flow-multicut gap of fully-planar instances is constant.
The integral flow-multicut gap can be as large as $\Theta(|E(H)|)$ for general planar $G$ \cite{garg_primal-dual_1997}, so further restrictions on $(G, H)$ are needed to establish constant bounds.
Other classes known to give constant integral flow-multicut gaps include: when $G$ is a tree \cite{garg_primal-dual_1997}, when $G + H$ is series-parallel \cite{cornaz_max-multiflowmin-multicut_2011} (note that these are less general than fully-compliant), when $G$ has fixed cyclomatic number \cite{bentz_disjoint_2009}, when $G+H$ has fixed genus \cite{huang2023approximating}, and when $G$ has fixed tree-width \cite{chekuri2013maximum}.
A common ingredient in these results is a bound on the {\em fractional} flow-multicut gap for $K_{r,r}$-minor free $G$ by \cite{tardos_improved_1993, KPR}.

Other works using uncrossed flow come from the study of unsplittable flow in the congestion model (where specified demand quantities must be fully routed). 
For directed single-source planar instances ($H$ is a star) \cite{traub2024single}, it is shown how to convert a feasible fractional uncrossed flow into an unsplittable flow that exceeds capacities by a small amount and has no increase in cost.
For (undirected) fully-planar instances, \cite{espinosa2026unsplittable} shows how to convert a laminar flow into an unsplittable flow $y$ with small capacity violation (no cost guarantee).

\section{Uncrossed Multiflows}
\label{sec:prelims}

\subsection{Multiflow Preliminaries}
All graphs are undirected in this paper.
Consider a supply graph $G = (V, E(G))$ and demand graph $H = (V, E(H))$.
Nodes that are ends of an edge in $H$ are {\em terminals}, and the edges themselves are {\em demands} or {\em commodities}.
Together, $(G, H)$ define a {\em layout}\footnote{Sometimes in the literature, $(G, H)$ is called a {\em multiflow topology}. We avoid this term to avoid confusion with the topology of the plane.
}. 
For nodes $s, t \in V$, let $\mathcal P_{st}$ denote the set of $st$-paths.
By {\em path} we always mean simple path.
A {\em multiflow} is a function $x: \mathcal P \to \mathbb R_{\ge 0}$, where $\mathcal P = \bigcup_{(s, t) \in E(H)} \mathcal P_{st}$.
Multiflows are fractional unless explicitly stated as integral.
A multiflow can also be viewed as a collection of single-commodity flows $(x^{st}: (s, t) \in E(H))$.
For $e \in E(G)$, we let $x(e) = \sum_{P \in \mathcal P : e \in P} x(P)$ (the amount of flow on $e$), and for $(s, t) \in E(H)$ we let $|x^{st}| = \sum_{P \in \mathcal P_{st}} x(P)$ (the amount of demand $(s, t)$ that is routed).

A {\em maximization instance} has the form $(G, H, u)$, where $u: E(G) \to \mathbb R_{\ge 0}$ are {\em capacities}.
A multiflow $x$ is {\em feasible} for the instance if for every $e \in E(G)$, $x(e) \le u(e)$.
The objective is to maximize the value of a feasible multiflow, where the value of a multiflow $x$ is $|x| = \sum_{P \in \mathcal P} x(P) = \sum_{(s, t) \in E(H)} |x^{st}|$.
If $u = \vec 1$ and we restrict to integral $x$, this problem is the {\em maximum edge-disjoint paths} problem (MEDP).
A {\em weighted} maximization instance also comes equipped with per-unit profits $w: E(H) \to \mathbb R_{\ge 0}$, in which case the value of $x$ is instead defined $|x|_w = \sum_{(s, t) \in E(H)} w(s,t) |x^{st}|$.

A {\em congestion instance} has the form $(G, H, u, d)$, where $u: E(G) \to \mathbb R_{\ge 0}$ are capacities and $d: E(H) \to \mathbb R_{\ge 0}$ are {\em demand values}.
A multiflow $x$ is {\em feasible} for the instance if for every $e \in E(G)$ we have $x(e) \le u(e)$ and for every $(s, t) \in E(H)$ we have $|x^{st}| = d(s, t)$.
As a decision problem, the goal is to determine whether $(G, H, u, d)$ has a feasible multiflow.
As an optimization problem, the objective is to find the minimum $\alpha \ge 1$ such that $(G, H, \alpha u, d)$ has a feasible multiflow.

In either setting, we say an instance is {\em planar} if $G$ is planar.
We emphasize that $G + H$ is not necessarily planar.
We always assume a fixed embedding of planar graph $G$.


\subsection{Uncrossed Paths and Multiflows}
\label{sec:uncrossed}

First, we define what it means for paths to cross. 
Edge-disjoint paths can cross at a common node, but in general we need to consider common subpaths (see \Cref{fig:crossing}).

\begin{figure}[htbp]
\begin{subfigure}[c]{0.23\textwidth}
\centering
\begin{tikzpicture}
	\node  (v) at (0, 0) {};
	\node  (1) at (-1, 1) {};
	\node  (2) at (1, 1) {};
	\node  (3) at (-1, -1) {};
    \node  (4) at (1, -1) {};
    \draw[fill] (v) circle (3pt);
    \draw[black] (v) to node[below left=-1]{$e$} (1);
    \draw[line width=3pt, semitransparent, blue] (v) to (1);
    \draw[black] (v) to node[below right=-1]{$g$} (2);
    \draw[line width=3pt, semitransparent, ForestGreen] (v) to (2);
    \draw[black] (v) to node[above left=-1]{$g'$} (3);
    \draw[line width=3pt, semitransparent, ForestGreen] (v) to (3);
    \draw[black] (v) to node[above right=-1]{$e'$}(4);
    \draw[line width=3pt, semitransparent, blue] (v) to (4);
\end{tikzpicture}
\label{fig:crossing-node}
\end{subfigure}%
\hspace{-0.03\textwidth}
\begin{subfigure}[c]{0.3\textwidth}
\centering
\begin{tikzpicture}
	\node  (z1) at (-2, 0) {};
    \draw (z1) node[above right] {$z$};
    \node  (z2) at (1, -2) {};
    \draw (z2) node[below left] {$z'$};
	\node  (1) at (-3, -1) {};
	\node  (2) at (-3, 0) {};
	\node  (3) at (2, -2) {};
    \node  (4) at (2, -1) {};
    \node  (m1) at (-1, -1) {};
    \node  (m2) at (0, -1) {};
    
    \draw[fill] (z1) circle (3pt);
    \draw[fill] (z2) circle (3pt);
    \draw[fill] (m1) circle (3pt);
    \draw[fill] (m2) circle (3pt);
    \draw[black,double] (z1) to (m1);     
    \draw[black,double] (m1) to (m2);
    \draw[black,double] (z2) to (m2);
    \draw[blue] (z1) to node[below right=-1]{$e$} (1);
    \draw[line width=3pt, semitransparent, blue] (z1) to (1);
    \draw[black] (z1) to node[above left=-1]{$g$} (2);
    \draw[line width=3pt, semitransparent, ForestGreen] (z1) to(2);
    \draw[black] (z2) to node[above right=-1]{$g'$} (3);
    \draw[line width=3pt, semitransparent, ForestGreen] (z2) to (3);
    \draw[black] (z2) to node[above left=-1]{$e'$} (4);
    \draw[line width=3pt, semitransparent, blue] (z2) to (4);
\end{tikzpicture}
\label{fig:crossing-subpath}
\end{subfigure}%
\hspace{0.1\textwidth}
\begin{subfigure}[c]{0.3\textwidth}
\centering
\begin{tikzpicture}
	\node  (z1) at (-2, 0) {};
    \draw (z1) node[above right] {$z$};
    \node  (z2) at (1, -2) {};
    \draw (z2) node[below left] {$z'$};
	\node  (1) at (-3, -1) {};
	\node  (2) at (-3, 0) {};
	\node  (3) at (2, -2) {};
    \node  (4) at (2, -1) {};
    \node  (m1) at (-1, -1) {};
    \node  (m2) at (0, -1) {};
    
    \draw[fill] (z1) circle (3pt);
    \draw[fill] (z2) circle (3pt);
    \draw[fill] (m1) circle (3pt);
    \draw[fill] (m2) circle (3pt);
    \draw[black,double] (z1) to (m1);
    \draw[black,double] (m1) to (m2);
    \draw[black,double] (z2) to (m2);
    \draw[black] (z1) to node[below right=-1]{$g$} (1);
    \draw[line width=3pt, semitransparent, ForestGreen] (z1) to (1);
    \draw[black] (z1) to node[above left=-1]{$e$} (2);
    \draw[line width=3pt, semitransparent, blue] (z1) to (2);
    \draw[black] (z2) to node[above right=-1]{$g'$} (3);
    \draw[line width=3pt, semitransparent, ForestGreen] (z2) to (3);
    \draw[black] (z2) to node[above left=-1]{$e'$}(4);
    \draw[line width=3pt, semitransparent, blue] (z2) to (4);
\end{tikzpicture}
\label{fig:uncrossed-subpath}
\end{subfigure}
\caption{(Left) Two edge-disjoint paths that cross at a common node. (Center) Two paths that cross at a common subpath. (Right) Two paths that share a common subpath, but are uncrossed. Uncoloured edges are shared between the two paths.}
\label{fig:crossing}
\end{figure}

\begin{definition}
    [Crossing]
    \label{def:crossing-subpath}
    Let $Z$ be a maximal shared subpath of paths $P$ and $Q$ that does not include any of their endpoints.
    Say the ends of $Z$ are $z$ and $z'$.
    The edges of $(P \cup Q) \setminus Z$ incident with $z$ or $z'$ have a clockwise ordering about $Z$ in the plane.
    The paths {\em cross} at $Z$ if the edges of $P$ and $Q$ alternate clockwise about $Z$.
    That is, if $P = P_1, e, Z, e', P_2$ and $Q = Q_1, g, Z, g', P_2$, the edges occur $e, g, e', g'$ or $e, g', e', g$ clockwise about $Z$.
    Two paths are {\em uncrossed} if they do not cross at any such subpath\footnote{We emphasize that paths can only cross at maximal common subpaths that do {\em not} contain their endpoints. Both paths must have edges before and after the common subpath in order to cross. For example, if $P$ contains an endpoint of $Q$ and $Z$ is the maximal common subpath containing that endpoint, then $P$ and $Q$ do not cross at $Z$.}.
\end{definition}



The {\em support paths} of a multiflow $x$ are the paths with $x(P) > 0$.
The following is the natural extension of ``uncrossed'' to the multiflow setting (though we discuss a stronger condition in \Cref{sec:strongly-uncrossed}).

\begin{definition}[Uncrossed Multiflow]
    \label{def:uncrossed-multiflow}
    A multiflow is {\em uncrossed} if its support paths are pairwise uncrossed.
\end{definition}

We are interested in when a feasible multiflow $x$ can be converted into an uncrossed one that routes the same demands.
Our discussion here focuses on the congestion setting, but still applies to the maximization setting (since we can define target demand values $d$ equal to whatever quantities $x$ routes).
As shown in \Cref{fig:problematic-instances}, even if $(G, H, u, d)$ is feasible, it may not have an uncrossed feasible multiflow\footnote{In fact, it is not always possible to find an uncrossed flow with bounded congestion. One may modify the canonical grid instance \cite{GargVY97} so it has a half-integral routing of demands, but any uncrossed multiflow incurs polynomial congestion.}.
However several classes do admit uncrossed solutions.
We say $(G,H)$ is an {\em uncrossable layout} if for all $u, d$, if $(G, H, u, d)$ has a feasible multiflow, then $(G, H, u, d)$ has a feasible uncrossed multiflow.
Matsumoto et al. \cite{matsumoto1986planar} give an algorithm that, given a feasible fully-planar ($G+H$ planar) instance, produces a feasible uncrossed flow.
Thus fully-planar layouts are uncrossable.
It is also well-known that single-source layouts ($H$ is a star) are uncrossable.
Another uncrossable layout is given when $G$ is a cycle with parallel edges and $H$ is arbitrary, which we call {\em ring} layouts.
In \Cref{sec:new-classes} we introduce {\em pairwise-planar} layouts and prove they are uncrossable.

\begin{figure}[htbp]
\begin{subfigure}{0.5\textwidth}
\centering
\begin{tikzpicture}[scale=0.7]
    \node (s1) at (180:2) {};
    \draw (s1) node[left] {$s_1$};
    \node (t1) at (0:2) {};
    \draw (t1) node[right] {$t_1$};
    \node (s2) at (-90:2) {};
    \draw (s2) node[below] {$s_2$};
    \node (t2) at (90:2) {};
    \draw (t2) node[above] {$t_2$};

    \node (rs1) at (180:1) {};
    \node (rt1) at (0:1) {};
    \node (rs2) at (-90:1) {};
    \node (rt2) at (90:1) {};

    \foreach \v in {s1,t1,s2,t2,rs1,rt1,rs2,rt2}  {    
        \draw[fill] (\v) circle (3pt);
    }

    \foreach \v in {s1,t1,s2,t2} {
        \draw[black] (\v) to (r\v);
    }

    \draw[black] (rs1) to[bend left] (rt2);
    \draw[black] (rt2) to[bend left] (rt1);
    \draw[black] (rt1) to[bend left] (rs2);
    \draw[black] (rs2) to[bend left] (rs1);

    \draw[black,dashed] (s1) to[bend right=20] (t1);
    \draw[black,dashed] (s2) to[bend left=20] (t2);
\end{tikzpicture}
\label{fig:notuncrossable}
\end{subfigure}%
\begin{subfigure}{0.5\textwidth}
    \centering
\begin{tikzpicture}[scale=0.7]
	\node  (s1) at (-2, 0) {};
    \draw (s1) node[left = 2] {$s_1$};
	\node  (t1) at (2, 0) {};
    \draw (t1) node[right = 2] {$t_1$};
	\node  (s2) at (0, -2) {};
    \draw (s2) node[below = 2] {$s_2$};
	\node  (t2) at (0, 1) {};
    \draw (t2) node[above = 2] {$t_2$};
    \node  (x) at (0, 0) {};
    \node (y1) at (-1, 2) {};
    \node (y2) at (1, 2) {};
    \draw[fill] (s1) circle (3pt);
    \draw[fill] (t1) circle (3pt);
    \draw[fill] (s2) circle (3pt);  
    \draw[fill] (t2) circle (3pt);
    \draw[fill] (x) circle (3pt); 
    \draw[fill] (y1) circle (3pt);
    \draw[fill] (y2) circle (3pt);
    \draw[black] (s1) to (x) to (t1);
    \draw[black] (s2) to (x) to (t2);
    \draw[black] (x) to[bend left] (y1);
    \draw[black] (y1) to (y2);
    \draw[black] (y2) to[bend left] (x);
    \draw[dashed] (s1) to[bend right] (t1);
    \draw[dashed] (s2) to[bend left] (t2);
\end{tikzpicture}
\label{fig:keyhole}
\end{subfigure}
\caption{Problematic instances for uncrossability. Capacities and demand values are unit. (Left) A congestion instance with a feasible multiflow but no uncrossed multiflow. (Right) The {\em Keyhole Instance}, a congestion instance with a feasible uncrossed trail-multiflow but no uncrossed path-multiflow.}
\label{fig:problematic-instances}
\end{figure}

\subsection{Another View on Crossing: Parallelizations}


There are more intuitive ways to think about uncrossed paths that we frequently use.

{\bf Widened Graph:}
Suppose we ``widen'' $G$, making its nodes into small disks 
and its edges into small channels. 
A family of pairwise uncrossed paths can have its paths drawn as curves simultaneously, each curve following the node disks and edge channels corresponding to its path, so that the curves are disjoint.

{\bf Parallelizations:} To avoid dealing with the formal topology of widened graphs, we use an equivalent notion based on ``separating'' the paths at each edge.
Consider a family of paths $\mathcal P$ and for each $e \in E(G)$ let $\mathcal P_e = \set{P \in \mathcal P : e \in E(P)}$.
Let $G'$ be obtained by replacing each $e$ with $|\mathcal P_e|$ parallel copies of itself.
Intuitively, the parallel edges represent the curves within the channel of $e$.


\begin{figure}[htbp]
\begin{subfigure}{0.45\textwidth}
\centering
\begin{tikzpicture}
	\node  (z1) at (-2, 0) {};
    \draw (z1) node[above right] {$z$};
    \node  (z2) at (1, -2) {};
    \draw (z2) node[below left] {$z'$};
	\node  (1) at (-3, -1) {};
	\node  (2) at (-3, 0) {};
	\node  (3) at (2, -2) {};
    \node  (4) at (2, -1) {};
    \node  (m1) at (-1, -1) {};
    \node  (m2) at (0, -1) {};
    
    \draw[fill] (z1) circle (3pt);
    \draw[fill] (z2) circle (3pt);
    \draw[fill] (m1) circle (3pt);
    \draw[fill] (m2) circle (3pt);
    
    \draw[black] (z1) to[bend left=15] (m1);
    \draw[line width=3pt, semitransparent, blue] (z1) to[bend left=15] (m1);
    \draw[black] (m2) to[bend left=15] (z2);
    \draw[line width=3pt, semitransparent, blue] (m2) to[bend left=15] (z2);
    \draw[black] (z1) to[bend right=15] (m1);
    \draw[line width=3pt, semitransparent, ForestGreen] (z1) to[bend right=15] (m1);
    \draw[black] (m2) to[bend right=15] (z2);
    \draw[line width=3pt, semitransparent, ForestGreen] (m2) to[bend right=15] (z2);
    \draw[black] (m1) to[bend left=15] (m2);
    \draw[line width=3pt, semitransparent, blue] (m1) to[bend left=15] (m2);
    \draw[black] (m2) to[bend left=15] (m1);
    \draw[line width=3pt, semitransparent, ForestGreen] (m2) to[bend left=15] (m1);

    \draw[black] (z1) to node[below right=-1]{$g$} (1);
    \draw[line width=3pt, semitransparent, ForestGreen] (z1) to (1);
    \draw[black] (z1) to node[above left=-1]{$e$} (2);
    \draw[line width=3pt, semitransparent, blue] (z1) to  (2);
    \draw[black] (z2) to node[above right=-1]{$g'$} (3);
    \draw[line width=3pt, semitransparent, ForestGreen] (z2) to (3);
    \draw[black] (z2) to node[above left=-1]{$e'$}(4);
    \draw[line width=3pt, semitransparent, blue] (z2) to (4);
\end{tikzpicture}
\label{fig:parallelization}
\end{subfigure}%
\hspace{0.1\textwidth}
\begin{subfigure}{0.45\textwidth}
\centering


    




\begin{tikzpicture}
    \draw[fill,nearly transparent,gray] (0:0) circle[radius=1];
    
    \node (b1) at (30:1) {}; 
    \node (b2) at (150:1) {};
    \node (g1) at (-30:1) {};
    \node (g2) at (-150:1) {};
    \draw[fill=blue] (b1) circle (3pt);
    \draw[fill=blue] (b2) circle (3pt);
    \draw[fill=ForestGreen,thick] (g1) circle (3pt);
    \draw[fill=ForestGreen,thick] (g2) circle (3pt);

    \draw (b1) node[above right] {};
    \draw (b2) node[above=2] {};
    \draw (g1) node[above right=1] {};
    \draw (g2) node[above right=2] {};

    \node (b1') at (-20:3) {};
    \node (b2') at (165:2) {};
    \node (g1') at (-45:2) {};
    \node (g2') at (-145:2) {};

    \draw[black] (b1) to[bend left] (b1');
    \draw[line width=3pt, semitransparent, blue] (b1) to[bend left] (b1');
    \draw[black] (b2) to node[above left=-1]{$e$} (b2');
    \draw[line width=3pt, semitransparent, blue] (b2) to (b2');
    \draw[black] (g1) to (g1');
    \draw[line width=3pt, semitransparent, ForestGreen] (g1) to (g1');
    \draw[black] (g2) to node[below right]{$g$} (g2');
    \draw[line width=3pt, semitransparent, ForestGreen] (g2) to( g2');

    \draw[black] (b1) to (b2);
    \draw[line width=3pt, semitransparent, blue] (b1) to (b2);
    \draw[black] (g1) to (g2);
    \draw[line width=3pt, semitransparent, ForestGreen] (g1) to (g2);

\end{tikzpicture}
\label{fig:ringified}
\end{subfigure}
\caption{(Left) An uncrossed parallelization of the uncrossed paths from \Cref{fig:crossing}. (Right) Disk-expansion of the node $z$ after uncrossed parallelization.}
\label{fig:parallel-disk}
\end{figure}


Now, for each $e \in E(G)$ we decide on some one-to-one mapping $\phi_e$ from $\mathcal P_e$ to the parallel class of $e$ in $G'$.
For each $P \in \mathcal P$, we obtain a path $P_{\phi}$ in $G'$ whose edges are $\set{\phi_e(P) : e \in P}$.
This induces a family of edge-disjoint paths $\mathcal P_{\phi}$ in $G'$.
We call $\mathcal P_{\phi}$ a {\em parallelization} of $\mathcal P$.
If we (arbitrarily) orient each $e \in E(G)$, it gives a linear ordering $\le_e$ of $\mathcal P_e$.
Different choices for $(\phi_e : e \in E(G))$ yield different parallelizations.
The family $\mathcal P$ is pairwise uncrossed if and only if there is some parallelization that is pairwise uncrossed at every node of $G'$ (rather than needing to consider longer shared subpaths).
Using parallelizations helps us simplify subsequent definitions by assuming paths are edge-disjoint.
Given a collection of paths, one can compute in polynomial time an uncrossed parallelization or find a crossing pair of paths.
\resolved{
\joe{mention that there is an algorithm to compute such an ordering}
\bruce{pls do. i think its critical for this project that if we are given a flow we can check if its uncrossed.}
}

{\bf Disk-Expansion:} After parallelization, we sometimes go a step closer to the widened graph idea by ``expanding'' a node $v$.
We draw a small disk with $v$ at its centre (small enough that its boundary intersects only $\delta(v)$, and hits each edge once).
Where an edge copy intersects the boundary, we place a new node.
We terminate the edges at these new nodes.
For an uncrossed parallelization, it is then possible to (simultaneously) draw edges in the open interior of the disk between nodes that correspond to consecutive edges of the same path, such that the drawing is planar.
An example is shown in Figure \ref{fig:parallel-disk}.

\subsection{Strongly Uncrossed Multiflows}
\label{sec:strongly-uncrossed}

\resolved{
\joe{Can we combine this with the intro to shorten it?}
\joe{Also I don't think we only deal with unit capacities and demands anymore}
\bruce{i was worried about that but the only place i am aware of needing it is in the reduction.}
}

\resolved{
\joe{Maximization vs. Congestion - I propose explicitly saying ``congestion instance'', ``maximization instance'' so it is always clear.}
For a maximization instance, a multiflow is {\em feasible} if it respects the capacities.
\bruce{agree. we also sometimes look at maximization instances where we cannot route more than one demand. we dont emphasize this much except we make a remark in section 7 where it is needed to be comsistent with julias work.}
}

\resolved{
\subsection{Uncrossed Multiflows}
\label{sec:multi-uncrossed}
\guyslain{do we really need a new subsection here?}
\bruce{maybe not. i think joe was considering moving it to the intro. a separate section makes it easier for a reader to find. if we have a section then i think here we should say that given a flow we can check if its uncrossed in polytime. my general preference is for an intro to run quickly to the point, and leave some details out but which can be easily. maybe we just merge whatever we want to keep from 2.2 and 2.3?}
}

\resolved{
\joe{I'm commenting out the following, because the idea of innermost edges, etc. hasn't been introduced yet. Moreover, we aren't ``converting'' a multiflow into an uncrossed one - they are all just uncrossed by definition}
}
\resolved{\guyslain{paragraph commented below is confusing. I don't think it is necessary to discuss why rings are not fully-planar. The last sentence is not very relevant, and we will show uncrossable instances that are not fully planar anyways.}\bruce{fine with me to remove. we imply this with our venn diagram so if anyone wants details we can supply in rebuttal or in appendix.}
}

\resolved{
\guyslain{Paragraph commented, I don't think it is quite correct. Adding $T$ breaks planarity. So we would need to solve it ignoring the planar embedding, then cut edges to $T$ into leaves, then embed, then uncross, but an uncrossing that allows to exchange tth endpoints of the two demands (as long as we still have $s$ on one side). Not worth it I believe.}\bruce{agree it would need further work and not worth it.}
}

Consider an uncrossed multiflow.
For any $(s, t) \in E(H)$, its support paths have a clockwise ordering about $s$ (and $t$) given by the parallelization.
We call the regions of the plane bounded between two consecutive paths {\em sectors} for this commodity (defined more formally in \Cref{sec:congestion}).

The proof of \Cref{thm:congestion} uses that sectors of different commodities are non-crossing as sets - that is, we want that if $A, B \subseteq \mathbb R^2$ are sectors of different commodities, then either $A \cap B = \emptyset$ or one contains the other.
One might think this non-crossing property follows from \Cref{def:uncrossed-multiflow}, however \Cref{fig:quasicrossing-super} shows a counterexample\footnote{The desired non-crossing sector property does hold for uncrossed flows in {\em leaf instances}, where terminals all have degree 1.
Moving each terminal to its own new leaf node is a standard pre-processing operation, but this is not safe for uncrossed flows - it may turn an uncrossable layout into one that is not.
For example, ring instances admit uncrossed flows but Figure~\ref{fig:problematic-instances} (left) does not, even though it is obtained by adding leaves to a ring instance.}.
In the figure, no pair of flow paths cross, but the green $s_2 t_2$-paths ``slip between'' the blue $s_1 t_1$-paths to move between two distinct sectors.
We define {\em strongly uncrossed} multiflows by imposing stricter restrictions on the support paths of different commodities at their terminals.
Many layout classes admit these strongly uncrossed flows, including fully-planar, single-source, and rings, as well as pairwise-planar.

\begin{figure}[htbp]
     \centering
     \begin{subfigure}[b]{0.3\textwidth}
         \centering
          \begin{tikzpicture}[scale=0.7]
    \node  (s1) at (0, 0) {};
    \draw (s1) node[below = 2] {$s_1$};
    \node  (s2) at (-2, 2) {};
    \draw (s2) node[above = 3] {$s_2$};
    \node  (t2) at (2, 2) {};
    \draw (t2) node[above = 3] {$t_2$};
    \node  (t1) at (0, 4) {};
    \draw (t1) node[above = 2] {$t_1$};
    
    \node[inner sep=0pt, minimum size=0pt]  (4) at (-3,2) {};
    \node[inner sep=0pt, minimum size=0pt]  (5) at (3,2) {};
    
    \draw[black] (s1) to (t1);
    \draw[line width=3pt, semitransparent, blue]  (s1) to (t1);
    \draw[black]  (s1) to (s2);
    \draw[line width=3pt, semitransparent, ForestGreen]  (s1) to (s2);
    \draw[black]  (t1) to (s2);
    \draw[line width=3pt, semitransparent, ForestGreen]  (t1) to (s2);
    \draw[black]  (s1) to (t2);
    \draw[line width=3pt, semitransparent, ForestGreen]  (s1) to (t2);
    \draw[black]  (t1) to (t2);
    \draw[line width=3pt, semitransparent, ForestGreen]  (t1) to (t2);

    \foreach \v in {s1,s2,t2,t1,4,5} {
        \draw[fill,black] (\v) circle (3pt);
    }
    
    \draw[black] (s1) to[ out=180,in=-90] (4) to[out=90, in=-180] (t1);
    \draw[line width=3pt, semitransparent, blue] (s1) to[ out=180,in=-90] (4) to[out=90, in=-180] (t1);
    \draw[black] (s1) to[out=0,in=-90] (5) to[out=90, in=0] (t1);
    \draw[line width=3pt, semitransparent, blue] (s1) to[out=0,in=-90] (5) to[out=90, in=0] (t1);

    \draw[dashed] (s1) to[bend left] (t1);
    \draw[dashed] (s2) to (t2);
\end{tikzpicture}
     \end{subfigure}
     \hfill 
     \begin{subfigure}[b]{0.3\textwidth}
         \centering
          \begin{tikzpicture}[scale=0.7]
	\node  (s1) at (0, 0) {};
	\node  (s2) at (-2, 2) {};
	\node  (t2) at (2, 2) {};
	\node  (t1) at (0, 4) {};
    \node[inner sep=0pt, minimum size=0pt]  (4) at (-3,2) {};
    \node[inner sep=0pt, minimum size=0pt]  (5) at (3,2) {};

    \fill[nearly transparent, blue] (s1.center) to[out=180,in=-90] (4.center) to[out=90, in=-180] (t1.center) -- (s1.center) -- cycle;

    \fill[nearly transparent, ForestGreen] (s2.center) -- (s1.center) -- (t2.center) -- (t1.center) -- (s2.center) -- cycle;

    \draw (s1) node[below = 2] {$s_1$};
    \draw (s2) node[above = 3] {$s_2$};
    \draw (t2) node[above = 3] {$t_2$};
    \draw (t1) node[above = 2] {$t_1$};
    
	\draw[black] (s1) to (t1);
    \draw[line width=3pt, semitransparent, blue]  (s1) to (t1);
    \draw[black]  (s1) to (s2);
    \draw[line width=3pt, semitransparent, ForestGreen]  (s1) to (s2);
    \draw[black]  (t1) to (s2);
    \draw[line width=3pt, semitransparent, ForestGreen]  (t1) to (s2);
    \draw[black]  (s1) to (t2);
    \draw[line width=3pt, semitransparent, ForestGreen]  (s1) to (t2);
    \draw[black]  (t1) to (t2);
    \draw[line width=3pt, semitransparent, ForestGreen]  (t1) to (t2);

    \foreach \v in {s1,s2,t2,t1,4,5} {
        \draw[fill,black] (\v) circle (3pt);
    }
    
    \draw[black] (s1) to[ out=180,in=-90] (4) to[out=90, in=-180] (t1);
    \draw[line width=3pt, semitransparent, blue] (s1) to[ out=180,in=-90] (4) to[out=90, in=-180] (t1);
    \draw[black] (s1) to[out=0,in=-90] (5) to[out=90, in=0] (t1);
    \draw[line width=3pt, semitransparent, blue] (s1) to[out=0,in=-90] (5) to[out=90, in=0] (t1);

\end{tikzpicture}
     \end{subfigure}
     \hfill 
     \begin{subfigure}[b]{0.3\textwidth}
         \centering
          \begin{tikzpicture}

    \draw[fill,nearly transparent,gray] (0:0) circle[radius=1];
    
    \node (b1) at (0:1) {};
    \node (b2) at (90:1) {};
    \node (b3) at (180:1) {};
    \node (g1) at (45:1) {};
    \node (g2) at (135:1) {};
    \draw[fill=blue] (b1) circle (3pt);
    \draw[fill=blue] (b2) circle (3pt);
    \draw[fill=blue] (b3) circle (3pt);
    \draw[fill=ForestGreen] (g1) circle (3pt);
    \draw[fill=ForestGreen] (g2) circle (3pt);

    \node (b1') at (0:2) {};
    \node (b2') at (90:2) {};
    \node (b3') at (180:2) {};
    \node (g1') at (45:2) {};
    \node (g2') at (135:2) {};

    \foreach \v in {b1,b2,b3} {
        \draw[black] (\v) to (\v');
        \draw[line width=3pt, semitransparent, blue] (\v) to (\v');
    }

    \foreach \v in {g1, g2} {
        \draw[black] (\v) to (\v');
        \draw[line width=3pt, semitransparent, ForestGreen] (\v) to (\v');
    }

    \draw[black] (g1) to[bend left] (g2);
    \draw[line width=3pt, semitransparent, ForestGreen] (g1) to[bend left] (g2);


\end{tikzpicture}
     \end{subfigure}

     \caption{(Left) An uncrossed multiflow. The blue paths are for $(s_1, t_1)$ and the green paths are for $(s_2, t_2)$. No pair of support paths cross. (Center) A sector from the $s_1 t_1$-flow (blue shaded region) that crosses a sector from the $s_2t_2$-flow (green shaded region). (Right) The disk-expansion of $s_1$. The green path separates the blue nodes in the disk. There is a quasicrossing at $s_1$.}
     \label{fig:quasicrossing-super}
\end{figure}

\begin{definition}
    [Quasicrossing, Strongly Uncrossed]
    \label{def:quasicrossing}
    Let $f$ be a multiflow, with a fixed parallelization (so we assume support paths are edge-disjoint).
    Consider a node $z$ and the clockwise ordering of edges $e_1, \dotsc, e_k$ about $z$.
    For a commodity $h$, a {\em type-1 $(h, z)$-interval} is $[i, j]$ (where $i < j$) such that $e_i, e_j$ are consecutive edges of a single support path for $h$ that transits through $z$.
    A {\em type-2 $(h, z)$-interval} is $[i, j]$ ($i < j$) such that $e_i, e_j$ both terminate (different) support paths for $h$ at $z$.
    Note that two paths cross at $z$ if and only if their corresponding type-1 intervals cross\footnote{Two intervals $[i, j], [i', j']$ cross if $i < i' < j < j'$ or $i' < i < j' < j$.}.
    We say a {\em quasicrossing} occurs at $z$ if there exists an $(h, z)$-interval $I$ and an $(h', z)$-interval $I'$ that cross where $h \neq h'$ and at least one of the intervals is type-2.
    If there are no crossings and no quasicrossings at any node, we say $f$ is {\em strongly uncrossed}.

\end{definition}






\resolved{
\guyslain{next paragraph was not clear, please check my changes:}\bruce{leaving it for joe to check}\joe{looks good, thanks!}
}

\resolved{
\bruce{we discuss this in the intro so its probably worth removing here?}
}

\resolved{
\joe{A bit awkward to include this here, but our proof does require strongly uncrossed... should we swap the order?}\bruce{its an option to remove even or is it,used in comgestiom? i think its better to not have it here since it has less impact since it may not apply to the max version.}
\joe{We might need to for NPC to justify why its okay to just think about paths... but then it should be in the appendix}
}





\subsection{Walk-Multiflows}
\label{sec:walks}


It is standard for flows to only route on (simple) paths, since a flow walk could just be short-circuited into a flow path.
Disturbingly, this operation is not benign in the uncrossed flow setting.
The {\em Keyhole Instance} of Figure~\ref{fig:problematic-instances} is an example that requires the use of a non-simple flow trail to obtain an uncrossed solution. 
Since we would like to understand which instances have uncrossed solutions, we must consider {\em trail-multiflows} and {\em walk-multiflows} in addition to the usual {\em path-multiflows}.
A {\em trail} can repeat nodes but not edges, while a {\em walk} can repeat both.
Ultimately we can show our rounding results (\Cref{thm:MEDP}, \Cref{thm:congestion}) and hardness results (\Cref{thm:inapprox}, \Cref{thm:npc}) presented earlier hold for all three models.
The necessary modifications are minor and we defer the details to the full version.

\section{Uncrossed Flows in Pairwise-Planar Instances}
\label{sec:new-classes}

Here we introduce a new class of multiflow layouts and prove they are uncrossable.
We assume a fixed embedding of $G$.
We say $(s, t) \in E(H)$ {\em fits} face $F$ of $G$ if both $s, t$ lie in $F$ (i.e. the demand could be drawn within the face).
A layout $(G, H)$ is {\em facial} if each $(s, t)$ fits at least one face of $G$ (but demands fitting a common face may still cross each other).
\begin{definition}
    We say $(G, H)$ is {\em pairwise-planar} if $(G, H)$ is facial, and moreover for any pair $h_1, h_2 \in E(H)$, $h_1$ and $h_2$ can be embedded inside (possibly different) faces of $G$ so they do not cross each other.
\end{definition}
Unlike fully-planar, we do not require simultaneous uncrossed embedding of all demands.
This matters because a demand edge may fit multiple faces (a {\em multi-faced} demand).
\Cref{fig:pairwise-planar} shows an example.
However we emphasize that the embedding of $G$ is fixed and does not change as we embed different pairs of $E(H)$.

Clearly pairwise-planar generalizes fully-planar, but how substantial is this generalization?
We briefly argue it is a much richer class.
First, it is easy to construct examples where $H$ contains arbitrarily large cliques, meaning $H$ and $G+H$ have arbitrarily large genus (though $G$ remains planar).
It also generalizes rings ($G$ is a cycle with parallel edges, $H$ is arbitrary).
This is worth remarking because fully-planar does not capture ring layouts.
More broadly, pairwise-planar captures what are known as {\em fully-compliant} layouts (under mild connectivity assumptions).
Introduced in \cite{chekuri2013flow}, in such layouts $G$ is series-parallel and every $h \in E(H)$ satisfies that $G+h$ is series-parallel.
In the appendix (see \Cref{sec:fully-compliant}) we argue the following.
\begin{proposition}
    \label{prop:fully-compliant-pairwise-planar}
    If $(G, H)$ is fully-compliant and $G$ is 2-connected,
    then $(G, H)$ is pairwise-planar for any planar embedding of $G$.
\end{proposition}

Pairwise-planar extends further than fully-planar and fully-compliant in another fundamental way through a property known as {\em cut-sufficiency}.
For a multiflow instance $(G, H, u, d)$, we say the {\em cut condition} holds if for every $S \subseteq V$, $u(\delta_G(S)) \ge d(\delta_H(S))$, where $\delta(S)$ denotes the edges with one end in $S$ and one end in $V \setminus S$.
The cut condition is necessary, but not in general sufficient, for the existence of a feasible (fractional) multiflow.
A layout $(G, H)$ is called cut-sufficient if, for all capacities and demands $(u, d)$, the cut condition is sufficient for the existence of a feasible multiflow.
The Max-Flow Min-Cut Theorem can be equivalently stated as: single-source $(G, H)$ are cut-sufficient.
Understanding the relationship between the cut condition and feasibility is a fundamental question in network flows, and cut-sufficient instances tend to admit stronger integrality properties and approximation algorithms.
A result of Seymour \cite{seymour1981matroids} shows that fully-planar $(G, H)$ are cut-sufficient, and Chekuri et al. \cite{chekuri2013flow} show that fully-compliant instances are cut-sufficient.
However, pairwise-planar layouts need not be cut-sufficient.
The quintessential examples of series-parallel layouts that are not cut-sufficient are the {\em odd spindles}\footnote{In a {\em spindle}, $G = K_{2, p}$ where $p \ge 3$ and $H$ consists of a simple $p$-cycle between the nodes in the larger side of the partition and an edge between the two nodes in the other side. A spindle is {\em odd} if $p$ is odd. In some sense, the odd spindles are known to characterize non-cut-sufficiency for series-parallel layouts, see \cite{chakrabarti2012cut}.} \cite{chekuri2010flow, chakrabarti2012cut}, but these are pairwise-planar.

\begin{figure}[htbp]
\centering
\begin{tikzpicture}[scale=0.65]
    \node (s1) at (0, 0) {};
    \node (s2) at (4, 0) {};
    \node (s3) at (8, 0) {};
    \node (t1) at (12, 0) {};
    \node (t2) at (16, 0) {};
    \node (t3) at (20, 0) {};
    
    \node (x1) at (1, 3.7) {};
    \node (x2) at (19, 3.7) {};
    
    \node (y1) at (6, 3.7) {};
    \node (y2) at (10, 3.7) {};
    \node (y3) at (14, 3.7) {};
    \node (z1) at (10, 7.2) {};

    \node[cloud,draw,minimum width = 1.8 cm, minimum height = 1 cm] (c1) at (2, 0) { $G_1$ };
    \node[cloud,draw,minimum width = 1.8 cm, minimum height = 1 cm] (c2) at (6, 0) { $G_2$ };
    \node[cloud,draw,minimum width = 1.8 cm, minimum height = 1 cm] (c3) at (10, 0) { $G_3$ };
    \node[cloud,draw,minimum width = 1.8 cm, minimum height = 1 cm] (c4) at (14, 0) { $G_4$ };
    \node[cloud,draw,minimum width = 1.8 cm, minimum height = 1 cm] (c5) at (18, 0) { $G_5$ };
    
    \node[cloud,draw,minimum width = 1.8 cm, minimum height = 1 cm] (c6) at (4, 3.7) { $G_6$ };
    \node[cloud,draw,minimum width = 1.8 cm, minimum height = 1 cm] (c7) at (8, 3.7) { $G_7$ };
    \node[cloud,draw,minimum width = 1.8 cm, minimum height = 1 cm] (c8) at (12, 3.7) { $G_8$ };
    \node[cloud,draw,minimum width = 1.8 cm, minimum height = 1 cm] (c9) at (16, 3.7) { $G_9$ };

    \foreach \v in {s1,s2,s3,t1,t2,t3,x1,x2,y1,y2,y3, z1} {
        \draw[fill] (\v) circle (3pt);
    }

    \foreach \a / \b in {s1/t1,s2/t2,s3/t3} {
        \draw[dashed] (\a) to[bend left=40] (\b);
    }

    \foreach \a / \b in {s1/c1,c1/s2,s2/c2,c2/s3,s3/c3,c3/t1,t1/c4,c4/t2,t2/c5,c5/t3,x1/c6,c6/y1,y1/c7,c7/y2,y2/c8,c8/y3,y3/c9,c9/x2} {
        \draw[black] (\a) to (\b);
    }

    \draw[black] (s1) to (x1);
    \draw[black] (x2) to (t3);
    \draw[black] (x1) to[bend left] (z1);
    \draw[black] (z1) to[bend left] (x2);

    \draw[dashed] (x1) to[bend left=40] (y3);
    \draw[dashed] (y1) to[bend left=40] (x2);
    \draw[dashed] (y2) to (z1);
    \draw[dashed] (y2) to (s3);
    \draw[dashed] (y2) to (t1);
\end{tikzpicture}
\caption{A pairwise-planar layout. Each $G_i$ is a connected planar graph. The demand edges cannot be simultaneously embedded so that $G+H$ is planar ($G+H$ contains a $K_{3, 3}$-minor and a $K_5$-minor). However, taking advantage of multi-faced demands, any {\em single pair} of demands can be embedded without crossing. One could make the layout more complicated by putting demands inside the $G_i$'s following similar patterns.}
\label{fig:pairwise-planar}
\end{figure}

\subsection{Uncrossing Procedure}

In this section we show why feasible pairwise-planar instances have uncrossed flows.
We focus on congestion instances $(G, H, u, d)$ (where a feasible flow must satisfy all demands), since given a feasible flow $x$ for a maximization instance $(G, H, u)$, one can define demand values $d(s, t) = |x^{st}|$ and apply the congestion result.
We present an inductive proof that a feasible uncrossed flow exists.
It is straightforward to convert this into an algorithm, with one catch for polynomial time implementation that we explain at the end.

For fully-planar instances, \cite{garg2022integer,espinosa2026unsplittable} use an algorithm by Matsumoto et al. \cite{matsumoto1986planar} that builds an uncrossed flow from scratch starting from the zero flow.
Our procedure has more in common with \cite{huang2023approximating} (which concerns finding cycles in higher-genus surfaces that cross at most once) in that we begin with a (potentially highly crossing) feasible flow $x$ and ``uncross'' it bit-by-bit.
The following is an important lemma.

\begin{lemma}
    \label{corollary:cross-once}
    Let $(G, H, u, d)$ be a planar congestion instance that is feasible.
    There exists a feasible flow $f^*$ such that every pair of support paths crosses at most once, and any pair of support paths that share a terminal do not cross.
\end{lemma}

We say a flow satisfying this property is {\em almost-uncrossed}.
Proving this lemma formally is somewhat tedious, but such arguments have appeared in the literature before - the essence of the argument is in Lemma 1 of \cite{naves_hardness_2012} and a similar method is used in \cite{huang2023approximating}.
Loosely, the argument goes as follows.
Begin with a feasible flow $x$.
One can show that if a pair of support paths cross twice (or cross once but share a terminal), then one can re-route flow on their edges to decrease the ``amount of crossing'' between the two paths (``uncrossing'' the two paths) and argue that this re-routing does not increase crossing with other support paths.
Thus the total amount of crossing decreases, and we iterate until $x$ is uncrossed.
For completeness we give a careful proof of this lemma (and precisely define ``amount of crossing'') in the appendix (see \Cref{sec:local-uncrossing}).

\Cref{corollary:cross-once} immediately implies single-source layouts are uncrossable.
With a small trick it also proves fully-planar layouts are uncrossable: augment the instance by adding leaves. That is, for each $h = (s, t) \in E(H)$, add nodes $s_h, t_h$ into the interior of the face where $h$ is embedded, add supply edges $(s_h, s), (t_h, t)$, move the endpoints of the demand $h$ to $(s_h, t_h)$, and extend the flow paths appropriately to route this instance.
At first, adding leaves may actually increase the amount of crossing.
However, since the demands do not cross, it can be shown that every pair of support paths that cross must cross at least twice.
Then \Cref{corollary:cross-once} allows one to conclude there exists a feasible uncrossed flow in the  augmented instance.
The leaves can then be contracted away without introducing any new crossing.

Extending this trick to pairwise-planar instances is not immediate.
After adding leaves, it may not be the case that support paths that cross must cross at least twice.
We cannot ``commit'' a demand to a particular face up front.
Consider if a demand $h$ fits into two different faces $F_1$ and $F_2$, but there are demands $h_1$ and $h_2$ that fit those respective faces and cross $h$ in those faces.
If we add leaves into $F_1$, then we can ``uncross'' support paths for $h$ and $h_2$, however we cannot uncross support paths between $h$ and $h_1$, and vice versa.

We take the following approach.
Begin with a feasible flow $z$ that we want to uncross.
We temporarily ``commit'' each demand to one of the faces it fits (if multi-faced, choose arbitrarily).
We accept for now that flow paths of demands committed to the same face may cross, but compute a flow $x$ so that demands committed to different faces do not cross.
In fact we aim for a more strict condition on $x$: we show how to find $x$ so that, if we define $x^F$ as the multiflow corresponding to demands committed to face $F$, the $x^F$'s are  ``uncoupled'' from each other.
Here, ``uncoupled'' is a structural property (defined precisely below) which guarantees that if we can uncross each $x^F$ separately, regardless of how we do so, we can still safely glue the resulting flows back together without introducing crossing. 
Then we can just focus on uncrossing each $x^F$ inductively/recursively.
We now make this ``committing'' process and ``uncoupled'' property precise.

\begin{definition}
    [Facial Decomposition]
    Let $(G, H, u, d)$ be a facial multiflow instance.
    Let $\mathcal F$ be the set of faces of $G$.
    A {\em facial decomposition} is a collection $(H_F : F \in \mathcal F)$ of subgraphs of $H$ indexed by $F$ such that:
    \begin{itemize}
        \item $(E(H_F) : F \in \mathcal F)$ forms a partition of $E(H)$, and
        \item for every face $F \in \mathcal F$, every demand $h \in E(H_F)$ fits $F$ (but cannot necessarily all be embedded simultaneously in $F$ without crossing).
    \end{itemize}
\end{definition}

For a facial decomposition $(H_F : F \in \mathcal F)$ and a feasible flow $x$, we let $x^F$ be the portion of $x$ used to route $H_F$.
Note that if $h \in E(H)$ is multi-faced, the facial decomposition still only puts $h$ into one of the $H_F$'s.
Now we define the ``uncoupled'' condition and show that we can find a flow satisfying this property.
We let $\support(x^F)$ to denote the set of support paths for the flow $x^F$, and let $G[\support(x^F)]$ denote the graph induced by the edges in those support paths.

\begin{definition}
    [Nested Supports Property]
    Let $(G, H, u, d)$ be a facial multiflow instance, and let $(H_F: F \in \mathcal F)$ be a facial decomposition.
    For a flow $x$ and a face $F$, we define the {\em stellated support graph} $G[x, F]$ as follows: begin with the support graph $G[\support(x^F)]$, add a node $v_F$ into the interior of $F$, and then for each demand $(s, t) \in E(H_F)$ add edges $(s, v_F)$ and $(t_, v_F)$, embedded (without crossing) in the face $F$.
    We say $x$ satisfies the {\em nested supports property} if for every pair of distinct faces $F_1, F_2 \in \mathcal F$, the graph $G[x,F_1]$ is entirely contained within a single face of the graph $G[x,F_2]$, and vice versa.
\end{definition}

\begin{lemma}
    \label{lemma:nested-support-lemma}
    Let $(G, H, u, d)$ be a facial multiflow instance that is feasible, and let $(H_F : F \in \mathcal F)$ be a facial decomposition.
    There exists a feasible flow $x$ satisfying the nested supports property.
\end{lemma}

The proof of this lemma is slightly tricky to make rigorous, so we include only a sketch here. The full proof appears in the appendix (see \Cref{sec:nested-supports}).
\begin{proof}[Proof Sketch of \Cref{lemma:nested-support-lemma}]
    We augment the instance by adding leaves into the faces: for every face $F$ and demand $h = (s, t) \in E(H_F)$, add nodes $s_h, t_h$ into the face $F$ and supply edges $(s_h, s), (t_h, t)$ of capacity equal to $d(h)$, and move the demand $h$ from $(s, t)$ to $(s_h, t_h)$.
    We keep the same facial decomposition.
    We now invoke \Cref{corollary:cross-once} and obtain a feasible flow $x$ that is almost-uncrossed for this new instance.
    It is not hard to show that for support paths $P_1, P_2$ associated with distinct faces, if we consider the embedded cycles $C(P_1), C(P_2)$ obtained by joining the endpoints of the respective paths with their respective demand edges, that these cycles do not cross (though if the paths are associated with the same face, they may cross).
    
    Now, for each face $F$, we let $G_F$ be the union of $G[\support(x^F)]$ and all supply edges incident with the face $F$ itself (regardless of whether they are in the support).
    It is possible to argue that for distinct faces $F, F'$, $G_F$ must be contained in a single face of $G_{F'}$, and vice versa.
    Once this property is established, we can then argue that the nested supports property holds for the flow $x$.
\end{proof}
Assuming such a flow exists, we show that we can focus on uncrossing each face's flow independently.

\begin{lemma}
    \label{lemma:face-independence}
    Let $(G, H, u, d)$ be a facial multiflow instance.
    Let $(H_F : F \in \mathcal F)$ be a facial decomposition, and let $x$ be a feasible flow satisfying the nested supports property.
    For every $F \in \mathcal F$, let $u_F$ be the capacity vector defined by $u_F(e) = x^{F}(e)$ and let $d_F$ be the demand vector $d$ restricted to $H_F$.
    If for every $F \in \mathcal F$, the instance $(G, H_F, u_F, d_F)$ has a feasible uncrossed solution $y^F$,
    then $y = \sum_{F \in \mathcal F} y^F$ is a feasible uncrossed solution for $(G, H, u, d)$.
\end{lemma}

\begin{proof}
    The feasibility of $y$ is clear from the construction of the subproblems.
    Also observe that, by construction, $y$ still satisfies the nested supports property (since $G[y,F]$ is a subgraph of $G[x, F]$ for each face $F$ of $G$).
    
    We prove that $y$ is uncrossed assuming each $y^F$ is uncrossed.
    It remains to check for distinct faces $F_1, F_2$ and $P_1 \in \support(y^{F_1}), P_2 \in \support(y^{F_2})$, that $P_1$ and $P_2$ do not cross.
    
    Consider two distinct faces $F_1, F_2$.
    Let $P_1 \in \support(y^{F_1}), P_2 \in \support(y^{F_2})$.
    Let $h_1 = (s_1, t_1), h_2 = (s_2, t_2)$ be the demands routed by $P_1, P_2$ respectively.
    Suppose $P_1$ and $P_2$ cross, for the sake of contradiction.
    Consider the (simple) cycle $C$ in $G[y, F]$ obtained as $C = P_1 \bullet (s_1, v_{F_1}) \bullet (v_{F_1}, t_1)$, where $\bullet $ denotes concatenation.
    Since $P_1$ and $P_2$ cross, $P_2$ has some segment in the open exterior of $C$ and some segment in the open interior of $C$.
    However, since $C$ is a simple cycle, these two segments of $P_2$ must lie in distinct faces of $G[y, F]$, contradicting that $y$ satisfies the nested supports property.
\end{proof}

Using \Cref{lemma:face-independence}, we are able to reduce to finding a feasible uncrossed flow for each instance $(G, H_F, u_F, d_F)$ in isolation. No matter how we do this, we can combine the flows with \Cref{lemma:face-independence}. Importantly, this means that to tackle the instance $(G, H_F, u_F, d_F)$, we are able to re-commit multi-faced demands that were originally committed to $F$ to other faces.
Now, we give the main result.

\begin{theorem}
    \label{thm:pairwise-planar}
    Let $(G, H)$ be pairwise-planar and let $(G, H, u, d)$ be a congestion instance.
    If the instance is feasible, there exists a feasible uncrossed flow.
\end{theorem}
\begin{proof}
    We prove the result by induction on $|E(H)|$.
    If $|E(H)| = 1$ then the instance is single-source and the result follows, so suppose $|E(H)| \ge 2$.
    Let $\mathcal H = (H_F : F \in \mathcal F)$ be an arbitrary facial decomposition.
    
    We first consider the case that at least two faces $F$ have that $H_F$ is non-empty.
    We obtain a feasible flow $x$ for $(G, H, u, d)$ that satisfies the nested supports property for $\mathcal H$ by way of \Cref{lemma:nested-support-lemma}.
    For each face $F$, let $u_{F}, d_{F}$ be defined as in the statement of \Cref{lemma:face-independence}.
    Each instance $(G, H_{F}, u_{F}, d_{F})$ is feasible via the flow $x^F$.
    Additionally, since there are at least two faces $F$ such that $H_F$ is non-empty, we have $|E(H_F)| < |E(H)|$ for every face $F$.
    Hence by induction, there is a feasible uncrossed flow $y^F$ for each instance $(G, H_F, u_F, d_F)$.
    Then by \Cref{lemma:face-independence}, the sum of these flows is a feasible uncrossed flow for $(G, H, u, d)$.

    Now, we consider the case where $H_F$ is empty for all faces $F$ except one, call it $F^*$.
    That is, $E(H) = E(H_{F^*})$, and in particular every demand fits $F^*$.
    If none of the demands cross within $F^*$, then the layout is fully-planar and we are done.
    Otherwise, there exist $h, h' \in E(H)$ such that $h, h'$ cross when drawn in the face $F^*$.
    However, since $(G, H)$ is pairwise-planar, there exists another face $F' \neq F^*$ such that at least one of these two demands, without loss of generality say $h'$, also fits $F'$.
    Now we define a facial decomposition $\widehat {\mathcal H} = (\widehat H_F : F \in \mathcal F)$ as follows.
    We define $E(\widehat H_{F^*}) = E(H) \setminus \set{h'}$, $E(\widehat H_{F'}) = \set{h'}$, and $E(\widehat H_F) = \emptyset$ for all other faces $F$.
    That is, we now view $h'$ as being associated with $F'$ rather than $F^*$.
    
    We obtain a feasible flow $x$ for the instance $(G, H, u, d)$ satisfying the nested supports property for this new facial decomposition $\widehat {\mathcal H}$ by way of \Cref{lemma:nested-support-lemma}.
    We define capacities $u_F, d_F$ for each face $F$ as in the statement of \Cref{lemma:face-independence} using the new decomposition $\widehat{\mathcal H}$.
    Each instance $(G, \widehat H_F, u_F, d_F)$ is feasible.
    Additionally, each satisfies that $|E(\widehat H_F)| < |E(H)|$, and thus $(G, \widehat H_F, u_F, d_F)$ has a feasible uncrossed flow $y^F$ by induction.
    Then the sum of these flows is a feasible uncrossed flow for $(G, H, u, d)$ by \Cref{lemma:face-independence}.
\end{proof}

With a little extra effort, we can modify \Cref{lemma:face-independence} and \Cref{thm:pairwise-planar} to show that pairwise-planar instances admit strongly uncrossed flows.

It is straightforward to convert this inductive argument into a polynomial time algorithm, provided we have a polynomial time algorithm to compute a feasible almost-uncrossed flow, as in \Cref{corollary:cross-once}.
However, this is somewhat tricky.
To prove \Cref{corollary:cross-once}, we re-route flows between twice-crossing paths and reduce the total amount of crossing, but the amount of progress depends on the flow values of the crossing paths and could be very small (again, see \Cref{sec:local-uncrossing} for more details).
We defer the details to the full version, but one can show how to efficiently compute an almost-uncrossed flow that is {\em approximately} feasible - that is, for any $\epsilon > 0$, in polynomial time we can find an almost-uncrossed flow that is feasible for $(G, H, (1 + \epsilon)u, d)$.
The details are similar to the uncrossing algorithm in \cite{huang2023approximating}.
Essentially, by rounding the flow on each path down slightly and removing support paths with very little flow, we can ensure that the remaining flow values are all common multiples of a small but sufficiently large fraction, which ensures the uncrossing steps make enough progress to terminate in polynomial time.
We summarize the final result below.
\begin{theorem}
    Let $(G, H)$ be pairwise-planar and let $(G, H, u, d)$ be a congestion instance.
    Let $\epsilon > 0$ be any constant.
    If the instance is feasible, then there exists a strongly uncrossed flow $f'$ such that $f'$ is feasible for $(G, H, (1+\epsilon)u, d)$.
    The flow $f'$ can be computed in polynomial time.
\end{theorem}

Note that $f'/(1+\epsilon)$ is feasible for $(G, H, u, d/(1+\epsilon))$, so this approximate guarantee is useful for maximziation as well as congestion.
Combining this theorem with our rounding results \Cref{thm:MEDP} and \Cref{thm:congestion}, we obtain \Cref{cor:pairwise-planar-algorithm}.

\subsection{Integral Flow-Multicut Gap}

In this subsection we prove \Cref{cor:pairwise-planar-multicut-gap}.
Here we are concerned with maximization instances $(G, H, u)$.
A {\em multicut} is a set of supply edges $R \subseteq E(G)$ such that $G - R$ does not contain any $st$-path for any $(s, t) \in E(H)$.
The minimum capacity of a multicut is at least the maximum value of a (fractional) multiflow.
The ratio of these quantities is the {\em (fractional) flow-multicut gap}\footnote{Note: this is different from the similarly named {\em flow-cut gap}, which is defined for congestion instances and relates the minimum congestion value to the sparsity of cuts (not necessarily multicuts). See \cite{chekuri2013flow} for instance.}
The {\em integral flow-multicut gap} is the ratio of the minimum capacity of a multicut to the maximum value of an integral multiflow.

An immediate corollary of \Cref{thm:pairwise-planar} and \Cref{thm:MEDP} is that for pairwise planar $(G, H)$, the integral flow-multicut gap is bounded by $O(1)$ times the (fractional) flow-multicut gap.
We combine this with the following result by Tardos and Vazrani \cite{tardos_improved_1993} to conclude \Cref{cor:pairwise-planar-multicut-gap}.
\begin{theorem}[Tardos and Vazrani \cite{tardos_improved_1993}]
    For any maximization instance $(G, H, u)$ where $G$ is planar, the (fractional) flow-multicut gap is bounded by $O(1)$.
\end{theorem}


\section{Maximization Setting: Rounding and Hardness}
\label{sec:maximization-all}

In this section, we prove \Cref{thm:MEDP} and \Cref{thm:inapprox}, which cover our rounding and inapproximability results for maximum uncrossed multiflow.

\subsection{Uncrossed String Graphs and Approximate Integer Decomposition}

First we prove \Cref{thm:MEDP}.
Our approach is to establish a so-called approximate integer decomposition property \cite{chekuri2007multicommodity,chekuri2009approximate} for uncrossed multiflows.
The outline of the strategy is as follows.
Suppose we are given an uncrossed multiflow $f$, which achieves a value of $\beta$.
Imagine scaling $f$ by a positive integer $k$ such that $kf$ is integral.
Of course, $kf$ may not be feasible.
However suppose we are able to show that $kf$ can be decomposed as the sum of a small number $b$ of feasible integral flows.
That is $kf = \sum_{i = 1}^b f^i$ where each $f^i$ is integral and feasible.
At least one of the $f^i$ flows achieves a $1/b$ fraction of the value of $kf$, and hence a $k/b$ fraction of the value of $f$.

How can we achieve such a decomposition?
Suppose (for simplicity) that $kf$ sends 1 unit of flow along paths $\mathcal P = \set{P_1, \dotsc, P_\ell}$.
Furthermore suppose that each edge of $G$ has unit capacity and hence each appears on at most $k$ paths.
Our goal then is to partition the paths into classes such that every edge appears at most once in each class; the classes then define the flows $f^i$.
Equivalently, we seek to partition the paths so that paths in the same class do not share an edge.
Phrased this way, we can recognize this as a graph colouring problem on the edge-intersection graph of $\mathcal P$.

The edge-intersection graph $U$ of $(G, \mathcal P)$ has node set $\mathcal P$, and an edge between distinct $P_1, P_2 \in \mathcal P$ if there exists $e \in E(P_1) \cap E(P_2)$.
The {\em load} of $e \in E(G)$ is the number of paths in which $e$ occurs, and the {\em load} of $(G, \mathcal P)$ is the maximum load of an edge $e \in E(G)$.
Graphs $U$ constructed in this way are a special type of the well-studied family of {\em string graphs}, which are the intersection graphs of curves in the plane.

\begin{definition}
[Uncrossed String Graph, Realization]
An {\em uncrossed string graph} $U$ is the edge-intersection graph of a family of uncrossed paths $\mathcal P$ in a planar graph $G$.
We say $(G, \mathcal P)$ is a {\em realization} of $U$.

\end{definition}

Uncrossed string graphs have been studied before for their own sake \cite{esperet_coloring_2009,van_batenburg_coloring_2017}.
Given that $U$ has a load $k$ realization, we want that its chromatic number $\chi(U)$ is some small function of $k$, ideally $O(k)$.
This question is strongly related to the concept of a {\em $\chi$-bounded} graph family.
A family of graphs $\mathcal G$ is $\chi$-bounded if there exists a function $f$ such that for all $U \in \mathcal G$, $\chi(U) \le f(\omega(U))$, where $\omega(U)$ is the clique number of $U$.
It is worth noting that general string graphs are not $\chi$-bounded, and it is not obvious that uncrossed string graphs should be.
An $O(k)$ bound was conjectured originally in \cite{esperet_coloring_2009}, and eventually proven by Fox and Pach \cite{fox_touching_nodate}.

\begin{theorem}
    [Fox and Pach]
    \label{thm:fox-pach}
    Let $U$ be an uncrossed string graph with realization $(G, \mathcal P)$ such that each edge $e \in E(G)$ has load at most $k$.
    Then $\chi(U) \le 6ek + 1$.
\end{theorem}

The result of Fox and Pach has not been published, however the proof is elegant and short.
It establishes via a probabilistic proof (similar to a standard proof giving a bound on the crossing number of a graph) that $U$ has at most $3ek|\mathcal P|$ edges.
Hence a greedy colouring procedure that removes the minimum degree vertex (which has degree at most $6ek$) and colours the rest of the graph inductively obtains a $6ek+1$ colouring.
We were not initially aware of Fox and Pach's result and independently proved a weaker bound that $\chi(U) \le O(k^2 \log |\mathcal P|)$.
For the sake of completeness we present this proof in \Cref{sec:colouring-bound}.

We now use the colouring result to establish an approximate integer decomposition property and prove the following theorem, which is a generalization of \Cref{thm:MEDP} with weighted profits on $E(H)$.

\begin{theorem}
    \label{thm:max-rounding-full}
    There exists a universal constant $\gamma \ge 1$ such that the following holds.
    Let $(G, H, u)$ be a planar maximization instance where $u$ is integral, equipped with per-unit profits $w(h) \ge 0$ for $h \in E(H)$.
    Let $f$ be a feasible uncrossed multiflow.
    There exists a feasible integral multiflow $f^*$ that is uncrossed
    and achieves value $|f^*|_w \ge \frac 1 \gamma |f|_w$.
    This solution can be computed in polynomial time.
\end{theorem}

\begin{proof}
    We first present the case where the capacities are unit-valued $u = \vec 1$.
    We show how to handle general capacities at the end of the proof.
    
    Let $\beta = |f|_w$.
    Let $\mathcal P$ be the family of support paths for $f$.
    Since $u = \vec 1$, we have that $f(P) < 1$ for each $P \in \mathcal P$.
    Now, let $k$ be a positive integer such that $kf$ is integral.
    Since $kf$ is integral, for each $P$ we have that $f(P) = \alpha_P / k$ where $\alpha_P \le k$ is a positive integer.
    By treating each $P$ as $\alpha_P$ distinct paths, we may assume that each path routes exactly $\frac 1 k$ units of flow.
    Then, since $\sum_{P : e \in P} f(P) \le 1$ by feasibility, at most $k$ paths use each edge of $G$.
    
    Let $U$ be the edge-intersection graph of $(G, \mathcal P)$, which has load at most $k$.
    By \Cref{thm:fox-pach}, $U$ is $b$-colourable where $b = O(k)$.
    Now for each colour class $i$, we let $f_i$ denote the flow obtained by sending one unit of flow on each path in that colour class, and zero flow on all other paths.
    Then $f_i$ is integral and satisfies that $\sum_{P : e \in P} f_i(P) \le 1$ for each $e \in E(G)$, and $kf = \sum_{i = 1}^b f_i$ since each path's flow was blown up by a factor of $k$.
    Since the value of $kf$ is $k\beta$, at least one of the $f_i$ flows has a value of $k \beta / b$.
    Since $b = O(k)$, then this flow has value $\Omega(1)$ times that of $f$, as desired.

    However, this proof is existential thus far, and does not immediately give a polynomial time algorithm because $k$ could be very large.
    To fix this, we first have a pre-processing step where we convert $f$ into a feasible flow $\hat f$ that uses a subset of $f$'s support paths (so that it remains uncrossed) such that each value $\hat f(P)$ is a multiple of $\frac 1 {\poly(n)}$, and so that $|\hat f|_w$ is a large fraction of $|f|_w$.
    We then run the above procedure on $\hat f$ instead of $f$, giving a polynomial time algorithm.
    We prove the following.
    \begin{lemma}
    \label{lemma:randomized-rounding-log}
    Let $\epsilon \in (0, 1)$ be any constant.
    There exists a positive constant $N_{\epsilon}$ such that the following holds.
    Let $(G, H, u = \vec 1)$ be a unit-capacity planar maximization instance on $n \ge N_{\epsilon}$ nodes, equipped with per-unit profits $w(h) \ge 0$ for $h \in E(H)$.
    Let $f$ be a feasible multiflow such that $|f|_w \ge w_{\max} = \max_{h \in E(H)} w(h)$.
    There exists a multiflow $\hat f$ satisfying:
    \begin{enumerate}
        \item $\hat f$ is feasible,
        \item the support paths of $\hat f$ are a subset of the support paths of $f$,
        \item $|\hat f|_w \ge (1 - \epsilon) |f|_w$, and
        \item for every path $P$, $\hat f(P)$ is a multiple of $\frac 1 {\log^2 n}$.
    \end{enumerate}
    There is a polynomial time algorithm that computes such an $\hat f$.
    \end{lemma}
    We prove this in \Cref{sec:polytime-random-rounding} using a randomized rounding and re-scaling argument (which can be de-randomized to give a deterministic algorithm).
    
    The two conditions that $|f|_w \ge w_{\max}$ and $n \ge N_{\epsilon}$ are only minor technicalities.
    If $|f|_w < w_{\max}$, then we need not bother with rounding $f$: simply take $f^*$ as routing one unit of flow on any $st$-path where $w(s, t) = w_{\max}$.
    The $n \ge N_{\epsilon}$ condition is also not a huge obstacle.
    If $n < N_{\epsilon}$ then $n$ is bounded by a constant, and we can obtain $\hat f$ whose fractionality is constant as follows.
    We set up a linear program formulation of the $w$-weighted maximum multiflow problem using path-variables $x(P)$, but only using paths $P$ where $P$ is a support path of $f$. 
    Standard linear programming theory implies there is a feasible solution $\hat f$ whose value is at least that of $f$, is uncrossed (since its support paths are a subset of $f$'s), and whose fractionality is bounded as a function of $n \le N_{\epsilon}$, but this is still constant.

    This concludes the proof of the unit capacity case.
    For general capacities, we reduce to the unit capacity case.
    Ultimately we want to split each edge into parallel unit capacity copies, but need to be careful because this is not a polynomial time operation in general.
    A key fact we use is the following.
    \begin{lemma}
    \label{lemma:polybounded-support}
        Let $f$ be an uncrossed flow in a planar multiflow layout $(G, H)$.
        Then $f$ has at most  $O(|\mathcal F(G)||E(H)|)$ support paths, where $\mathcal F(G)$ is the set of faces of $G$.
    \end{lemma}
    \begin{proof}
        It suffices to show that for each $h = (s, t) \in E(H)$, $x^h$ has only $O(|\mathcal F(G)|)$ support paths.
        Order the support paths of $x^h$ clockwise about $s$ using the uncrossed parallelization, say $P_1, \dotsc, P_\ell$.
        Observe that each face $F$ appears in the region bounded between a unique consecutive pair of paths $P_i, P_{i + 1}$ in this clockwise order.
        Moreover the region between each consecutive pair of paths $P_i, P_{i  + 1}$ contains at least one face.
        The result follows.
    \end{proof}
    
    The reduction to unit capacities now proceeds as follows.
    Define an integral flow $f_{\text{int}}$ by $f_{\text{int}}(P) = \lfloor f(P)\rfloor$ for all $P$.
    Let $f_{\text{frac}} = f - f_{\text{int}}$, which satisfies that $f_{\text{frac}} (P) < 1$ for all $P$, and define capacities $u'(e) = u(e) - f_{\text{int}}(e)$.
    Observe that $f_{\text{frac}}$ is feasible for capacities $u'$.
    Let $K$ be the number of support paths of $f_{\text{frac}}$, which by \Cref{lemma:polybounded-support} is polynomially bounded in $|E(G)| + |E(H)|$ (recall that the number of faces of a planar graph is linear in the number of edges).
    Now, since $f_{\text{frac}}(P) < 1$ for all $P$, it follows that $|f_{\text{frac}}| = \sum_Pf_{\text{frac}}(P) < K$.
    Therefore $f_{\text{frac}}$ is still feasible if we bound the capacities to $u''(e) = \min \set{u'(e), K}$, but these capacities are polynomially bounded in the size of $G$ and $H$.
    Then splitting each edge $e$ into $u''(e)$ many unit-capacity copies, we apply the algorithm for the unit capacity case to $f_{\text{frac}}$ and obtain an integral multiflow $f^*$ that is feasible for $(G, H, u'')$ and has value $|f^*|_w \ge \frac 1 \gamma |f_{\text{frac}}|_w$.
    The final flow we output is $f_{\text{int}} + f^*$, which is feasible for $(G, H, u)$ and has value $|f_{\text{int}} + f^*|_w \ge |f_{\text{int}}|_w + \frac 1 \gamma |f_{\text{frac}}|_w \ge \frac 1 \gamma |f|_w$.
\end{proof}

It is straightforward to generalize this to the node-capacitated setting - the colouring and integer decomposition property proceed in exactly the same way.

\subsection{Inapproximability of Maximum Uncrossed Fractional Multiflow}
\label{sec:inapprox}

Next we prove \Cref{thm:inapprox} on the inapproximability of maximum uncrossed (fractional) multiflow.
We make use of a strong inapproximability result given by Chuzhoy, Kim, and Nimavat for maximum edge-disjoint paths (MEDP) and maximum node-disjoint paths (MNDP) in planar graphs; i.e. maximum integer multiflow for unit capacities.
One minor technicality is that in their formulation, a feasible solution is only allowed to route one unit of flow for each demand.
We thus call these problems {\em capped MEDP} and {\em capped MNDP} respectively.

\begin{theorem}[Theorems 1.1 and 1.2, \cite{chuzhoyjournalv017a006}]
\label{thm:hardnesschuzhoy}
For every constant $\epsilon > 0$ , there is no $2^{\Omega(\log^{1-\epsilon} n)}$-approximation 
polynomial-time algorithm for capped MEDP or capped MNDP, assuming that $NP \not\subseteq DTIME(n^{polylog~n})$.
Moreover, there is no $n^{O\left(\frac{1}{\left(\log \log n\right)^2}\right)}$-approximation 
polynomial-time algorithm for capped MEDP or capped MNDP, assuming that for some constant $\delta > 0$,
$NP \not\subseteq DTIME(2^{n^{\delta}})$.
These results hold even when the input graph is a maximum degree $3$ planar graph, or even if $G$ is a wall graph.
\end{theorem}


We now prove Theorem~\ref{thm:inapprox} which relies on this result.
We focus on the edge-capacitated version (the node-capacitated version is similar).

\begin{proof}[Proof of \Cref{thm:inapprox}]
    We prove that result by giving a polytime reduction from capped MEDP in wall graphs to maximum uncrossed multiflow in wall graphs, that preserves the approximation ratio up to a constant factor.
    For an instance $(G, H)$ of the capped MEDP problem where $G$ is a wall graph (which is planar and has maximum degree 3), define the following quantities.
    \begin{itemize}
        \item Let $D_E(G, H)$ denote the maximum number of demands in $H$ which can be routed on edge-disjoint paths in $G$ (i.e. the optimal value for capped MEDP).
        \item Let $U_E(G, H)$ denote the maximum value of a (fractional) uncrossed multiflow for $(G, H, u = \vec 1)$.
    \end{itemize}
    
    
    Observe that any capped MEDP solution for $(G, H)$ is a feasible (integral) uncrossed multiflow for $(G, H, u = \vec 1)$.
    To see this, consider a collection of edge-disjoint paths that form a capped MEDP solution.
    Recall from \Cref{sec:uncrossed} that two edge-disjoint paths cross only if they cross at a node (rather than a longer common subpath).
    Since $G$ has maximum degree 3, a node crossing is not possible.
    Thus we have $D_E(G, H) \le U_E(G, H)$.
    
    Suppose we have a $\beta$-approximation algorithm for maximum uncrossed (fractional) multiflow in wall graphs.
    We obtain an approximation algorithm for capped MEDP as follows.
    We run our $\beta$-approximation algorithm for maximum uncrossed multiflow on $(G, H, u = \vec 1)$ to obtain flow $f$ whose value is at least $U_E(G, H) / \beta$.
    We apply the polytime algorithm from \Cref{thm:max-rounding-full} to compute a feasible uncrossed integral multiflow $f'$ whose value is $1/ \gamma$ that of $f$.
    It may still not be a valid capped MEDP solution, since it may route more than one flow path for the same commodity.
    However, since $G$ has maximum degree 3, it routes on at most 3 paths for each commodity, so by selecting exactly one path for each commodity we obtain a capped MEDP solution whose value is at least $\frac 1 3$ that of $f'$.
    Thus this algorithm outputs a capped MEDP solution whose value is at least $\frac {U_E(G, H)} {3 \gamma \beta} \ge \frac {D_E(G, H)} {3 \gamma \beta}$, and is thus a $3 \gamma \beta$-approximation algorithm.

\end{proof}

\section{Congestion Setting: Rounding}
\label{sec:congestion}

In this section, we focus on planar congestion instances $(G, H, u, d)$ where the demands must be fully routed.
We show Theorem \ref{thm:congestion}, which states that we can convert (fractional) strongly uncrossed multiflows into integral flows that violate the capacities by a factor of at most 2, and if we accept an additional additive error of less than $d_{\max} = \max_{h \in E(H)} d(h)$ the flow can be made unsplittable.
\Cref{sec:sectors} presents the relevant properties of strongly uncrossed flows, and \Cref{sec:congestion-main} gives the proof of \Cref{thm:congestion}.

\subsection{Clockwise Orderings and Sectors}
\label{sec:sectors}
Let $f$ be a strongly uncrossed multiflow for planar congestion instance $(G, H, u, d)$.
Let $\mathcal P$ be the set of support paths, and let $\mathcal P^h$ be the set of support paths routing demand $h \in E(H)$.
Fix a strongly uncrossed parallelization $\phi$.
We parallelize the edges and perform disk-expansion at each node (see \Cref{sec:uncrossed}) so that the paths become node-disjoint.
Then, for each commodity $h = (s, t) \in E(H)$, we add a node and edges into the disk corresponding to $s$ to join the endpoints of the paths for $h$, and similarly for the disk of $t$.
This can be done in a planar fashion for all demands simultaneously since the flow is strongly uncrossed (see \Cref{sec:strongly-uncrossed} and \Cref{fig:strongly_uncrossed_star}).
We denote the new graph $G_\phi$.
For a path $P \in \mathcal P$, we let $\widehat P$ be its corresponding path in $G_\phi$.
For $P, Q \in \mathcal P^h$, $\widehat P$ and $\widehat Q$ are internally node-disjoint paths that share their endpoints.
For $P, Q \in \mathcal P$ routing different demands, $\widehat P$ and $\widehat Q$ are (completely) node-disjoint.


\begin{figure}[htbp]
     \centering
     \begin{subfigure}[b]{0.45\textwidth}
         \centering
         \begin{tikzpicture}

    \draw[fill,nearly transparent,gray] (0:0) circle[radius=1];
    
    \node (b1) at (51*6:1) {};
    \node (b2) at (51*2:1) {};
    \node (b3) at (51*5:1) {};
    \node (g1) at (51*3:1) {};
    \node (g2) at (51*4:1) {};
    \node (o1) at (51*0:1) {};
    \node (o2) at (51*1:1) {};
    \draw[fill=blue] (b1) circle (3pt);
    \draw[fill=blue] (b2) circle (3pt);
    \draw[fill=blue] (b3) circle (3pt);
    \draw[fill=ForestGreen] (g1) circle (3pt);
    \draw[fill=ForestGreen] (g2) circle (3pt);
    \draw[fill=red] (o1) circle (3pt);
    \draw[fill=red] (o2) circle (3pt);

    \node (b1') at (51*6:2) {};
    \node (b2') at (51*2:2) {};
    \node (b3') at (51*5:2) {};
    \node (g1') at (51*3:2) {};
    \node (g2') at (51*4:2) {};
    \node (o1') at (51*0:2) {};
    \node (o2') at (51*1:2) {};

    \foreach \v in {b1,b2,b3} {
        \draw[black] (\v) to (\v');
        \draw[line width=3pt, semitransparent, blue] (\v) to (\v');
    }

    \foreach \v in {g1, g2} {
        \draw[black] (\v) to (\v');
        \draw[line width=3pt, semitransparent, ForestGreen] (\v) to (\v');
    }

    \foreach \v in {o1, o2} {
        \draw[black] (\v) to (\v');
        \draw[line width=3pt, semitransparent, red] (\v) to (\v');
    }

    \draw[black] (g1) to[bend left=45] (g2);
    \draw[line width=3pt, semitransparent, ForestGreen] (g1) to[bend left=45] (g2);

    
    
\end{tikzpicture}
     \end{subfigure}
     \hfill 
     \begin{subfigure}[b]{0.45\textwidth}
         \centering
         \begin{tikzpicture}

    \draw[fill,nearly transparent,gray] (0:0) circle[radius=1];
    
    \node (b1) at (51*6:1) {};
    \node (b2) at (51*2:1) {};
    \node (b3) at (51*5:1) {};
    \node (g1) at (51*3:1) {};
    \node (g2) at (51*4:1) {};
    \node (o1) at (51*0:1) {};
    \node (o2) at (51*1:1) {};
    \draw[fill=blue] (b1) circle (3pt);
    \draw[fill=blue] (b2) circle (3pt);
    \draw[fill=blue] (b3) circle (3pt);
    \draw[fill=ForestGreen] (g1) circle (3pt);
    \draw[fill=ForestGreen] (g2) circle (3pt);
    \draw[fill=red] (o1) circle (3pt);
    \draw[fill=red] (o2) circle (3pt);

    \node (b1') at (51*6:2) {};
    \node (b2') at (51*2:2) {};
    \node (b3') at (51*5:2) {};
    \node (g1') at (51*3:2) {};
    \node (g2') at (51*4:2) {};
    \node (o1') at (51*0:2) {};
    \node (o2') at (51*1:2) {};

    \foreach \v in {b1,b2,b3} {
        \draw[black] (\v) to (\v');
        \draw[line width=3pt, semitransparent, blue] (\v) to (\v');
    }

    \foreach \v in {g1, g2} {
        \draw[black] (\v) to (\v');
        \draw[line width=3pt, semitransparent, ForestGreen] (\v) to (\v');
    }

    \foreach \v in {o1, o2} {
        \draw[black] (\v) to (\v');
        \draw[line width=3pt, semitransparent, red] (\v) to (\v');
    }

    \draw[black] (g1) to[bend left=45] (g2);
    \draw[line width=3pt, semitransparent, ForestGreen] (g1) to[bend left=45] (g2);

    \node (bstar) at (-135:0) {};
    \draw[fill=blue, semitransparent] (bstar) circle (3pt);
    \foreach \i in {1,2,3} {
        \draw[blue,thick,decorate, decoration={snake, segment length=2mm, amplitude=0.5mm}] (bstar) to (b\i);
    }
    
    \node (ostar) at (30:0.5) {};
    \draw[fill=red, semitransparent] (ostar) circle (3pt);
    \foreach \i in {1,2} {
        \draw[red,thick,decorate, decoration={snake, segment length=2mm, amplitude=0.5mm}] (ostar) to (o\i);
    }
    
\end{tikzpicture}
     \end{subfigure}

     \caption{(Left) The disk expansion of a node $z$. Paths of the same colour are for the same commodity. The blue and red paths terminate at $z$, while the green path transits through $z$. There is no crossing or quasicrossing at $z$. (Right) We add a node into the disk for each commodity terminating there (blue, red) and connect the new node to the endpoints of that commodity's support paths. Since there is no crossing or quasicrossing at $z$, this operation can be done without causing edge crossings.}
     \label{fig:strongly_uncrossed_star}
\end{figure}


Consider a commodity $h = (s, t) \in E(H)$.
We use the notation $G_\phi[h]$ to mean the subgraph graph of $G_{\phi}$ induced by the paths $\widehat{\mathcal P}^h = \set{\widehat P : P \in \mathcal P^h}$.
In $G_\phi$, the paths of $\widehat P^h$ have a clockwise cyclic ordering about $s$, say $\widehat P_1, \dotsc, \widehat P_k$.
Informally, a {\em sector} is the region between consecutive paths in the clockwise ordering.
We make this more precise as follows.


Consider two paths $\widehat P_i$ and $\widehat P_j$ in $G_\phi$.
Orient $\widehat P_i$ from $s$ to $t$, and orient $\widehat P_j$ from $t$ to $s$.
The directed curve $\widehat P_i \bullet \widehat P_j$, obtained by concatenating these oriented curves, is a simple closed (Jordan) curve, and partitions the plane into two regions.
We define
$\CW^\circ(\widehat P_i, \widehat P_j)$ to be the set of points $q$ in $\mathbb R^2 \setminus (\widehat P_i \cup \widehat P_j)$ such that $q$ is ``to the right'' of $\widehat P_i \bullet \widehat P_j$; that is, if we draw a directed curve from $q$ to $\widehat P_i \bullet \widehat P_j$ (avoiding $\widehat P_i \bullet \widehat P_j$ except at the extremity), it hits $\widehat P_i \bullet \widehat P_j$ from the right.
We let $\CW(\widehat P_i, \widehat P_j)$ denote the closed region $\CW^\circ(\widehat P_i, \widehat P_j) \cup \widehat P_i \cup \widehat P_j$.
We analogously define
$\ACW^\circ(\widehat P_i, \widehat P_j)$ and $\ACW(\widehat P_i, \widehat P_j)$ by taking the opposite orientation.
Any curve that does not cross and is internally disjoint from $\widehat P_i$ and $\widehat P_j$ either stays entirely within $\CW(\widehat P_i, \widehat P_j)$ or stays entirely within $\ACW(\widehat P_i, \widehat P_j)$.

Now, for our clockwise-ordered family of paths $\widehat P_1, \dotsc, \widehat P_k$, we refer to $\CW^\circ(\widehat P_i, \widehat P_{i + 1})$ for $i = 1, \dotsc, k$ (addition modulo $k$) as the {\em open sectors} of $f^h$, and define the {\em closed sectors} similarly using $\CW(\widehat P_i, \widehat P_{i + 1})$.
The open sectors partition $\mathbb R^2 \setminus (\widehat P_1 \cup \dotsb \cup \widehat P_k)$. 
One sector contains the infinite face of $G_\phi[h]$; we call this the {\em outer sector} of $h$, and the others the {\em inner sectors} of $h$.
The {\em outer paths} of $h$ are the two paths $P,Q \in \mathcal P^h$ whose corresponding paths $\widehat P, \widehat Q \in \widehat{\mathcal P}^h$ define the outer sector of $h$.

Now consider an edge $e$ of the original graph $G$ that appears in $\mathcal P^h$.
If $e$ lies in any outer path of $\mathcal P^h$, it is an {\em outer edge} of $h$, otherwise it is an {\em inner edge} of $h$.
For our congestion result, we exploit the following key properties.


\begin{lemma}
    \label{lemma:private-region}
    
    Let $f$ be a strongly uncrossed walk-multiflow, with a fixed strongly uncrossed parallelization $\phi$.
    For distinct commodities $h, h' \in E(H)$, the sets of inner edges of $h$ and $h'$ are disjoint.
\end{lemma}

\begin{proof}
    To prove the result, we show that if an edge $e$ is shared between the two single-commodity flows in $G$, then it is an outer edge for one of the commodities.
    \begin{claim}
    $G_\phi[h']$ is entirely contained within a single sector of $h$, and $G_\phi[h]$ is entirely contained within a single sector of $h'$.
    \end{claim}
    \begin{subproof}
        By symmetry, it suffices to prove the first part.
        Recall that in $G_\phi$, two paths for distinct commodities are completely disjoint, so each individual $\widehat P \in \widehat {\mathcal P}^{h'}$ lies in at most one sector of $h$.
        However, all paths of $\widehat {\mathcal P}^{h'}$ share common terminals in $G_\phi$, so it must be the same sector for all of $\widehat {\mathcal P}^{h'}$.
    \end{subproof}

    \begin{claim}
        Either $G_\phi[h']$ is contained in the outer sector of $h$, or $G_\phi[h]$ is contained in the outer sector of $h'$.
    \end{claim}
    \begin{subproof}
        By the previous claim and by symmetry, it suffices to show that if $G_\phi[h']$ is contained within an inner sector of $h$, then $G_\phi[h]$ is contained within the outer sector of $h'$.
        Suppose $G_\phi[h']$ is contained within the inner sector of $h$ defined by paths $\widehat P_i$ and $\widehat P_{i + 1}$.
        Since $\CW(\widehat P_i, \widehat P_{i+1})$ is an inner sector, it is bounded, which means $G_\phi[h']$ is contained in the interior of the curve $\widehat P_i \bullet \widehat P_j$ rather than the exterior.
        This implies that $\widehat P_i$ and $\widehat P_{i + 1}$ are contained within the infinite face of $G_\phi[h']$, and so therefore are in the outer sector of $h'$.
        Since the entirety of $G_\phi[h]$ is within the same sector of $h'$, it is within the outer sector of $h'$.
    \end{subproof}
    
    Now, without loss of generality, suppose that $G_\phi[h]$ is contained within the outer sector of $h'$.
    Then $G_\phi[h]$ and $G_\phi[h']$ both contain one of the copies of the shared edge $e$.
    The copy belonging to $h$ is contained in the infinite face of $G_\phi[h']$, and (by how $G_\phi$ is constructed) so too is the copy belonging to $h'$.
    So then $e$ is an outer edge of $h'$, not an inner edge.
\end{proof}

Now, associating the circular ordering $P_1, \dotsc, P_k$ to $\mathcal P^h$ corresponding to the clockwise ordering $\widehat P_1, \dotsc, \widehat P_k$ of $\widehat {\mathcal P}^h$, we prove the second key property.

\begin{lemma}
\label{lemma:squeeze-property}
Suppose that in $G$, $P_i, P_j \in \mathcal P^h$ share an edge $e$.
One of the following holds:
\begin{itemize}
    \item every path between $P_i$ and $P_j$ in the circular ordering also uses $e$, or
    \item every path between $P_j$ and $P_i$ in the circular ordering also uses $e$.
\end{itemize}
\end{lemma}
\begin{proof}
    % ...
    In $G_\phi$, consider subdividing the copies of $e$ belonging to $\widehat P_i$ and $\widehat P_j$ by adding one node into each edge.
    Then add an edge $e'$ between the two new nodes, adding $e'$ such that it only crosses copies of $e$.
    The new edge either separates $\CW(\widehat P_i, \widehat P_j)$ or it separates $\ACW(\widehat P_i, \widehat P_j)$, where by separates we mean that any path of $\widehat {\mathcal P}^h$ that is within the sector must cross $e'$.
    The paths between $\widehat P_i$ and $\widehat P_j$ in the clockwise ordering all lie in $\CW(\widehat P_i, \widehat P_j)$, and the paths between $\widehat P_j$ and $\widehat P_i$ in the clockwise ordering all lie in $\ACW(\widehat P_i, \widehat P_j)$.
    If a path $\mathcal P_r$ crosses $e'$, then it means that it uses a copy of $e$, and hence its corresponding path $P_r$ in $G$ also uses $e$.
    The result follows.
\end{proof}


\subsection{Main Congestion Argument}
\label{sec:congestion-main}

We are now ready to prove \Cref{thm:congestion}.


\begin{proof}[Proof of \Cref{thm:congestion}]
    First, we show how to prove Item 2 (unsplittable flow), since Item 1 ultimately follows from it.
    Let $f$ be a feasible strongly uncrossed flow with a fixed strongly uncrossed parallelization $\phi$.
    Recall that we let $f(e) = \sum_{P : e \in P} f(P)$ and let $f^h$ be the single-commodity flow for $h \in E(H)$.
    
    Consider a commodity $h$.
    Let $P_1, \dotsc, P_k$ be paths of $\mathcal P^h$, ordered according to the cyclic clockwise ordering of their corresponding paths $\widehat P_1, \dotsc, \widehat P_k \in \widehat {\mathcal P}^h$ in $G_\phi$.
    Re-index so that $P_1$ and $P_k$ are the outer paths.
    Let $i^*$ be the smallest index $i$ such that $\sum_{j \le i} f^h(P_j) \ge \frac {d(h)} 2$.
    We define $Q^h = P_{i^*}$, and call it the {\em central path} for $h$.
    Note that it is straightforward to find $Q^h$ in polynomial time.

    \begin{claim}
        \label{claim:outer-congestion}
        If $e \in Q^h$ and $e$ is an outer edge for $h$ (i.e. it lies in $P_1$ or $P_k$) , then $f^h(e) \ge \frac {d(h)} 2$.
    \end{claim}
    \begin{subproof}
        Since $e$ is an outer edge, it is in either $P_1$ or $P_k$, in addition to $Q^h = P_{i^*}$.
        From \Cref{lemma:squeeze-property} we have that $e \in P_r$ either for all $1 \le r \le i^*$ or for all $i^* \le r \le k$.
        In the former case, $f^h(e) \ge \sum_{1 \le r \le i^*} f^h(P_r) \ge \frac {d(h)} 2$.
        For the latter case, $f^h(e) \ge \sum_{i^* \le r \le k} f^h(P_r) = d(h) - \sum_{1 \le r < i^*} f^h(P_r) > d(h) - \frac {d(h)} 2 = \frac {d(h)} 2$, where the last inequality follows since $\sum_{1 \le r < i^*} f^h(P_r) < \frac {d(h)} 2$ by definition of $i^*$.
    \end{subproof}


    To obtain an unsplittable flow $f'$, for each $h \in E(H)$ we route $d(h)$ units of flow along the path $Q^h$.
    For the congestion bound, consider an arbitrary edge $e \in E(G)$.
    We will bound the amount of flow through $e$.
    Let $H_e = \set{h \in E(H) : e \in Q^h}$,
    so $f'(e) = \sum_{h \in H_e} d(h)$.
    Define $O_e = \set{h \in H_e : e \text{ is an outer edge of } h}$ and also $I_e = \set{h \in H_e : e \text{ is an inner edge of } h}$.
    These sets partition $H_e$, so $f'(e) = \sum_{h \in O_e} d(h) + \sum_{h \in I_e} d(h)$.
    We know that $|I_e| \le 1$ from \Cref{lemma:private-region}.
    Additionally, by \Cref{claim:outer-congestion}, for each $h \in O_e$ we have that $f^h(e) \ge \frac {d(h)} 2$.
    We now split into two cases.

    Case 1: $\sum_{h \in O_e} \frac {d(h)} 2 < u(e)$.
    Then,
    \begin{equation*}
        f'(e) < 2 u(e) + \sum_{h \in I_e} d(h) \le 2 u(e) + d_{\max}.
    \end{equation*}

    Case 2: $\sum_{h \in O_e} \frac {d(h)} 2 \ge u(e)$.
    Then $\sum_{h \in O_e} f^h(e) \ge \sum_{h \in O_e} \frac {d(h)} 2 \ge u(e)$, so the commodities of $O_e$ saturate $e$ in $f$.
    It follows that $I_e = \emptyset$.
    Hence,
    \begin{equation*}
        f'(e) = \sum_{h \in O_e} d(h) \le \sum_{h \in O_e} 2 f^h(e) \le 2 f(e)\le 2 u(e) < 2 u(e) + d_{\max}.
    \end{equation*}

    This concludes the proof of Item 2. The above is easily turned into a polynomial time algorithm.
    We now prove Item 1 from Item 2.

    First, we split each $e \in E(G)$ into $u(e)$ parallel copies with capacity 1, and each $h \in E(H)$ into $d(h)$ parallel copies with demand 1, so that we have a new planar multiflow instance $(G', H', u' = \vec 1, d' = \vec 1)$ with unit capacities and demands.
    The splitting step is technically not polynomial time, but at the end we discuss how to convert this to an efficient algorithm.
    An integral congestion 2 multiflow for the new instance clearly implies an integral congestion 2 multiflow for the original instance.
    Given the feasible strongly uncrossed multiflow for the original instance, it is straightforward to construct a feasible strongly uncrossed multiflow for the new instance.
    To do this, for commodity $h = (s,t) \in E(H)$, order the support paths clockwise about $s$.
    The first unit-demand $st$-flow is obtained by iterating over the paths, assigning up to $f^{h}(P)$ amount of flow on each path $P$, until 1 unit of flow has been routed.
    We then continue iterating from this path and assign flow for the second unit-demand $st$-flow, ensuring we do not assign more than $f^{h}(P)$ flow in total to each path $P$, and so on.
    
    Then, Item 2 implies an there exists an unsplittable (and hence, integral) multiflow $f'$ for the unit weight instance, where the congestion on each edge is strictly less than 3.
    Since the flow is integral, the congestion is at most 2.

    We now briefly explain how to convert this into a polynomial time algorithm.
    Consider a demand $h = (s, t) \in E(H)$.
    The following is an equivalent description of how the pseudo-polynomial time algorithm above behaves.
    It picks a starting support path and scans clockwise about $s$ counting the total amount of flow it sees, finding the first path at which 0.5 units of total flow has been routed, then 1.5 units in total, then 2.5 units in total, and so on.
    These are what would be the central paths were we to split $h$ up into $d(h)$ separate demands.
    The algorithm iterates over this list, and routes one unit of flow on each path (if a walk appears multiple times in the list because its flow is large, we route one unit of flow on it each time it appears).
    However, this can be implemented in polynomial time: just by scanning clockwise about $s$ and performing a constant number of simple arithmetic calculations at each path, we can determine how many times each path $P$ would appear in this list of central paths without explicitly constructing the list, which tells us how much flow to route on $P$.
\end{proof}

\section{Reducing Uncrossed Edge-Disjoint Paths to Planar Node-Disjoint Paths}
\label{sec:reduction}

In this section we prove \Cref{thm:reduction}.
An instance of edge-disjoint paths or node-disjoint paths consists of $G, H$ (we can think of these as having unit capacities and demands).
We must find (integral) edge- or node-disjoint paths respectively that fully route $H$ in $G$.
The problem of interest for \Cref{thm:reduction} is {\em uncrossed} edge-disjoint paths, where $G$ is planar and we add the extra constraint that the paths are uncrossed.
We reduce to the standard planar node-disjoint paths problem, and then use a result of Schrijver \cite{schrijver1990homotopic} which says that node-disjoint paths can be solved in polynomial time in planar graphs if the number of faces the demand edges are incident upon is bounded.

Our proof uses contractions. We let $G/e$ to denote the graph obtained by contracting edge $e$ in a graph $G$; we use similar notation for a multiflow $f$.

\begin{lemma}
\label{lem:contractsafe}
    Let $f$ be an uncrossed flow in $G$ and $e$ some supply edge. Then $f/e$ is uncrossed for $G/e,H/e$.
\end{lemma}

We first present a gadget that reduces the degree of an arbitrary vertex of degree at least 5, by one, adding vertices of maximum degree 4. Using this gadget repeatedly, we get an equivalent instance with maximum degree 4. Then every degree-4 vertex can be replaced by a $C_4$, as illustrated in \Cref{fig:degree4-reduction}, to obtain a graph with maximum degree 3.

\begin{figure}
    \centering
    \begin{tikzpicture}[x=1cm,y=1cm]
        \foreach \u/\x/\y in {u/0/0, u1/3.5/0,u2/4/-0.5,u3/4.5/0,u4/4/0.5} {
            \node[stvertex] (\u) at (\x,\y) {};
        }
        \foreach \u/\v in {u1/u2,u2/u3,u3/u4,u4/u1} {
            \draw (\u) -- (\v);
        }   
        \foreach \u/\x/\y in {
            u/-1/0,u/0/1,u/1/0,u/0/-1,
            u1/3/0,u2/4/-1,u3/5/0,u4/4/1} 
        {
            \draw (\u) -- (\x,\y);
        }
        \draw (u) node[above left] {$u$};
        \draw (4,0) node {$u$};
        \draw (2,0) node {$\longrightarrow$};
    \end{tikzpicture}
    \caption{A gadget to replace a degree-4 vertex by 4 degree-3 vertices.}
    \label{fig:degree4-reduction}
\end{figure}

Consider a vertex $v$ with $\deg(v) \geq 5$. In order to obtain maximum degree 4, the natural approach replaces $v$ by a grid large enough to make any routing incident to $v$ feasible in that grid. This is more straightforward in instances where the terminals are leaves. As we have seen however, creating leaves may destroy uncrossability. We need to be more careful then  because $v$ can be a terminal for some of the demand pairs. Therefore we must pre-specify some vertices in the grid to substitute $v$ in those demand pairs. Because of the uncrossed constraint, the choice of these new terminals is more delicate. 
We give an explicit construction, presented in a recursive fashion: by showing how to decrease the degree of $v$ by one in $G$ and $H$ simultaneously.

\begin{figure}[htbp]
    \centering
    \begin{tikzpicture}[x=0.5cm,y=0.5cm]
        \node[stvertex] (v) at (0,0) {};
        \foreach \a in {52,104,...,312} {
          \foreach \r in {1,2,...,6} {
            \node[stvertex] (u\r'\a) at (\a:\r) {};
          }
        }
        \foreach \r in {2,3,...,6} {
            \node[stvertex] (u\r'0) at (0:\r) {};
        }
        \node[stterminal,red] (u1'0) at (0:1) {};
        \foreach \a in {52,104,...,312} {
          \draw (u) -- (u1'\a);
        }
        \foreach \a in {0,52,...,312} {
          \draw (u6'\a) -- (\a:7);
          \foreach \r [remember=\r as \lastr (initially 1)] in {2,3,...,6} {
              \draw (u\lastr'\a) -- (u\r'\a);
          }
        }
        \foreach \a [remember=\a as \lasta (initially 312)] 
          in {0,52,...,312} 
        {
            \foreach \r in {1,2,...,6} {
               \draw (u\r'\lasta) -- (u\r'\a);
            }
        }
        \draw (v) node[right] {\small $v'$ };
        \draw (u1'0) node[above right] {\small $t$};
        \foreach \i [evaluate=\i as \a using 52*\i] in {0,1,...,6} {
            \draw (\a:7.5) node {\small $v_\i$};
        }
        \draw[line width=3pt, semitransparent, blue] 
            (208:7) -- (u4'208) -- (u4'156) -- (u3'156) -- (u3'104) 
            -- (u2'104) -- (u2'52) -- (u1'52) -- (u1'0); 
        \draw[line width=3pt, semitransparent, yellow] 
            (156:7) -- (u4'156) -- (u4'104) -- (u3'104) -- (u3'52) 
            -- (u2'52) -- (u2'0) -- (u1'0) -- (u1'312) -- (v); 
        \draw[line width=3pt, semitransparent, yellow] 
            (312:7) -- (u4'312) -- (u4'260) -- (u3'260) -- (u3'208) 
            -- (u2'208) -- (u2'156) -- (u1'156) -- (u1'104) -- (v); 
        \draw[line width=3pt, semitransparent, ForestGreen] 
            (52:7) -- (u4'52) -- (u4'0) -- (u3'0) -- (u3'312) 
            -- (u2'312) -- (u2'260) -- (u1'260) -- (u1'208) -- (v); 
        \draw[line width=3pt, semitransparent, ForestGreen] 
            (0:7) -- (u4'0) -- (u4'312) -- (u3'312) -- (u3'260) 
            -- (u2'260) -- (u2'208) -- (u1'208) -- (u1'156) -- (v); 
        \draw[line width=3pt, semitransparent, red] 
            (104:7) -- (u4'104) -- (u4'52) -- (u3'52) -- (u3'0) 
            -- (u2'0) -- (u2'312) -- (u1'312) -- (u1'260) -- (v);

    \end{tikzpicture}
    \caption{The spider web gadget: how to decrease the degree of a vertex by one in $G$ and $H$, adding only vertices of degree at most $4$. Notice the red terminal that can be used as terminal for one of the demand pair ending at this vertex. In highlight with colors, an example of reconfiguration of four paths as given by the proof, with two paths which had been ending at $v$ (red and blue), and two paths passing through $v$ (yellow and green).}
    \label{fig:degree-reduction}
\end{figure}

\begin{proposition}
    Let $G,H$ be an instance of the uncrossed edge-disjoint paths problem, and $v$ be a vertex of $G$ with $\deg_G(v) > 4$, or $\deg_G(v) \geq 2$ and $\deg_H(v) \geq 1$. Define $G',H'$ from $G,H$ by replacing $v$ in $G$ by the spider web gadget from \Cref{fig:degree-reduction}. If $\deg_H(v) \geq 1$, then replace arbitrary demand $(s,v)$ in $H$  by a demand $(s,t)$ where $t$ is the red terminal in the gadget (which then has degree $1$ in $H'$). Any other terminals at $v$ are located at $v'$. Then $G,H$ admits an uncrossed routing if and only if $G',H'$ does.    
\end{proposition}

\begin{proof}
    Let $d = \deg(v)$. We label each vertex except $v$ in the spider web as $u_{i,j}$ such that $u_{i,j}$ is at the polar coordinate $(\frac{2i\pi}{d}, j)$ (for instance $t = u_{0,1}$), with $i \in \{0,\ldots,d-1\}$ and $j \in \{1,\ldots,d-1\}$. Let $v_0,\ldots,v_{d-1}$ be the neighbours of $v$, with $v_i$ adjacent to $u_{i,d-1}$. To avoid confusion, we will call $v'$ the central vertex in the spider web, and $v$ the original vertex in $G$. Also we will consider all indices modulo $d$. Given a routing in $G',H'$ by contracting all new vertices in the web (and deleting loops) they become the original $v$, and  we obtain a routing in $G,H$, proving the \emph{if} implication. 
    
    Consider an uncrossed routing $\mathcal{P}$ in $G,H$. We compute a corresponding uncrossed routing in $G',H'$. The first case is when $H$ does not contain a demand ending at $v$. We claim then there is a routing that does not use $v'$. Note that our figure does not depict this case and in particular $v'$ will not be used and can be safely removed, so that all of the new nodes have degree at most $4$. Indeed to see this routing, consider a support flow path $P$ containing $v_iv$, $vv_j$ (with $i < j$). Let $P'$ be the flow path obtained from $P$ by replacing $v_iv, vv_j$ by a subpath $v_i$, $u_{i,d-1}$, $u_{i,d-2}$, $\ldots u_{i,d-k}$, $u_{i+1,d-k}$, $u_{i+2,d-k}$,\ldots , $u_{j, d-k}$, $u_{j,d-k+1}$,\ldots, $u_{j,d-1}$, $v_j$, with $k = j - i + 1$. This subpath uses the rays indexed by $i$ and $j$ and the ring indexed by $d-k$. We claim that no two such subpaths share an edge. Indeed all the rays are distinct, and if two paths used the same ring, one from $v_i$ to $v_j$ and the other from $v_{i'}$ to $v_{j'}$, as $j - i = j' - i'$, we must have $v_i$, $v_{i'}$, $v_j$, $v_{j'}$ appearing in that order cyclically around $v$. Hence the two corresponding paths cross, contradiction.

    In the second case $H$ has some demands terminating at $v$; this case is depicted in \Cref{fig:degree-reduction}. Let $(s,v)$ be the arbitrary demand in $H$ terminating at $v$ that is replaced in $H'$ by a demand $(s,t)$, where $t=u_{0,1}$.
     Let $v_iv$ be the last edge in the $s,v$-path $P$ in $\mathcal{P}$. Replace in $P$ the edge $v_iv$ by the subpath $v_i$, $v_{i,d-1}$, $v_{i,d-2}$,\ldots, $v_{i,i}$, $v_{i-1,i}$, $v_{i-1,i-1}$,\ldots,$v_{2,2}$, $v_{1,2}$, $v_{1,1}$, $v_{0,1} = t$. For any other edge $vv_j$ in $H$ that appears in some path $Q \in \mathcal{P}$ (either transiting through or terminating at $v$), similarly replace $vv_j$ in $H$ by a subpath $v$, $v_{j-k,1}$, $v_{j-k+1,1}$, $v_{j-k+1,2}$,\ldots,$v_{j-k+i,i} = v_{j,i}$, $v_{j,i+1}$,\ldots, $v_{j,d-1}$, $v_j$. Each such subpath is thus spiraling counterclockwise from $v$ (or $t$ for $P$) into the $i$ innermost rings, then goes straight on its own ray, and thus all such subpaths, which are rotations of each other, are disjoint. Moreover, the relative cyclic order of the paths incident to $v'$ is the same as in $v$ (except for the path $P$ that is not incident to $v'$). Therefore, this routing in $G',H'$ is uncrossed.   
\end{proof}

We may repeatedly apply the preceding lemma, in polytime until we have a new "equivalent" instance $G',H'$ with the following properties. If $v$ has no terminals, then it has degree at most four and any terminal node has degree at most $3$. The last step of the reduction now applies a degree 4 reduction (Figure~\ref{fig:degree4-reduction}) to each remaining vertex of degree 4.  For each terminal node of degree at most $2$ in $G$, we leave it. If a terminal node $v$ has $d_G(v)=3$, then we subdivide each edge incident to $v$ and add a triangle onto these new degree $2$ nodes.
We show that this last reduction step is fine as follows.


\begin{lemma}
\label{lem:preservesmaxedp}
    Let $G$ be a maximum degree $4$ planar graph and $H$ an associated demand graph which forms a matching and if $\deg_H(v)=1$, then $\deg_G(v) \leq 3$. Let $G''$ be the graph obtained by performing the last phase of reductions. $H$ has an integral uncrossed flow in $G$ if and only if it has a totally node-disjoint flow in $G''$. 
\end{lemma}

\begin{proof} 
    We use the familiar idea that  feasibility of edge- versus node-disjoint paths are equivalent in maximum degree three instances. If $H$ has a node-disjoint solution in $G''$, then by planarity, there is an uncrossed edge-disjoint solution in $G$. We obtain an edge-disjoint solution in $G$ by first contracting the new $4$ cycles. This cannot create a crossing at the contracted node due to the fact we had a node-disjoint solution. Second, consider one of the nodes $v$ where we applied triangle operation. Regardless of whether a demand transits through the triangle gadget, we obtain an uncrossed edge-disjoint flow by contracting $v$ together with the new degree $2$ nodes.  
    
    Conversely, suppose that we have an uncrossed edge-disjoint path solution in $G$.
    At any non-terminal node $v$, if there are two paths using it, since they are non-crossing they can be realized as node-disjoint on the 4-cycle introduced by the reduction. 
 Now consider a terminal node $v$. If it has degree at most $2$, then no other path could use edges incident to $v$, so assume it has degree $3$. We can convert this to a node-disjoint solution by first letting the terminal $v$ route via the node which subdivided the edge it used in $G$. This allows us to use the edge between the two other subdivided nodes for the transiting demand, thus keeping node disjointness at the gadget.
\end{proof}

This completes the proof of the overall reduction.

\begin{theorem}
\label{thm:reduction-final}
Suppose that $G,H$ is a planar instance of edge-disjoint paths.
We can compute a planar instance  $\widehat{G},\widehat{H}$ of node-disjoint paths such that $G,H$ has an integral uncrossed flow if and only if $\widehat{G},\widehat{H}$ has (totally) node-disjoint paths for the demands in $\hat{H}$. 
Moreover, $\widehat{G},\widehat{H}$ can be computed in polytime, and $|V(\widehat{G})| \leq c \Delta^3 |V(G)|$ where $\Delta$ is the maximum degree in $G$ and $c$ is a constant.
\end{theorem}



\resolved{
\bruce{question1: what do flows that use trails end up looking like in the reduced instance? question2: can we produce a reduction which preserves maximization? IE that a set S of demands is routable in the origianl graph if and only if it is NDP feasible in reduced graph? Not even clear to me for Lemma 4.3.}}

\bibliographystyle{plain}
\bibliography{mybibliography}

@article{matsumoto1986planar,
  title={Planar multicommodity flows, maximum matchings and negative cycles},
  author={Matsumoto, Kazuhiko and Nishizeki, Takao and Saito, Nobuji},
  journal={SIAM Journal on Computing},
  volume={15},
  number={2},
  pages={495--510},
  year={1986},
  publisher={SIAM}
}

@inproceedings{chuzhoy2018almost,
  title={Almost polynomial hardness of node-disjoint paths in grids},
  author={Chuzhoy, Julia and Kim, David HK and Nimavat, Rachit},
  booktitle={Proceedings of the 50th Annual ACM SIGACT Symposium on Theory of Computing},
  pages={1220--1233},
  year={2018}
}

@article{chuzhoyjournalv017a006,
 author = {Chuzhoy, Julia and Kim, David Hong Kyun and Nimavat, Rachit},
 title = {Almost Polynomial Hardness of Node-Disjoint Paths in Grids},
 year = {2021},
 pages = {1--57},
 doi = {10.4086/toc.2021.v017a006},
 publisher = {Theory of Computing},
 journal = {Theory of Computing},
 volume = {17},
 number = {6},
 URL = {https://theoryofcomputing.org/articles/v017a006},
}

@article{chekuri2009approximate,
  title={Approximate integer decompositions for undirected network design problems},
  author={Chekuri, Chandra and Shepherd, F Bruce},
  journal={SIAM Journal on Discrete Mathematics},
  volume={23},
  number={1},
  pages={163--177},
  year={2009},
  publisher={SIAM}
}

@inproceedings{traub2024single,
  title={Single-Source Unsplittable Flows in Planar Graphs},
  author={Traub, Vera and Koch, Laura Vargas and Zenklusen, Rico},
  booktitle={Proceedings of the 2024 Annual ACM-SIAM Symposium on Discrete Algorithms (SODA)},
  pages={639--668},
  year={2024},
  organization={SIAM}
}

@article{schlomberg2024improved,
  title={An improved integrality gap for disjoint cycles in planar graphs},
  author={Schlomberg, Niklas},
  journal={arXiv preprint arXiv:2404.17813},
  year={2024}
}

@inproceedings{schlomberg2023packing,
  title={Packing cycles in planar and bounded-genus graphs},
  author={Schlomberg, Niklas and Thiele, Hanjo and Vygen, Jens},
  booktitle={Proceedings of the 2023 Annual ACM-SIAM Symposium on Discrete Algorithms (SODA)},
  pages={2069--2086},
  year={2023},
  organization={SIAM}
}

@article{goemans1998primal,
  title={Primal-dual approximation algorithms for feedback problems in planar graphs},
  author={Goemans, Michel X and Williamson, David P},
  journal={Combinatorica},
  volume={1},
  number={18},
  pages={37--59},
  year={1998}
}

@article{garg2022integer,
  title={Integer plane multiflow maximisation: one-quarter-approximation and gaps},
  author={Garg, Naveen and Kumar, Nikhil and Seb{\H{o}}, Andr{\'a}s},
  journal={Mathematical Programming},
  volume={195},
  number={1},
  pages={403--419},
  year={2022},
  publisher={Springer}
}

@article{vygen1995np,
  title={NP-completeness of some edge-disjoint paths problems},
  author={Vygen, Jens},
  journal={Discrete Applied Mathematics},
  volume={61},
  number={1},
  pages={83--90},
  year={1995},
  publisher={Elsevier}
}

@article{huang2021approximation,
  title={An approximation algorithm for fully planar edge-disjoint paths},
  author={Huang, Chien-Chung and Mari, Mathieu and Mathieu, Claire and Schewior, Kevin and Vygen, Jens},
  journal={SIAM Journal on Discrete Mathematics},
  volume={35},
  number={2},
  pages={752--769},
  year={2021},
  publisher={SIAM}
}

@article{huang2023approximating,
  title={Approximating maximum integral multiflows on bounded genus graphs},
  author={Huang, Chien-Chung and Mari, Mathieu and Mathieu, Claire and Vygen, Jens},
  journal={Discrete \& Computational Geometry},
  volume={70},
  number={4},
  pages={1266--1291},
  year={2023},
  publisher={Springer}
}

@misc{pach,
  author = "P\'ach, Janos",
  date = "2025",
  howpublished = "personal communication"
}

@article{GargVY97,
	Author = {Naveen Garg and Vijay V. Vazirani and Mihalis Yannakakis},
	Journal = {Algorithmica},
	Number = {1},
	Pages = {3-20},
	Title = {Primal-Dual Approximation Algorithms for Integral Flow and Multicut in Trees.},
	Volume = {18},
	Year = {1997}}

@inproceedings{chekuri2013maximum,
  title={Maximum edge-disjoint paths in k-sums of graphs},
  author={Chekuri, Chandra and Naves, Guyslain and Shepherd, F Bruce},
  booktitle={International Colloquium on Automata, Languages, and Programming},
  pages={328--339},
  year={2013},
  organization={Springer}
}

@book{korte1990paths,
  title={Paths, flows, and VLSI-layout},
  author={Korte, Bernhard and Lov{\`a}sz, L{\`a}szlo and Hans J{\"u}rgen Pr{\"o}mel and Alexander Schrijver},
  publisher={Springer-Verlag},
  year={1990}
}

@inproceedings{espinosa2026unsplittable,
  title={Unsplittable Flow Cut Gap in Undirected Graphs},
  author={Espinosa, David Alem{\'a}n and Kumar, Nikhil and Poremba, Joseph and Shepherd, Bruce},
  booktitle={Proceedings of the 2026 Annual ACM-SIAM Symposium on Discrete Algorithms (SODA)},
  pages={1570--1605},
  year={2026},
  organization={SIAM}
}

@inproceedings{chekuri2010flow,
  title={Flow-cut gaps for integer and fractional multiflows},
  author={Chekuri, Chandra and Shepherd, F Bruce and Weibel, Christophe},
  booktitle={Proceedings of the Twenty-First Annual ACM-SIAM Symposium on Discrete Algorithms},
  pages={1198--1208},
  year={2010},
  organization={Society for Industrial and Applied Mathematics}
}

@book{schrijver2003combinatorial,
  title={Combinatorial optimization: polyhedra and efficiency},
  author={Schrijver, Alexander},
  volume={24},
  year={2003},
  publisher={Springer}
}

@article{chekuri2007multicommodity,
  title={Multicommodity demand flow in a tree and packing integer programs},
  author={Chekuri, Chandra and Mydlarz, Marcelo and Shepherd, F Bruce},
  journal={ACM Transactions on Algorithms (TALG)},
  volume={3},
  number={3},
  pages={27},
  year={2007},
  publisher={ACM}
}

@inproceedings{chakrabarti2012cut,
  title={When the cut condition is enough: A complete characterization for multiflow problems in series-parallel networks},
  author={Chakrabarti, Amit and Fleischer, Lisa and Weibel, Christophe},
  booktitle={Proceedings of the forty-fourth annual ACM symposium on Theory of computing},
  pages={19--26},
  year={2012},
  organization={ACM},
  doi = {10.1145/2213977.2213980}
}

@article{chekuri2013flow,
  title={Flow-cut gaps for integer and fractional multiflows},
  author={Chekuri, Chandra and Shepherd, F Bruce and Weibel, Christophe},
  journal={Journal of Combinatorial Theory, Series B},
  volume={103},
  number={2},
  pages={248--273},
  year={2013},
  publisher={Elsevier},
  doi = {10.1016/j.jctb.2012.11.002}
}

@article{seymour1981matroids,
  title={Matroids and multicommodity flows},
  author={Seymour, Paul D},
  journal={European Journal of Combinatorics},
  volume={2},
  number={3},
  pages={257--290},
  year={1981},
  publisher={Elsevier},
  doi = {10.1016/S0195-6698(81)80033-9}
}

@article{schrijver1991disjoint,
  title={Disjoint homotopic paths and trees in a planar graph},
  author={Schrijver, Alexander},
  journal={Discrete \& Computational Geometry},
  volume={6},
  pages={527--574},
  year={1991},
  publisher={Springer}
}

@article{schrijver1990homotopic,
  title={Homotopic routing methods},
  author={Schrijver, Alexander},
  journal={Paths, Flows, and VLSI-layout},
  pages={329--371},
  year={1990},
  publisher={Springer-Verlag, Berlin, Germany}
}

@inproceedings{KPR,
 author = {Klein, Philip and Plotkin, Serge A. and Rao, Satish},
 title = {Excluded minors, network decomposition, and multicommodity flow},
 booktitle = {Proceedings of the twenty-fifth annual ACM symposium on Theory of computing},
 series = {STOC '93},
 year = {1993},
 isbn = {0-89791-591-7},
 location = {San Diego, California, USA},
 pages = {682--690},
 numpages = {9},
 url = {http://doi.acm.org/10.1145/167088.167261},
 doi = {10.1145/167088.167261},
 acmid = {167261},
 publisher = {ACM},
 address = {New York, NY, USA},
}

@article{middendorf_complexity_1993,
	title = {On the complexity of the disjoint paths problem},
	volume = {13},
	issn = {0209-9683, 1439-6912},
	url = {http://link.springer.com/10.1007/BF01202792},
	doi = {10.1007/BF01202792},
	language = {en},
	number = {1},
	urldate = {2023-07-20},
	journal = {Combinatorica},
	author = {Middendorf, Matthias and Pfeiffer, Frank},
	month = mar,
	year = {1993},
	pages = {97--107},
}

@article{esperet_coloring_2009,
	title = {Coloring a set of touching strings},
	volume = {34},
	copyright = {https://www.elsevier.com/tdm/userlicense/1.0/},
	issn = {15710653},
	url = {https://linkinghub.elsevier.com/retrieve/pii/S1571065309000766},
	doi = {10.1016/j.endm.2009.07.035},
	language = {en},
	urldate = {2025-06-03},
	journal = {Electronic Notes in Discrete Mathematics},
	author = {Esperet, Louis and Gonçalves, Daniel and Labourel, Arnaud},
	month = aug,
	year = {2009},
	pages = {213--217},
}

@unpublished{fox_touching_nodate,
	title = {Touching {Strings}},
	author = {Fox, Jacob and Pach, János},
    year = {2012},
}

@article{van_batenburg_coloring_2017,
	title = {Coloring {Jordan} {Regions} and {Curves}},
	volume = {31},
	issn = {0895-4801, 1095-7146},
	url = {https://epubs.siam.org/doi/10.1137/16M1092726},
	doi = {10.1137/16M1092726},
	language = {en},
	number = {3},
	urldate = {2025-06-03},
	journal = {SIAM Journal on Discrete Mathematics},
	author = {Van Batenburg, Wouter Cames and Esperet, Louis and Müller, Tobias},
	month = jan,
	year = {2017},
	pages = {1670--1684},
}

@article{naves_hardness_2012,
	title = {The hardness of routing two pairs on one face},
	volume = {131},
	copyright = {http://www.springer.com/tdm},
	issn = {0025-5610, 1436-4646},
	url = {http://link.springer.com/10.1007/s10107-010-0343-0},
	doi = {10.1007/s10107-010-0343-0},
	language = {en},
	number = {1-2},
	urldate = {2025-09-06},
	journal = {Mathematical Programming},
	author = {Naves, Guyslain},
	month = feb,
	year = {2012},
	pages = {49--69},
}

@article{raghavan_probabilistic_1988,
	title = {Probabilistic construction of deterministic algorithms: {Approximating} packing integer programs},
	volume = {37},
	copyright = {https://www.elsevier.com/tdm/userlicense/1.0/},
	issn = {00220000},
	shorttitle = {Probabilistic construction of deterministic algorithms},
	url = {https://linkinghub.elsevier.com/retrieve/pii/0022000088900037},
	doi = {10.1016/0022-0000(88)90003-7},
	language = {en},
	number = {2},
	urldate = {2025-11-07},
	journal = {Journal of Computer and System Sciences},
	author = {Raghavan, Prabhakar},
	month = oct,
	year = {1988},
	pages = {130--143},
}

@article{tardos_improved_1993,
	title = {Improved bounds for the max-flow min-multicut ratio for planar and {K}-free graphs},
	volume = {47},
	issn = {00200190},
	url = {https://linkinghub.elsevier.com/retrieve/pii/0020019093902282},
	doi = {10.1016/0020-0190(93)90228-2},
	language = {en},
	number = {2},
	urldate = {2026-05-06},
	journal = {Information Processing Letters},
	author = {Tardos, Éva and Vazirani, Vijay V.},
	month = aug,
	year = {1993},
	pages = {77--80},
}

@article{garg_primal-dual_1997,
	title = {Primal-dual approximation algorithms for integral flow and multicut in trees},
	volume = {18},
	copyright = {http://www.springer.com/tdm},
	issn = {0178-4617, 1432-0541},
	url = {http://link.springer.com/10.1007/BF02523685},
	doi = {10.1007/BF02523685},
	language = {en},
	number = {1},
	urldate = {2026-05-07},
	journal = {Algorithmica},
	author = {Garg, N. and Vazirani, V. V. and Yannakakis, M.},
	month = may,
	year = {1997},
	pages = {3--20},
}

@article{cornaz_max-multiflowmin-multicut_2011,
	title = {Max-multiflow/min-multicut for {G} + {H} series-parallel},
	volume = {311},
	issn = {0012365X},
	url = {https://linkinghub.elsevier.com/retrieve/pii/S0012365X11002421},
	doi = {10.1016/j.disc.2011.05.025},
	language = {en},
	number = {17},
	urldate = {2026-05-07},
	journal = {Discrete Mathematics},
	author = {Cornaz, Denis},
	month = sep,
	year = {2011},
	pages = {1957--1967},
}

@article{bentz_disjoint_2009,
	title = {Disjoint paths in sparse graphs},
	volume = {157},
	copyright = {https://www.elsevier.com/tdm/userlicense/1.0/},
	issn = {0166218X},
	url = {https://linkinghub.elsevier.com/retrieve/pii/S0166218X0900105X},
	doi = {10.1016/j.dam.2009.03.009},
	language = {en},
	number = {17},
	urldate = {2026-05-06},
	journal = {Discrete Applied Mathematics},
	author = {Bentz, Cédric},
	month = oct,
	year = {2009},
	pages = {3558--3568},
}

\appendix

\section{Uncrossing Procedures and Uncrossable Classes}

By cycle we always mean simple cycle.
For a cycle $C$, we use $\Interior(C)$ and $\InteriorOpen(C)$ for the closed and open interior respectively of the region bounded by $C$.

\subsection{Fully-Compliant Layouts}
\label{sec:fully-compliant}

Let $G$ be series-parallel.
A demand $h$ is {\em fully-compliant} if $G+h$ is series-parallel.
We say $(G, H)$ is {\em fully-compliant} if every $h \in E(H)$ is fully-compliant.
This notion was introduced in \cite{chekuri2013flow}, where it was also proven that such $(G, H)$ are cut-sufficient.
Here, we prove \Cref{prop:fully-compliant-pairwise-planar}
\begin{proof}[Proof of \Cref{prop:fully-compliant-pairwise-planar}]


    To show $(G, H)$ is pairwise-planar, it suffices to show for each demand $h = (s, t)$ that either $(s, t) \in E(G)$ or there exist two distinct faces $F_1, F_2$ of $G$ such that $s, t \in V(F_1)$ and $s, t \in V(F_2)$.
    Suppose $(s, t) \notin E(G)$.
    Let $C_0$ be a cycle\footnote{Recall that by cycle we always mean simple cycle.} containing both $s$ and $t$ (take two internally node-disjoint $st$-paths).

    Now, select a cycle $C$ containing both $s$ and $t$ such that $\Interior(C) \subseteq \Interior(C_0)$; subject to these constraints, choose $C$ such that $\Interior(C)$ is minimal.
    Let $G_C$ be the graph induced by $V(\Interior(C))$.
    Let $P, Q$ be the two internally node-disjoint $st$-paths comprising $C$.

    Let $V_P = V(P) \setminus \set{s, t}$ and $V_Q = V(P) \setminus \set{s, t}$.
    Since $(s, t) \notin E(G)$, neither of these sets is empty.
    Observe that there is no path in $G_C - \set{s, t}$ between $V_P$ and $V_Q$, as otherwise $G_C + h$ contains a $K_4$-minor, contradicting that $G + h$ is series-parallel.
    That is, $\set{s, t}$ is a 2-cut of $G_C$ separating $V_P$ and $V_Q$.
    Note that $G_C - \set{s, t}$ can only contain two connected components (one containing $V_P$, another containing $V_Q$), since if there were another in the interior of $G_C$ then we could find a cycle $C'$ containing $s, t$ such that $\Interior(C') \subsetneq \Interior(C)$, contradicting our choice of $C$.
    Let $G_P$ and $G_Q$ denote the two connected components plus their respective edges to $s$ and $t$.
    Since $G_P$ and $G_Q$ are disjoint except at $s, t$, one can draw a curve in $\Interior(C)$ between $s$ and $t$ that is disjoint from the embeddings of $G_P$ and $G_Q$ - this means there is a face $F_1$ of $G$ inside $\Interior(C)$ containing both $s, t$.

    Now, we can play the same game using cycle exteriors.
    Let $D$ be a simple cycle containing both $s$ and $t$ such that $\Exterior(D) \subseteq \Exterior(C_0)$ (i.e. $\Interior(C_0) \subseteq \Interior(D)$); subject to these constraints, choose $D$ so that $\Exterior(D)$ is minimal.
    By the same argument, there exists a face $F_2$ of $G$ in $\Exterior(D)$ such that $F_2$ contains both $s, t$.
    Since $\Exterior(D) \subseteq \Exterior(C_0)$, $F_2 \neq F_1$.

    
\end{proof}






\subsection{Local Uncrossing Lemma}
\label{sec:local-uncrossing}

In this subsection, we prove \Cref{corollary:cross-once}.

We always discuss a flow $f$ with a particular parallelization (not necessarily uncrossed, see Section~\ref{sec:uncrossed}), so we can assume support paths are edge-disjoint and thus cross only at nodes.
We define the following.
For $v \in V$ and distinct paths $P, Q$, let $\kappa(f, v, P, Q)$ be $f(P) \cdot f(Q)$ if $P$ and $Q$ cross at $v$, and 0 otherwise.
For $v \in V$, let $\kappa(f, v) = \sum_{P, Q} \kappa(f, v, P, Q)$\footnote{You could order the paths $P, Q$ to avoid double-counting them in the sums, but it ultimately does not matter for this proof.}.
We define the {\em crossing weight} of $f$ to be $\sum_v \kappa(f, v)$.
The key lemma is as follows.

\begin{lemma}
[Local Uncrossing Lemma]
\label{lemma:cross-twice-uncross}
    Let $(G, H, u, d)$ be a feasible planar multiflow instance.
    Let $f$ be a feasible flow that is minimum cost under unit edge costs.
    Suppose there exist support paths $P, Q$ that cross at two distinct nodes $z_1$ and $z_2$.
    Then there exists a feasible flow $f'$ for $(G, H, u, d)$ such that:
    \begin{itemize}
        \item $f'$ is minimum cost,
        \item $\kappa(f', v) \le \kappa(f, v)$ for all $v \notin \set{z_1, z_2}$, and
        \item $\kappa(f', z_i) \le \kappa(f, z_i) - \epsilon^2$ for $i = 1, 2$, where $\epsilon = \min \set{f(P), f(Q)}$.
    \end{itemize}
\end{lemma}

\begin{proof}
Let $f$ be a flow with a given parallelization.
Let $P, Q$ be two (edge-disjoint) support paths that cross at nodes $z_1$ and $z_2$.
Let $h_P = (s_P, t_P), h_Q = (s_Q, t_Q)$ be their respective demands (possibly equal).
By re-labeling their extremities, we may assume both $P$ and $Q$ use $z_1$ before $z_2$ starting from $s_P$ and $s_Q$ respectively.
Let $\epsilon = \min\set{f(P), f(Q)}$.
By embedding new edge-disjoint copies of $P, Q$ immediately beside their original copies, we can assume that $f(P) = f(Q) = \epsilon$.

We decompose each of $P, Q$ into three subpaths: for $P$ these are $P_1 = P[s_P, z_1], P_2 = P[z_1, z_2], P_3 = P[z_2, t_P]$, and for $Q$ these are $Q_1 = Q[s_Q, z_1], Q_2 = Q[z_1, z_2], Q_3 = Q[z_2, t_Q]$.
Now, define $P' = P_1 \bullet Q_2 \bullet P_3$ and $Q' = Q_1 \bullet P_2 \bullet Q_3$.
We transfer the $\epsilon$ units of flow on each of $P, Q$ to $P', Q'$ respectively to define a new flow $f'$ that routes the same demands $d$ and has the same cost as $f$.

One important catch is that $P', Q'$ may not be (simple) paths, even if $P, Q$ are paths.
For example, $Q_2$ may share nodes with $P_3$.
As discussed in \Cref{sec:walks}, short-cutting walks can introduce new crossings.
However, if $P', Q'$ are not paths, then by short-cutting we obtain a flow whose cost is strictly less than that of $f$, a contradiction.
\footnote{Instead of using the minimum cost assumption, one could instead generalize the lemma to show how to decrease the crossing for non-simple flows. We use the minimum cost assumption only because the argument is more concise.}

Consider the change in the crossing weight contributions.
First, consider any $v \notin \set{z_1, z_2}$.
The pair of edges that are consecutive for support paths transiting through $v$ are the same for both $f$ and $f'$ (just possibly not belonging to the same support paths, because of swapping $P_2$ and $Q_2$), and moreover the amount of flow on those edges is the same between $f$ and $f'$.
Hence, $\kappa(f', v) = \kappa(f, v)$ for such $v$.

Now consider $v \in \set{z_1, z_2}$ and a support path $R \notin \set{P, Q, P', Q'}$ that transits through $v$.
Of course, for any support path $T \notin \set{P, Q, P', Q'}$, we have that $\kappa(f', v, R, T) = \kappa(f, v, R, T)$, so we are interested in how $R$ interacts with $P, Q, P', Q'$.
Consider a sufficiently small disc centered at $v$.
The edges from the crossing paths $P, Q$ divide the disc into four regions.
The path $R$ starts in one region, then transits through $v$ into another region.
If the second region is consecutive (clockwise or counterclockwise) with the first, then $R$ crosses exactly one of $P, Q$ and it also crosses exactly one of $P', Q'$, and hence $\sum_{T \in \set{P', Q'}} \kappa (f', v, R, T) = \sum_{T \in \set{P, Q}} \kappa(f, v, R, T)$.
If the second region is instead opposite the first, then $R$ crosses both $P$ and $Q$, and regardless how many of $P', Q'$ it also crosses we have that $\sum_{T \in \set{P', Q'}} \kappa (f', v, R, T) \le \sum_{T \in \set{P, Q}} \kappa(f, v, R, T)$.

Finally, consider $v \in \set{z_1, z_2}$ and the contributions from the paths $P, Q, P', Q'$.
The paths $P$ and $Q$ cross with $\kappa(f, v, P, Q) = f(P) \cdot f(Q) = \epsilon^2$, while $P'$ and $Q'$ do not cross, so $\kappa(f', v, P', Q') = 0$.

Hence, for any $v \notin \set{z_1, z_2}$ we have that $\kappa(f', v) \le \kappa(f, v)$, and for $v \in \set{z_1, z_2}$ we have that $\kappa(f', v) < \kappa(f, v)$, a contradiction.

\end{proof}

By a very similar proof, one can show how to decrease the crossing weight when there exist two support paths that share a terminal and cross once at another node. 
A consequence of this lemma is that by taking a minimum cost flow, which subject to that also minimizes the crossing weight, we may assume flow paths cross at most once.
We thus obtain \Cref{corollary:cross-once} as a corollary.

\subsection{Nested Supports Lemma}
\label{sec:nested-supports}

This subsection is dedicated to a proof of \Cref{lemma:nested-support-lemma}.

For a face $F$ and support path $P \in \support(x^F)$, we let $C(P)$ be the embedded cycle obtained by taking the demand $h$ that $P$ routes and embedding it within the open interior of the face $F$ to join the ends of $P$.
Note that the embeddings of different demand edges within a face $F$ may cross each other (but of course, for different faces the demands do not cross, and moreover no demand edge crosses any supply edge).

Finally, some more terminology and notation.
Recall that for a (simple) cycle $C$, we use $\Interior(C)$ and $\InteriorOpen(C)$ for the closed and open interior respectively of the region bounded by $C$.
Additionally, when we say an object is contained within a face $F$, we are not referring to the boundary of $F$, but the region of the plane that corresponds to $F$.

\begin{lemma}
    [Face Uncrossing Lemma]
    \label{lemma:face-uncrossing}
    Let $(G, H, u, d)$ be a facial multiflow instance that is feasible, and let $(H_F : F \in \mathcal F)$ be a facial decomposition.
    There exists a feasible flow $x$ such that the following holds: for every pair of distinct faces $F_1, F_2 \in \mathcal F$ and any $P_1 \in \support(x^{F_1})$ and $P_2 \in \support(x^{F_2})$, the paths $P_1$ and $P_2$ do not cross, and (more strongly) the cycles $C(P_1)$ and $C(P_2)$ do not cross.
\end{lemma}

\begin{proof}
    Augment the instance by adding leaves as follows: for every face $F$ and demand $h = (s, t) \in H_F$, add leaves $(s_h, s), (t_h, t)$ of capacity equal to $d(h)$ and embed the leaves into the face $F$, and move the demand $h$ from $(s, t)$ to $(s_h, t_h)$.
    We keep the same facial decomposition and still embed the demand edge $(s_h, t_h)$ within the face $F$ (accepting that demand edges within $F$ may cross).
    It suffices to prove that, for this augmented instance, there exists a feasible flow $x$ such that for distinct faces $F_1, F_2 \in \mathcal F$ and paths $P_1 \in \support(x^{F_1}), P_2 \in \support(x^{F_2})$, the cycles $C(P_1)$ and $C(P_2)$ do not cross (since by contracting the added leaves, we obtain the desired flow for the original instance).

    Let $x$ be a feasible flow for the augmented instance such that any pair of support paths crosses at most once, which exists by \Cref{corollary:cross-once}.
    Let $F_1, F_2$ be distinct faces, and let $P_1 \in \support(x^{F_1}), P_2 \in \support(x^{F_2})$.
    Suppose for the sake of contradiction that $C(P_1)$ and $C(P_2)$ cross.
    Cycles that cross must cross at least twice.
    However, since $C(P_1), C(P_2)$ do not share their leaf edges or demand edges, these crossings must happen on subpaths of the paths $P_1, P_2$.
    But then this contradicts that $P_1, P_2$ cross at most once.
    
\end{proof}


\begin{proof}[Proof of \Cref{lemma:nested-support-lemma}]

    We again augment the instance by adding leaves, as in the proof of \Cref{lemma:face-uncrossing}.
    Let $x$ be a flow as resulting from \Cref{lemma:face-uncrossing}.
    First, we can assume without loss of generality that for each face $F$, $E(F) \subseteq \support(x^F)$ (and in particular, $F$ is a face of $G[\support(x^F)]$).
    If this is not the case for some face $F$, we can add new demands parallel to each edge of $E(F)$ of demand weight $\epsilon > 0$, and increasing the capacity of each edge in $E(F)$ by $\epsilon$, and extend $x$ by routing along these single-edge paths.
    This does not introduce any new crossings.
    Now, it suffices to show the nested supports property holds for $x$.
    
    By using parallelizations, we may assume all paths are edge-disjoint.
    Consider two distinct faces $F_1, F_2$.
    To simplify the notation, for $i = 1, 2$ let $x^i = x^{F_i}$ and let $G_i = G[\support(x^i)] = G[\support(x^i) \cup E(F_i)]$. 
    For $i = 1, 2$, we will use the notation $\overline i$ to mean the number in $\set{1, 2} \setminus \set{i}$.
    First, we prove that for each $i \in \set{1, 2}$, $G_i$ is contained in a single face of $G_{\overline i}$.
    We will deal with the stellated supports at the end.
    

    Observe that for $i \in \set{1, 2}$ and any $P \in \support(x^i)$, the other face $F_{\overline i}$ is contained in either $\Exterior(C(P))$ or $\Interior(C(P))$.
    
    \begin{claim}
        \label{claim:face-implies-all-paths}
        Let $i \in \set{1, 2}$.
        \begin{itemize}
            \item If the face $F_i$ is contained in $\Interior(C(P))$ for $P \in \support(x^{\overline i})$, then every path $Q \in \support(x^i)$ is contained in $\Interior(C(P))$, and hence $G_i$ is contained in $\Interior(C(P))$.
            \item If the face $F_i$ is contained in $\Exterior(C(P))$ for $P \in \support(x^{\overline i})$, then every path $Q \in \support(x^i)$ is contained in $\Exterior(C(P))$, and hence $G_i$ is contained in $\Exterior(C(P))$.
        \end{itemize}
    \end{claim}
    \begin{proof}
        Suppose that $F_i$ is contained in $\Interior(C(P))$ for $P \in \support(x^{\overline i})$.
        The cycle $C(Q)$ has an edge (namely, its demand edge) that lies in $\InteriorOpen(C(P))$.
        If it leaves $\Interior(C(P))$, it must cross the cycle $C(P)$, which means it must cross the path $P$.
        This contradicts the earlier claim.
        The case for $\Exterior(C(P))$ is similar.
    \end{proof}


    We say $F_i$ is {\em universally exterior} if $F_i$ is contained in $\Exterior(C(P))$ for all $P \in \support(x^{\overline i})$.
    Note that if $F_i$ is universally exterior, then by \Cref{claim:face-implies-all-paths}, $G_i$ is contained in $\Exterior(C(P))$ for all $P \in \support(x^{\overline i})$.

    \begin{claim}
        Either $F_1$ is universally exterior or $F_2$ is universally exterior.
    \end{claim}
    \begin{proof}
        There are two cases.

        Case 1: there exists $P_1 \in \support(x^1)$ such that $F_2$ is contained in $\Interior(C(P_1))$.
        Now, consider $P_2 \in \support(x^2)$.
        Then by \Cref{claim:face-implies-all-paths}, we have that $P_2$ is contained in $\Interior(C(P_1))$, and consequently $P_1$ is contained in $\Exterior(C(P_2))$.
        Furthermore, this implies that $F_1$ is contained in $\Exterior(C(P_2))$, since if $F_1$ were contained in $\Interior(C(P_2))$, then \Cref{claim:face-implies-all-paths} would imply that $P_1$ is contained in $\Interior(C(P_2))$.
        Hence $F_1$ is universally exterior.

        Case 2: there does not exist $P_1 \in \support(x^1)$ such that $F_2$ is contained in $\Interior(C(P_1))$.
        Then $F_2$ is contained in $\Exterior(C(P_1))$ for all $P_1 \in \support(x^{H_{F_1}})$, and $F_2$ is universally exterior.
    \end{proof}

    \begin{claim}
        If $F_i$ is universally exterior, then the other support graph $G_{\overline i}$ is entirely contained within a single face of $G_i$ (note that this face is possibly bounded or unbounded).
    \end{claim}
    \begin{proof}
        Let $F_i^*$ be the face of $G_i$ that contains $F_{\overline i}$.
        We will show that for all $P \in \support(x^{\overline i})$, the cycle $C(P)$ is contained in $F_i^*$, which proves the claim.

        Let $h$ be the (embedded) demand edge routed by $P$.
        Certainly the demand $h$ is contained in the interior of $F_i^*$ since it is embedded in the interior of $F_{\overline i}$.
        Now, if $C(P)$ ever leaves $F_i^*$, then there exists a node $v$ that is incident with four distinct edges that alternate in clockwise order between $C(P)$ and $\support(x^i)$.
        Say listed in clockwise order these are $e_1, e_2, e_3, e_4$, where $e_1, e_3 \in C(P)$ and $e_2, e_4 \in \support(x^i)$.
        However, since $C(P)$ is a simple cycle, then at least one of $e_2, e_4$ must enter $\InteriorOpen(C(P))$, which contradicts that $F_i$ is universally exterior.
    \end{proof}

    Now, suppose without loss of generality that $F_1$ is universally exterior, and thus $G_2$ is contained in a single face $F_1^*$ of $G_1$.
    If $F_2$ is also universally exterior, then we are done.
    So suppose not.
    Then there exists $P_1 \in \support(x^1)$ such that $F_2$ is contained in $\Interior(C(P_1))$, and hence by \Cref{claim:face-implies-all-paths} we have that $G_2$ is contained in $\Interior(C(P_1))$.
    Therefore the face $F_1^*$ of $G_1$ that contains $G_2$ must be a bounded face.
    It then follows that $G_1$ must be contained within the unbounded face of $G_2$.

    Thus we have proven that for each $i$, $G_i$ is contained in a single face of $G_{\overline i}$.
    For each $i$, let $S_i$ be the star graph added to $G_i$ to obtain $G[x, F_i]$ (so that $G[x, F_i] = G_i + S_i$.

    Now fix an $i \in \set{1, 2}$.
    Let $F^*_{i}$ be the face of $G_{i}$ that contains $G_{\overline i}$.
    First, observe that $G[x, F_{\overline i}] = G_{\overline i} + S_{\overline i}$ is still contained in that same face of $G_i$, since the star $S_{\overline i}$ is added into the open interior of $F_{\overline i}$ and its embedding is completely disjoint from $G_{\overline i}$.
    
    So $G[x, F_{\overline i}]$ is contained in a single face $F^*_i$ of $G_i$.
    The only way that $G[x, F_{\overline i}]$ is not contained in a single face of $G_i + S_i = G[x, F_i]$ is if adding the star $S_i$ divides $F^*_i$ into multiple faces.
    However, for this to happen we would need that $F^*_i = F_i$ (recall that $F_i$ is a face of $G_{i}$, since we assumed that $E(F_i) \subseteq \support(x^{F_{i}})$).
    However, it cannot be the case that $F^*_i = F_i$, since the leaves added to $F_{\overline i}$ are not contained in $F_i$ but are in $G_{\overline i}$.
    
    Thus $G[x, F_{\overline i}]$ is contained in a single face of $G[x, F_i]$.
    By a symmetric argument we have that $G[x, F_i]$ is contained in a single face of $G[x, F_{\overline i}]$
    This concludes the proof.

\end{proof}

\section{Randomized Rounding for Maximization}
\label{sec:polytime-random-rounding}

In this section, we prove \Cref{lemma:randomized-rounding-log}.

\begin{proof}
    Let $\nu, \gamma \in (0, 1)$ be constants, sufficiently small so that $\frac{1 - \gamma} {1 + \nu} \ge 1 - \epsilon$ (for example, $\nu = \gamma = \epsilon^2$ suffices).

    Let $\mathcal P$ be the set of support paths of $f$.
    We first proceed with some assumptions about the input data.
    By scaling $w$, we may assume $w(h) \le 1$ for each $h \in E(H)$ and $w_{\max} = 1$.
    Finally, observe that we may assume $f(P) < \frac 1 {\log^2 n}$ for each $P \in \mathcal P$.
    Otherwise we can decompose $f = f_1 + f_2$ where $f_2(P)$ is some multiple of $\frac 1 {\log^2 n}$ and $f_1(P) < \frac 1 {\log^2 n}$ (for each $P$, let $k_P$ be the maximum positive integer such that $f(P) - \frac{k_P}{ \log^2 n} \ge 0$ and define $f_2(P) = \frac{k_P} {\log^2 n})$.
    Then we can apply the lemma to $f_1$ to obtain $\hat f_1$, and $\hat f = \hat f_1 + f_2$ satisfies the desired properties for the original flow $f$.
    
    Proceeding under the above assumptions, define an integral flow $\bar f$ by independently assigning  each $\hat f(P)$ to $1$ with probability $\log^2 n \cdot f(P) \in [0, 1)$ and otherwise assigning it 0.
    We use the following Chernoff bounds.

    \begin{lemma}[Chernoff Bounds]
    \label{lemma:chernoff}
    Let $X_1, \dotsc, X_k$ be mutually independent random variables where $X_i \in [0, 1]$.
    Let $X = \sum_i X_i$ and let $\mu_{\max} \ge \expect{X} \ge \mu_{\min}$.
    The following inequalities hold:
    \begin{enumerate}
        \item For any $\delta \ge 1$ we have
            \begin{equation*}
            \prob{X \ge (1 + \delta) \mu_{\max}} \le \exp(-\delta \mu_{\max} / 3).
            \end{equation*}
        \item For any $\delta \le 1$ we have
            \begin{equation*}
                \prob{X \le (1 - \delta) \mu_{\min}} \le \exp(-\delta^2 \mu_{\min} / 2).
            \end{equation*}
    \end{enumerate}
    \end{lemma}

    How much does $\bar f$ congest an edge $e$?
    We have that $\expect{\sum_{P \ni e} \bar f(P)} = \log^2 n \cdot \sum_{P \ni v} f(P) \le \log^2 n$ since $f$ is feasible.
    Let $\mathcal B_e$ be the (bad) event that $\sum_{P \ni e} f(P) \ge (1 + \nu) \log^2 n$.
    Setting $\mu_{\max} = \log^2 n$ and $\delta = \nu$ in \Cref{lemma:chernoff}, we obtain:
    \begin{equation*}
        \prob{\mathcal B_v} \le \exp(-\nu\log^2 n) = n^{-\nu \log n}.
    \end{equation*}

    What value does $\bar f$ achieve?
    We have that $\expect{w(f)} = \sum_h \sum_{P \in \mathcal P_h} w(P) \expect{\bar f(P)} = \log^2 n \cdot \sum_h \sum_{P \in \mathcal P_h} w(h)f(P) = \log^2 n \cdot w(f)$.
    Let $\mathcal B_w$ be the (bad) event that $w(f') \le (1 - \gamma) \log^2 n \cdot w(f)$.
    Setting $\mu_{\min} = \log^2 n \cdot w(f)$ and $\delta = \epsilon$ in \Cref{lemma:chernoff}, we obtain:
    \begin{equation*}
        \prob{w(f') \le (1-\gamma) 
        \log^2 n \cdot w(f)}
        \le \exp(-\gamma^2 \log^2 n \cdot w(f) / 2)
        = n^{-\gamma^2 w(f) / 2}
        \le n^{-\gamma^2/2}
    \end{equation*}
    where the last inequality follows from $w(f) \ge w_{\max} = 1$.

    Thus by the union bound, the probability that one of the events $(\mathcal B_v : v \in E(G))$ or $\mathcal B_w$ occurs is at most $\frac {n^2} {n^{\nu log n}} + \frac 1 {n^{\gamma^2/2}}$, which converges to 0.
    Thus with high probability, we obtain an outcome $f'$ that avoids all of these bad events.

    Now let $\hat f = \frac 1 {(1+\nu)\log^2 n} f'$.
    Since $f'$ avoids $\mathcal B_e$ for each edge $e$, $\hat f$ is feasible.
    Additionally, since $f'$ avoids $\mathcal B_w$, we have that $w(\hat f) = \frac 1 {(1+\nu)} w(f') \ge \frac{(1-\gamma)} {(1+\nu)} w(f) \ge (1+\epsilon)w(f)$.

    This procedure can be de-randomized to obtain a deterministic algorithm using pessimistic estimators \cite{raghavan_probabilistic_1988}.
\end{proof}

\section{NP-Completeness of Fractional Uncrossed Multiflow}
\label{sec:npc}

In this section we prove the NP-completeness of deciding the existence of a fractional uncrossed flow satisfying all demands.  The proof is reminiscent of a reduction by Middendorf and Pfeiffer~\cite{middendorf_complexity_1993}, for the node-disjoint paths problem when $G+H$ is planar. We remark that this shows NP-completeness even for leaf instances of multiflows.
We refer to a fractional uncrossed flow which satisfies all demands as a {\em solution}.

In the first subsection, we show the problem is in NP (which is not necessarily obvious, given that uncrossed fractional multiflow is not a linear program).
In the next subsection, we introduce the key gadgets which are used in the proof of NP-hardness, and in the subsection after we present the main proof of NP-hardness.


\subsection{Polynomial-Sized Solutions}
\label{sec:poly-size}

Adding uncrossed-ness to the constraints of a multiflow problem means it is no longer a linear program, and it is not necessarily obvious that we should expect uncrossed solutions that are polynomial in the size of the input (i.e. in terms of $|V|, |E(G)|, |E(H)|$, and the bit-size of $u, d$).

We give a brief argument for why polynomial-sized solutions exist.
Let $f$ be a feasible uncrossed (fractional) multiflow for a maximization or congestion instance.
Consider a single commodity $h = (s, t) \in E(H)$, and consider the single-commodity uncrossed multiflow $f^h$, obtained by restricting $f$ to $h$.
Let $P_1, \dotsc, P_k$ be the clockwise ordering of the support paths about $s$.
Each sector (i.e. the region bounded between consecutive paths in this ordering) contains some face $F$ of $G$ that is not contained in any other sector.
Hence there are at most $O(|V|)$ sectors, and thus $O(|V|)$ paths.
Thus over all commodities the set of support paths $\mathcal P^*$ is at most $O(|V||E(H)|)$.
To obtain values of $f(P)$ that have polynomial bit-size, we can set up restricted versions of the maximization/congestion linear programs that only have variables for $P \in \mathcal P^*$.
By standard linear programming theory there exists a solution for these problems that has polynomial bit-size, and any solution for these LPs is uncrossed as they draw support paths only from $\mathcal P^*$.

A consequence of this argument is that the problem of determine whether $(G, H, u, d)$ has a feasible (fractional) solution that is uncrossed is in NP, as promised by \Cref{thm:npc}.
Given a polynomial-sized candidate solution, it is straightforward to check that each pair of support paths is uncrossed and that the solution is feasible in polynomial time.

\subsection{Gadgets}

Many of the ideas build on the structure of solutions for the instance  in Figure~\ref{fig:double-diamond} which we call Double Diamond.
\begin{lemma}[Double Diamond Instance]\label{lemma:double-diamond}
    Let $G,H$ be the instance of fractional uncrossed flow from \Cref{fig:double-diamond}. This instance has only two possible path-solutions, and both are integral.
\end{lemma}

\begin{figure}[htb]
    \centering
    \begin{tikzpicture}[x=0.75cm,y=0.75cm]
        \foreach \u/\x/\y in {
          s1/1/2,s2/3/2,t1/7/2,t2/9/2
        } {
            \node[stterminal] (\u) at (\x,\y) {};
        }
        \foreach \u/\x/\y in {
          a0/2/0,b0/8/0,a1/2/4,b1/8/4,c/0/2,d/4/2,e/6/2,f/10/2
        } {
            \node[stvertex] (\u) at (\x,\y) {};
        }
        \foreach \u/\v in {
          a0/b0,a0/c,a0/d,b0/e,b0/f,a1/b1,a1/c,a1/d,b1/e,b1/f,
          c/s1,s2/d,e/t1,t2/f}
        {
            \draw (\u) -- (\v);
        }
        \draw[dotted] (s1) to[in=160,out=20] (t1);        
        \draw[dotted] (s2) to[in=200,out=340] (t2);
        \foreach \u/\n in {s1/s_1,t1/t_1,b0/l_2,b1/h_2,d/u_2,f/v_2} {
          \draw (\u) node[right] {$\n$};
        }
        \foreach \u/\n in {s2/s_2,t2/t_2,a0/l_1,a1/h_1,c/u_1,e/v_1} {
          \draw (\u) node[left] {$\n$};
        }
        \draw[nearly transparent, blue,line join=round,line width = 5pt] 
          (s1) -- (c.center) -- (a1.center) -- (b1.center) -- (e.center) -- (t1);
        \draw[nearly transparent, red,line join=round,line width = 5pt] 
          (s2) -- (d.center) -- (a0.center) -- (b0.center) -- (f.center) -- (t2);
    \end{tikzpicture}
    \caption{{\em Double Diamond Instance}. An instance for the fractional uncrossed flow problem that has only integral solutions (one of these solutions is given by the two colored paths with value 1).}
    \label{fig:double-diamond}
\end{figure}

\begin{proof} 
    Suppose there is a support path $P$ in a fractional solution starting with $s_1u_1,u_1h_1,h_1u_2,u_2l_1$. Then, because the support path are uncrossed, there is no support path starting by $s_2u_2,u_2h_1$, as such a flow path $Q$ would cross $P$ at either the subpath $Z = u_2, h_1$ or $Z = u_2, h_1, u_1$. Then all support paths for the demand $s_2t_2$ uses $u_2l_1$, which, considering also $P$, implies that the total value of flow paths on that edge is more than $1$, contradiction.
    
    By similar arguments, one can show that the support paths use only one of the two highlighted paths from \Cref{fig:double-diamond} $P_1$ (in blue) and $P_2$ (in red) and their symmetric paths by a flip along an horizontal axis $Q_1$ and $Q_2$. To conclude, notice that $P_1$ crosses $Q_2$ and $Q_1$ crosses $P_2$, implying that the only two solutions are $P_1 + P_2$ and $Q_1 + Q_2$.      
\end{proof}

We next discuss the Terminal Gadget, which is a modification of the Double Diamond Instance.

\begin{lemma}[Terminal Gadget, Figure~\ref{fig:double-terminal-gadget}]
\label{lemma:double-terminal-gadget}
    Let $G,H$ be an instance of fractional uncrossed flow,
    with at least two demands $s_1t_1$, $s_2t_2$ as described in \Cref{fig:double-terminal-gadget} (modified from \Cref{fig:double-diamond}).
    Let $G_1$ be the subgraph induced by the nodes outside the circle.
    Then, for any fractional solution $z$ to $G,H$,  
    \begin{enumerate}[label = (\roman*)]
    \item\label{item:gadget1} $s_1t_1$ is routed integrally inside $G_1$;
    \item\label{item:gadget2} either all the support paths for $s_2t_2$ use $c_1x$ or they all use $d_1y$;
    \item\label{item:gadget3} no support path for any other demand uses an edge of $G_1$ (aside from $s_1t_1$, $s_2t_2$).
    \end{enumerate}
\end{lemma}

\begin{figure}
    \centering
    \begin{tikzpicture}[x=0.55cm,y=0.55cm]
        \foreach \u/\x/\y in {
          s1/1/2,s2/3/2,t1/7/2
        } {
            \node[stterminal] (\u) at (\x,\y) {};
        }
        \foreach \u/\x/\y in {
          a0/2/0,b0/8/0,a1/2/4,b1/8/4,c/0/2,d/4/2,e/6/2
        } {
            \node[stvertex] (\u) at (\x,\y) {};
        }
        \draw (12,2) circle[radius=2.5];
        \node[stvertex] (f1) at (10.5,4) {};
        \node[stvertex] (f2) at (10.5,0) {};
        \node[stterminal] (t2) at (13,2) {};
        \foreach \u/\v in {
          a0/b0,a0/c,a0/d,b0/e,b0/f2,a1/b1,a1/c,a1/d,b1/e,b1/f1,
          c/s1,s2/d,e/t1}
        {
            \draw (\u) -- (\v);
        }
        \draw[dotted] (s1) to[in=160,out=20] (t1);        
        \draw[dotted] (s2) to[in=200,out=340] (t2);
        \foreach \u/\n in {s1/s_1,t1/t_1,d/u_2,t2/t_2} {
          \draw (\u) node[right] {$\n$};
        }
        \foreach \u/\n in {b0/d_1,f2/y,a0/b_1} {
          \draw (\u) node[below] {$\n$};
        }
        \foreach \u/\n in {a1/a_1,b1/c_1,f1/x} {
          \draw (\u) node[above] {$\n$};
        }
        \foreach \u/\n in {s2/s_2,c/u_1,e/v_1} {
          \draw (\u) node[left] {$\n$};
        }
        \draw[nearly transparent, blue,line join=round,line width = 5pt] 
          (s1) -- (c.center) -- (a1.center) -- (b1.center) -- (e.center) -- (t1);
        \draw[nearly transparent, red,line join=round,line width = 5pt] 
          (s2) -- (d.center) -- (a0.center) -- (b0.center) -- (f2.center) -- (t2);

        \node[stterminal] (t) at (18,2) {};
        \node[stvertex] (tx) at (19,0) {};
        \node[stvertex] (ty) at (19,4) {};
        \draw (tx) -- (t) -- (ty);
        \draw[blue] (t.center) circle[radius=7pt];
        \draw (t) node[right=8pt] {$s_2$};
        \foreach \u/\l in {tx/x,ty/y} {
            \draw (\u) node[right] {$\l$};
        }
        
    \end{tikzpicture}
    \caption{Terminal Gadget. A gadget that forces the demand $s_2t_2$ to be routed either entirely through $x$ or through $y$.
    \resolved{\jp{todo: Make $x, y$ outside the circle. Relabel the $s_i, t_i$'s here.}.}
    The nodes outside the circle form an induced subgraph, and deleting $x$ and $y$ disconnect these nodes from the rest of the graph.
    Furthermore, no demand other than $s_1t_1$ and $s_2t_2$ can have a support path intersecting the edges of the gadget. The colored paths represent a possible routing. On the right, a schematic to represent the gadget on a terminal, using a blue circle to denote that all the flow must use one of the two incident arcs.}
    \label{fig:double-terminal-gadget}
\end{figure}

\begin{proof}
  By considering the cut $\{c_1x,d_1y\}$, which is crossed by the demand $s_2t_2$, at most half a unit of $s_1t_1$-flow can be routed on these edges.
  This implies that there exists a support $s_1t_1$- path contained in $G_1$.
  There are 4 possible such paths, two of length $5$ and two of length $7$.
  Let $P = s_1u_1a_1u_2b_1d_1,v_1t_1$ and suppose that $z(P) > 0$.
  Since less than 1 unit of capacity is now available on $u_2 b_1$ for $s_2 t_2$, there must be an support $s_2t_2$-path $Q$ starting with $s_2u_2a_1$.
  Hence $Q$ must continue to $u_1$
  and $b_1$,
  but this crosses $P$.
  Therefore $z(P) = 0$, and there must be a support $s_1t_1$-path $P_1$ of length $5$.
  We may assume $P_1 = s_1u_1a_1c_1v_1t_1$.

  Suppose that there is a support $s_2t_2$-flow path $P_2$, containing $c_1x$.
  Then as $P_1$ and $P_2$ do not cross, $P_2$ must contain a subpath $b_1u_1a_1c_1x$.
  Because $\{a_1c_1,b_1d_1\}$ is a tight cut and $a_1c_1 \in P_2$, $P_2$ does not contain $b_1d_1$.
  Hence $P_2$ starts with $s_2u_2b_1u_1\ldots$
  Because $z(P_2) > 0$ and $u_1a_1 \in P_2$, there is a support $s_1t_1$-path $Q_1$ starting with $s_1u_1b_1$.
  As $P_2$ and $Q_1$ do not cross, $Q_1$ continues with $b_1u_2$, but then it must cross $P_2$ at $u_2$.
  Hence no support $s_2t_2$-path contains $c_1x$, they must all contain $d_1y$.
  Furthermore, this shows that no support $s_1t_1$-path intersects the cut $\{c_1x,d_1y\}$, since $s_2t_2$ uses all the capacity of $d_1y$, and hence the support path would have to use $c_1 x$ twice. 
  So the only possible support $s_1t_1$-paths are $P_1$ and its symmetric $P'_1 = s_1u_1b_1d_1v_1t_1$. 

  Finally suppose that $z(P_1) > 0$ and $z(P'_1) > 0$. As $P_1 \cup P'_1$ forms a circuit separating $s_2$ from $t_2$, any active flow $s_2t_2$-path would cross either $P_1$ or $P'_1$. Thus, \labelcref{item:gadget1} holds. Up to symmetry, we may assume $z(P_1) = 1$, $z(P'_1) = 0$. From there, any active $s_2t_2$-flow path must start with $s_2u_2b_1d_1y$, that is \labelcref{item:gadget2} holds, and \labelcref{item:gadget3} follows.
\end{proof}

\begin{figure}
    \centering
    \begin{tikzpicture}[x=1cm,y=1cm]
        \node[stvertex] (a) at (0,2) {}; \draw (a) node[left] {$y_0$};
        \node[stvertex] (a2) at (7,2) {}; \draw (a2) node[right] {$x_4$};
        \node[stvertex] (b) at (0,0) {}; \draw (b) node[left] {$z_0$};
        \node[stvertex] (b2) at (7,0) {}; \draw (b2) node[right] {$w_4$};
        \foreach \i in {1,2,3} {
          \node[stvertex] (x\i) at ($(-1,2) + 2*(\i,0)$) {}; \draw (x\i) node[above] {$x_\i$};
          \node[stvertex] (y\i) at ($(0,2) + 2*(\i,0)$) {}; \draw (y\i) node[above] {$y_\i$};
          \node[stvertex] (w\i) at ($(-1,0) + 2*(\i,0)$) {}; \draw (w\i) node[below] {$w_\i$};
          \node[stvertex] (z\i) at ($(0,0) + 2*(\i,0)$) {}; \draw (z\i) node[below] {$z_\i$};
          \node[stterminal] (s\i) at ($(-1,1) + 2*(\i,0)$) {}; \draw (s\i) node[above left=4pt] {$s_\i$};
          \node[stterminal] (t\i) at ($(0,1) + 2*(\i,0)$) {}; \draw (t\i) node[below right=4pt] {$t_\i$};
          \draw (s\i) -- (x\i) -- (y\i) -- (t\i);
          \draw (s\i) -- (w\i) -- (z\i) -- (t\i);
          \draw[dotted] (s\i) -- (t\i);
          \draw[blue] (s\i) circle[radius=7pt];
          \draw[blue] (t\i) circle[radius=7pt];
i        }             
        \foreach \u/\v in {a/x1,y1/x2,y2/x3,y3/a2,b/w1,z1/w2,z2/w3,z3/b2} {
          \draw (\u) -- (\v);
          }
        \begin{scope}[xshift=9cm]
            \node[stvertex] (a) at (0,2) {}; \draw (a) node[left] {$y_0$};
            \node[stvertex] (a2) at (2,2) {}; \draw (a2) node[right] {$x_4$};
            \node[stvertex] (b) at (0,0) {}; \draw (b) node[left] {$z_0$};
            \node[stvertex] (b2) at (2,0) {}; \draw (b2) node[right] {$w_4$};
            \draw (a) -- (a2);
            \draw (b) -- (b2);
            \draw[blue] (0.9,-0.2) -- (0.9,2.2);
            \draw[blue] (1.1,-0.2) -- (1.1,2.2);
        \end{scope}
    \end{tikzpicture}
    \caption{Coupled Arcs Gadget. \resolved{\jp{todo. Change node labels to match previous gadget.}}A gadget that forbids the use of one arc among two, for any flow path from terminal outside the gadget. Recall that the circle refers to the schematic of the Terminal Gadget from Figure~\ref{fig:double-terminal-gadget}. On the right, a schematic representation of the Coupled Arcs Gadget.}
    \label{fig:coupled-arcs}
\end{figure}

\begin{lemma}[Coupled Arcs Gadget, Fig~\ref{fig:coupled-arcs}]
\label{lemma:coupled-arcs}
    Let $G,H$ be an uncrossable flow instance, containing the gadget of \Cref{fig:coupled-arcs} as an induced subgraph.
    Let $P_{xy}$ and $P_{wz}$ be the top and bottom paths $y_0x_1y_1x_2y_2x_3y_3x_4$ and $z_0w_1z_1w_2z_2w_3z_3w_4$ respectively.
    Then in any solution $x$, there is $Q \in \{P_{xy}, P_{wz}\}$ with the following property: for every external demand $st \in E(H) \setminus \{s_1t_1,s_2t_2,s_3t_3\}$ and for any support $st$-path $P$, the intersection of $E(P)$ with the gadget is either empty or exactly $Q$.
    Thus, this gadget behaves as though we have only two edges $y_0x_4$, $z_0w_4$, such that at most one of them can be used by external demands.
\end{lemma}

\begin{proof}
  We denote by $P_i$ and $Q_i$ the paths $s_ix_iy_it_i$ and $s_iw_iz_it_i$ respectively.
  We denote by $f_i(e)$ the flow value on edge $e$ induced by the flow paths for $s_it_i$.
  
  From \Cref{lemma:double-terminal-gadget}, for each of the demands $\{s_i,t_i\}$ ($i \in \{1,2,3\}$), either $f_i(s_ix_i) = 1$ ($s_i$ is routed up) or $f_i(s_iw_i) = 1$ ($s_i$ is routed down).
  Similarly each $t_i$ is routed down or up.
  We say that $s_it_i$ is \emph{coherent} if $s_i$ and $t_i$ are routed in the same direction.
  Observe that if a support path for an external demand enters the gadget, it cannot use any edges incident with $s_i$ or $t_i$ for $i = 1, 2, 3$.
  Hence our goal is to that at least one of $P_{xy}$ or $Q_{xy}$ is not part of any support path for any external demand.

  For any $r \in \{1, 2, 3\}$, if $s_r t_r$ is coherent (routing up, say), then any support $s_r t_r$-path is either exactly $P_r$ or it starts with $s_rx_ry_{r-1}\ldots{}y_0$ and ends with $x_4y_3\ldots{}y_rt_r$.
  In the latter case we say the path \emph{escapes}.
  If any coherent $s_r t_r$ does not have any escaping support paths (i.e., it routes entirely on $P_r$), then it is impossible for any external demand to route using $P_{xy}$ if $s_r t_r$ routes up, and impossible to use $P_{wz}$ if $s_r t_r$ routes down.
  In these cases, the lemma holds with $Q = P_{wz}$ in the former case and $Q = P_{xy}$ in the latter.
  So we may assume each coherent $s_r t_r$ has an escaping support path.

  For any incoherent pair, all the support paths leave and re-enter the gadget, and therefore use up 2 units of capacity on the cut $\{y_0x_1,z_0w_1,y_3x_4,z_3w_4\}$.
  Hence if there are 2 incoherent pairs, then no external demand may cross through the gadget, in which case the lemma holds.
  We may thus assume that there are $i < j$ such that $s_it_i$ and $s_jt_j$ are coherent. Let $k \in \{1,2,3\} \setminus \{i,j\}$.

  {\bf Case 1:}
  Suppose both $s_it_i$ and $s_j t_j$ are routed in the same direction, say up.
  Since both $s_i t_i$ and $s_j t_j$ have a support path that escapes, these paths cross, a contradiction.
  

  {\bf Case 2:}
  The remaining case is when both $s_it_i$ and $s_jt_j$ are routed in distinct directions, and $s_kt_k$ is not coherent (otherwise it would be in the same direction as $s_it_i$ or $s_jt_j$ and the previous case applies).
  Without loss of generality, we may assume that $s_it_i$ is routed up and $i < k$.
  To avoid crossing the escaping path for $s_i t_i$ at $y_i$, every support path for $s_k t_k$ must use $y_3 x_4$.
  However, this does not leave enough capacity for the escaping path for $s_i t_i$ that also uses this edge, a contradiction.
\end{proof}

By routing each demand up, or each demand down, one of the two paths $P_{xy}$, $P_{wz}$ is free to be used by support paths from external demands.




Thus, the gadget from \Cref{fig:coupled-arcs} can be used to replace two edges $e'$ and $e''$, with the effect of forbidding the use of both edges in the fractional uncrossed flow: it adds the constraint that either $x(e') = 0$ or $x(e'') = 0$.

\begin{lemma}
  [Clause Gadget, Figure~\ref{fig:clause-gadget}]
  \label{lemma:clause-gadget}
  The clause gadget (Figure~\ref{fig:clause-gadget}) has a solution even if we remove two of the three edges $e', e'', e'''$. It has no solution if we remove all three edges $e', e'', e'''$.
\end{lemma}

\begin{figure}
    \centering
    \begin{tikzpicture}[x=0.5cm,y=0.5cm]
        \foreach \n/\x/\y in {u1/-2/4,u2/4/4,u3/1/-2,v1/0/4,v2/6/4,v3/3/-2,x1/1.5/2,x2/2.5/2} {
          \node[stvertex] (\n) at (\x,\y) {};
        }
        \node[stterminal] (s) at (0,0) {};
        \node[stterminal] (t) at (4,0) {};
        \foreach \u/\v in {s/u1,v1/x1,x2/u2,v2/t,t/v3,u3/s,s/x1,x2/t,x1/x2} {
           \draw (\u) -- (\v);
        }
        \foreach \u/\v/\n in {u1/v1/e',u2/v2/e'',u3/v3/e'''} {
           \draw (\u) -- (\v) node[midway,above] {$\n$};
        }
        \draw[dotted] (s) -- (t);
        \draw[blue] (0.5,0.9) -- (3.5,0.9);
        \draw[blue] (0.7,1.1) -- (3.3,1.1);
        \draw (s) node[left] {$s_C$};
        \draw (t) node[right] {$t_C$};
    \end{tikzpicture}
    \caption{The Clause Gadget for a clause $C$. The two edges marked with a double blue line should be replaced by the gadget from \Cref{fig:coupled-arcs}, allowing only one of these two edges to be used by the flow.}
    \label{fig:clause-gadget}
\end{figure}

\begin{proof}
    It easily follows from \Cref{lemma:coupled-arcs} by a simple case analysis, as only one of the two edges crossed by the double blue line can be used in routing the demand edge.
\end{proof}

\subsection{Hardness of Finding Uncrossed Fractional Flows}

\begin{theorem}
    \label{thm:npc-final}
    The problem {\scshape Fractional Uncrossed Flow} is NP-complete.
\end{theorem}

\begin{proof}
  We reduce from {\scshape Planar 3-Satisfiability}, in which the given 3-Sat instance has a planar variable-clause incidence graph.
  We may assume that each variable appears at most three times, and furthermore that it is not the same polarity in all of its incident clauses.
  From the variable-clause incidence graph $G$, we construct an uncrossed multiflow instance $(G', H')$ as follows.

  For each variable $x$, we replace the node for $x$ in $G$ with a 3-cycle, where the edges are labelled with the clauses incident with $x$ as $e_{x, C}$.
  The clockwise order of the labels corresponds to the clockwise ordering of the clauses about $x$ in $G$.
  Now, we place a demand $s_x t_x$ between two nodes of the 3-cycle, so that the two $s_x t_x$ paths on the cycle correspond to the clauses where $x$ appears positively and negatively.
  Call these two paths $P_x$ and $N_x$ respectively.

  Now for each clause $C$, we replace the node for $C$ in $G$ with a clause gadget from \Cref{fig:clause-gadget}, where the edges $e'$, $e''$ and $e'''$ of the figure are labelled by the three variables occurring in $C$ as $e_{C, x}$.
  The clockwise order of the labels corresponds to the clockwise ordering of the variables about $C$ in $G$.

  Now, consider each incident variable-clause pair $(x, C)$.
  We replace the edges $e_{x, C}$ (from the 3-cycle of $x$) and $e_{C, x}$ (from the clause gadget of $C$) with a copy of the Coupled Arcs Gadget (\Cref{fig:coupled-arcs}).
  By installing the Coupled Arcs Gadget in close proximity to the edges of the original graph $G$ (recalling that the edge-labellings also mirror the original graph), we obtain a planar embedding of $G'$.
  We will abuse notation and still use $e_{x, C}$ and $e_{C, x}$ for the subdivided paths in the Coupled Arcs Gadget.
  By the properties of the Coupled Arcs Gadget, this has the effect that flow for demands outside of the Coupled Arcs Gadget can only use one of the edges $e_{x, C}$ or $e_{C, x}$.
  Additionally, support paths cannot ``hop between'' the two coupled edges, so the only two possible support paths for $s_{x} t_{x}$ are $P_x$ and $N_x$.
  Similarly, any support path for the demand $s_C t_C$ of a clause gadget stays entirely inside the clause gadget.

%

  If the satisfiability instance is satisfied by a truth assignment $\phi$, then we can find an uncrossed (integral) routing.
  First, for each variable $x$, if $\phi(x) = \mathit{true}$, we set $f(N_x) = 1$, otherwise set $f(P_x) = 1$.
  This clearly satisfies all of the variable demands $s_x t_x$.
  We now extend this routing to satisfy the clause demands and the internal demands of the Coupled Arcs Gadgets.
  
  Now, consider a clause $C$ and its gadget.
  Some literal for a variable $x$ in $C$ is made true by the assignment.
  Without loss of generality, assume $x$ occurs positively in $C$ and $\phi(x) = \mathit{True}$.
  In our partial routing, $e_{x, C}$ is in $P_x$, and is not used by the routing for $s_x t_x$.
  Therefore we can route the demands from the Coupled Arcs Gadget using $e_{x, C}$.
  This leaves $e_{C, x}$ available to route the demand $s_C t_C$ and still satisfy the Coupled Arcs Gadget strictly inside the clause gadget (as in Lemma \ref{lemma:clause-gadget}).
  The other two Coupled Arcs Gadgets incident with the clause gadget for $C$ can be routed inside the clause gadget using the other variable-labelled edges.
  This satisfies all of the demands that appear in clause gadgets and Coupled Arcs Gadgets.
  In this way, we find an integral uncrossed routing.

  Conversely, suppose there is a fractional uncrossed flow $f$.
  For each variable $x$, either $f(P_x) > 0$ or $f(N_x) = 1$.
  Set $\phi(x) = \mathit{false}$ in the former case, and $\mathit{true}$ otherwise.
  We claim that $\phi$ satisfies all the clauses.
  Consider a clause $C$ with variables $x_1, x_2, x_3$ (with some polarity).
  One of the edges $e_{C, x_1}, e_{C, x_2}, e_{C, x_3}$ must be used by a support path for the demand $s_C t_C$.
  Say $e_{C, x_i}$ is used.
  By the properties of the Coupled Arcs Gadget, this implies that the corresponding edge $e_{x_i, C}$ in the 3-cycle for $x_i$ is saturated by the internal demands of the Coupled Arcs Gadget and is thus not available for routing $s_x t_x$.
  If $x_i$ occurs positively in $C$, this means that $f(P_{x_i}) = 0$, and so $\phi(x_i) = \mathit{true}$, satisfying the clause.
  On the hand, if $x_i$ occurs negatively in $C$, this means that $f(N_{x_i}) = 0$ and $f(P_{x_i}) = 1$, and so $\phi(x_i) = \mathit{false}$, satisfying the clause.
  As this holds for any clause, $\phi$ is a satisfying assignment, concluding the proof of reduction.
\end{proof}

\subsection{Exists-Embedding Uncrossed Multiflow}
\begin{theorem}
    It is NP-hard to determine whether there is an embedding of $G$ in which $(G, H)$ has a fractional uncrossed multiflow.
\end{theorem}
\begin{proof}
    We reduce from the fixed-embedding version of fractional uncrossed multiflow.
    Start with a fixed embedding of $G$.
    Consider the following operation: add an edge $e = (u_1, u_2)$ between two nodes $u_1, u_2$ on the same face, then subdivide it into a path $P = u_1, v_1, v_2, u_2$, and add demand edges $u_1 v_1, u_2 v_2$.
    The new supply and demand edges have unit weight.

    \begin{claim}
        $(G, H)$ has an uncrossed solution (in the original embedding) if and only if the new instance has an uncrossed solution.
    \end{claim}
    \begin{subproof}
        It is straightforward to verify that if $(G, H)$ has an uncrossed solution, then so does the new instance.
        We show the converse, that if the new instance has an uncrossed solution, so too does $(G, H)$.
        Consider an arbitrary uncrossed solution $f$ for the new instance.
        It suffices to show that there is an uncrossed solution where the demand $u_i v_i$ is routed entirely on the supply edge $u_i v_i$ for $i = 1, 2$.

        Observe that, in $f$, all flow for $u_1 v_1$ that is not routed through $u_1 v_1$ must be routed through $u_2 v_2$.
        So, $f^{u_1 v_1}(u_1 v_1) + f^{u_1 v_1}(u_2 v_2) = 1$.
        Similarly, $f^{u_2 v_2}(u_2 v_2) + f^{u_2 v_2}(u_1 v_1) = 1$.
        So, between the two new demands, all capacity on the new supply edges $u_1 v_1, u_2 v_2$ is used.
        Therefore, the demands of the original demand graph $H$ are routed (uncrossed) inside the original supply graph $G$.
        So $(G, H)$ has an uncrossed solution.


    \end{subproof}

    We then repeat this operation until the supply graph is a (subdivision of) a 3-connected graph, which has a unique embedding.
    The claim implies that the original instance $(G, H)$ has an uncrossed solution (in its given embedding) if and only if the new instance has an uncrossed solution in its embedding, which by the uniqueness of the embedding is equivalent to the new instance having an uncrossed solution in some embedding.
\end{proof}

\section{Colouring Uncrossed String Graphs}
\label{sec:colouring-bound}

In this section, we prove the following result.

\begin{theorem}
\label{thm:notquitechibounded}
Let $U$ be an uncrossed string graph arising from a realization with at most $k$ paths on any edge.
Then $\chi(U) = O( k^2 \log |E(U)|)$.
\end{theorem}

The following is the key lemma.

\begin{lemma}
\label{lem:unionofplanar}
    Let $U$ be an uncrossed string graph with a load $k$ realization $(G, \mathcal P)$.
    Then $U$ is the union of $O(k^2 \log |E(U)|)$ planar graphs.
\end{lemma}

\begin{proof}
    We generate a random subgraph $H = (S, F)$ of $U$ as follows.
    First, we select random subset $S \subseteq \mathcal P$ by drawing each $P \in \mathcal P$ independently and uniformly at random with probability $q = \frac 1 k$.
    Now, for each $e \in E(G)$ such that $|S \cap \mathcal P_e| = 2$ exactly, we add to $F$ an edge between the two members of $S \cap \mathcal P_e$.
    Note that $H$ is planar.


    Let $P P' \in E(U)$, and consider the event that $P P' \in E(H)$.
    This event happens if there exists some $e$ with $P, P' \in \mathcal P_e$ where $P \in S$, $P' \in S$, and no other path in $\mathcal P_e$ is in $S$.
    Clearly there is at least one $e$ such that $P, P' \in \mathcal P_e$, and we lower bound the probability that $PP'$ is included due to this particular edge.
    Then,
    \begin{align*}
        \prob{P P' \in E(H)}
        \ge q \cdot q \cdot (1-q)^{k-2}
        \ge \frac 1 {k^2} \left(1 - \frac 1 k\right)^{k - 2} 
        \ge \frac 1 {e k^2},
    \end{align*}
    where the last inequality follows from the fact that $(1-1/x)^{x-2} \ge 1/e$ for all $x > 1$.

    We repeat this process $C = 2e k^2 \log |E(U)| $ times, independently, obtaining random subgraphs $H_1, \dotsc, H_C$ of $U$.
    For $PP' \in E(U)$, let $A_{PP'}$ be the event that for all $i \in [C]$, $PP' \notin E(H_i)$.
    Then,
    \begin{align*}
        \prob{A_{PP'}}
        \le \left(1 - \frac 1 {e k^2}\right)^{2 e k^2 \log |E(U)|}
        \le \left(\frac 1 e\right)^{2 \log |E(U)|}
        = |E(U)|^{-2},
    \end{align*}
    where the second inequality follows from the fact that $(1 - 1/x)^x \le 1/e$ for $x > 1$.

    Let $A$ be the event that there exists $PP' \in E(U)$ such that for all $i \in [C]$, $PP' \notin E(H_i)$.
    That is, there exists $PP' \in E(U)$ such that $A_{PP'}$ holds.
    By a basic union bound, $\prob{A} \le |E(U)| \cdot |E(U)|^{-2} < 1$.
    So there is an outcome in which $A$ does not hold.
    In this outcome, $U$ is covered by the graphs $H_1, \dotsc, H_C$.
    Hence $U$ is the sum of $C$ planar graphs.
\end{proof}

From this, \Cref{thm:notquitechibounded} follows easily.

\begin{proof}[Proof of \Cref{thm:notquitechibounded}]
    We show that $U$ is $O(k^2 \log |E(U)|)$-degenerate, from which the result follows.
    As the property assumed in the theorem statement is hereditary (i.e. closed under node deletion), it suffices to show that $U$ contains a node with at most $O(k^2 \log |E(U)|)$ neighbours.

    By \Cref{lem:unionofplanar}, $U$ is the sum of $O(k^2 \log |E(U)|)$ planar graphs.
    Each planar graph contributes at most $3 |V(U)|$ parallel edge classes, so it follows that $E(U)$ contains at most $O(k^2 |V(U)| \log |E(U)|)$ parallel edge classes.
    Hence it contains a node with at most $O(k^2 \log |E(U)|)$ neighbours.
\end{proof}

\end{document}